\documentclass[11pt,pdfa]{article}
\usepackage[in]{fullpage}
\usepackage{palatino}

\usepackage[style=alphabetic,minalphanames=3,maxalphanames=4,maxnames=99,backref=true]{biblatex}

\usepackage[normalem]{ulem}
\usepackage{comment}
\usepackage[shortlabels]{enumitem}
\usepackage{bm}
\usepackage{amsfonts}
\usepackage{xspace}
\usepackage{amsmath}
\usepackage{amssymb}
\usepackage{amsthm}
\usepackage{xcolor}
\usepackage{graphicx}
\usepackage{qtree}
\usepackage{tree-dvips}
\usepackage{float}
\usepackage{hyperref}
\usepackage{physics}
\usepackage{breakurl}
\usepackage{braket}
\usepackage{mathtools}
\usepackage{mathrsfs}
\usepackage{tikz}
\usepackage[linesnumbered,ruled,vlined]{algorithm2e}
\usepackage{qcircuit}
\usepackage[nameinlink,capitalize]{cleveref}

\newcommand{\subversion}[1]{}

  \hypersetup{colorlinks={true},linkcolor={blue},citecolor=magenta}

  \theoremstyle{plain}
  \newtheorem{theorem}{Theorem}[section]
  \newtheorem{lemma}[theorem]{Lemma}
  \newtheorem{corollary}[theorem]{Corollary}
  \newtheorem{definition}[theorem]{Definition}
  \newtheorem{remark}[theorem]{Remark}
  \newtheorem{claim}[theorem]{Claim}
  \newtheorem{importedtheorem}[theorem]{Imported Theorem}
 \newtheorem{construction}{Construction}
 \newtheorem{conjecture}{Conjecture}

  \usepackage[style=alphabetic,minalphanames=3,maxalphanames=4,maxnames=99,backref=true]{biblatex}

\newtheorem{fact}[theorem]{Fact}

\newcommand{\algo}{\mathcal}

\newcommand{\sk}{\mathsf{sk}\xspace}
\newcommand{\pk}{\mathsf{pk}\xspace}

\newcommand{\vk}{\mathsf{vk}\xspace}
\newcommand{\ct}{\mathsf{CT}\xspace}
\renewcommand{\vec}[1]{\mathbf{#1}}

\newcommand{\bit}{\{0,1\}}
\newcommand{\Eval}{\mathsf{Eval}}
\newcommand{\SIS}{\mathsf{SIS}}
\newcommand{\ISIS}{\mathsf{ISIS}}
\newcommand{\LWE}{\mathsf{LWE}}

\DeclareMathAlphabet\mathbfcal{OMS}{cmsy}{b}{n}
\newcommand{\FHE}{\ensuremath{\mathsf{FHE}}\xspace}
\newcommand{\proj}[1]{\ensuremath{|#1\rangle \langle #1|}}

\newcommand{\KeyGen}{\mathsf{KeyGen}}
\newcommand{\Del}{\mathsf{Del}}
\newcommand{\Vrfy}{\mathsf{Vrfy}}
\newcommand{\PPT}{\mathsf{PPT}}

\newcommand{\Invert}{\mathsf{Invert}}
\newcommand{\QPT}{\mathsf{QPT}}
\newcommand{\CPTP}{\mathsf{CPTP}}
\newcommand{\aux}{\mathsf{aux}}
\newcommand{\FT}{\mathsf{FT}}

\newcommand{\Mod}[1]{\ (\mathrm{mod}\ #1)}
\newcommand{\rand}{\raisebox{-1pt}{\ensuremath{\,\xleftarrow{\raisebox{-1pt}{$\scriptscriptstyle\$$}}\,}}}

\newcommand{\N}{\mathbb{N}}

\newcommand{\Z}{\mathbb{Z}}

\newcommand{\negl}{\mathsf{negl}}

\newcommand{\james}[1]{{\color{red} James: #1 }}
\newcommand{\dakshita}[1]{{\color{orange} Dakshita: #1 }}

\title{Publicly-Verifiable Deletion via Target-Collapsing Functions}
\author{James Bartusek\footnote{bartusek.james@gmail.com}\\UC Berkeley\and 
Dakshita Khurana\footnote{dakshita@illinois.edu}\\UIUC \and 
Alexander Poremba\footnote{{aporemba@caltech.edu}}\\Caltech
}
\date{}

\newcommand{\poly}{\mathrm{poly}}

\newcommand{\secp}{\lambda}

\def\cA{{\cal A}}

\def\cD{{\cal D}}

\def\cF{{\cal F}}

\def\cH{{\cal H}}

\def\cM{{\cal M}}

\def\cZ{{\cal Z}}

\def\bbI{{\mathbb I}}

\def\bbN{{\mathbb N}}

\def\poly{{\rm poly}}

\def\negl{{\rm negl}}

\newcommand{\Enc}{\mathsf{Enc}}
\newcommand{\Dec}{\mathsf{Dec}}

\DeclareMathOperator*{\expectation}{\mathbb{E}}
\newcommand{\E}{\expectation}

\newcommand{\Gen}{\mathsf{Gen}}

\newcommand{\Hyb}{\mathsf{Hyb}}

\newenvironment{boxfig}[2]{\begin{figure}[#1]\fbox{
    \begin{minipage}{\linewidth}
    \vspace{0.2em}\makebox[0.025\linewidth]{}    \begin{minipage}{0.95\linewidth}{{#2 }}
    \end{minipage}\vspace{0.2em}\end{minipage}}}{\end{figure}}

\newcommand{\PVD}{\mathsf{PVD}}

\renewcommand{\partial}{\mathsf{partial}}

\newcommand{\TD}{\mathsf{TD}}

\newcommand{\Exp}{\mathsf{Exp}}

\newcommand{\sA}{\mathsf{A}}

\newcommand{\Ver}{\mathsf{Ver}}

\newcommand{\Samp}{\mathsf{Samp}}

\newcommand{\td}{\mathsf{td}}
\newcommand{\Recover}{\mathsf{Recover}}
\newcommand{\nonnegl}{\mathsf{non}\text{-}\mathsf{negl}}

\begin{document}

\maketitle

\begin{abstract}

We build quantum cryptosystems that support publicly-verifiable deletion from standard cryptographic assumptions. We introduce target-collapsing as a weakening of collapsing for hash functions, analogous to how second preimage resistance weakens collision resistance; that is, target-collapsing requires indistinguishability between superpositions and mixtures of preimages of an honestly sampled image. 

We show that target-collapsing hashes enable publicly-verifiable deletion ($\PVD$), proving 
conjectures from [Poremba, ITCS'23] and demonstrating that the Dual-Regev encryption (and corresponding fully homomorphic encryption) schemes support $\PVD$ under the LWE assumption. 
We further build on this framework to obtain a variety of primitives supporting publicly-verifiable deletion from weak cryptographic assumptions, including:
\begin{itemize}
    \item Commitments with $\PVD$ assuming the existence of injective one-way functions, or more generally, {\em almost-regular} one-way functions. Along the way, we demonstrate that (variants of) target-collapsing hashes can be built from almost-regular one-way functions.
    \item Public-key encryption with $\PVD$ assuming trapdoored variants of injective (or almost-regular) one-way functions. We also
    demonstrate that the encryption scheme of [Hhan, Morimae, and Yamakawa, Eurocrypt'23] based on pseudorandom group actions
    has $\PVD$.
    \item $X$ with $\PVD$ for $X \in \{$attribute-based encryption, quantum fully-homomorphic encryption, witness encryption, time-revocable encryption$\}$,
    assuming $X$ and trapdoored variants of injective (or almost-regular) one-way functions.  
\end{itemize}
\end{abstract}

\newpage
\tableofcontents

\newpage
\section{Introduction}
Recent research has explored the exciting possibility of combining quantum information with computational hardness to enable classically infeasible cryptographic tasks. Beginning with proposals such as unforgeable money~\cite{Wiesner83}, this list has recently grown to include the possibility of provably deleting cryptographic information encoded into quantum states~\cite{Unruh2013,Broadbent_2020,hiroka2021quantum,cryptoeprint:2022/969,hiroka2021certified,Poremba22,cryptoeprint:2022/1178,BGGKMRR,AKNYY,cryptoeprint:2023/325}.

In this work, we further investigate the task of provable deletion of information via destructive measurements. 
We focus on building primitives that satisfy {\em publicly-verifiable deletion} ($\PVD$). This deletion property allows any participant in possession of a quantum encoding to publish a publicly-verifiable classical certificate proving that they deleted\footnote{In this work, we focus on \emph{information-theoretic} deletion of computationally hidden secrets, where the guarantee is that after deletion, even an unbounded adversary cannot recover the plaintext that was previously determined by their view \cite{cryptoeprint:2022/1178}.} the underlying plaintext. This is in contrast to the weaker {\em privately-verifiable deletion} property, where deletion can be verified only by parties that hold a secret verification key, and this key must remain hidden from the party holding the ciphertext.
Public verification is more desirable due to its stronger security guarantee: secret verification keys do not need to be stored in hidden locations, and security continues to hold even when the verification key is leaked. 
Furthermore, clients can outsource verification of deletion by publishing the verification key itself.

Our approach to building publicly verifiable deletion departs from templates used in prior works on deletion. While most prior works, building on~\cite{Unruh2013,Broadbent_2020}, rely on the combination of a quantum information-theoretic tool such as Wiesner encodings/BB84 states~\cite{Wiesner83,BB84} and a cryptographic object such as an encryption scheme, our work enables publicly-verifiable deletion by directly using simple cryptographic properties of many-to-one hash functions.
%, that we call {\em target-collapsing.}

%\dakshita{cite BI prominently somewhere in the intro I guess.}
%\dakshita{also discuss the information-theoretic deletion def from bk22 that we follow here.}

%\dakshita{@Alex: currently planning to write in the tech overview, when discussing collapsing hashes, that your paper had a concrete conjecture. Please feel free to suggest changes if you think your work is not credited properly.}

\paragraph{The Template, in a Nutshell.} 
When illustrating our approach to publicly-verifiable deletion, it will help to first consider enabling this for a simple cryptographic primitive: a commitment scheme. 
That is, we consider building a statistically binding non-interactive quantum bit commitment scheme where each commitment is accompanied by a classical, {\em public} verification key $\mathsf{vk}$. A receiver holding the commitment may generate a classical proof that they deleted the committed bit $b$, and this proof can be publicly verified against $\mathsf{vk}$. We would like to guarantee that as long as verification accepts, the receiver has information-theoretically removed $b$ from their view and will be unable to recover it given unbounded resources, despite previously having the bit $b$ determined by their view.

To allow verification to be a public operation, it is natural to imagine the certificate or proof of deletion to be a hard-to-find solution to a public puzzle. For instance, the public verification key could be an image $y$ of a (one-way) function, and the certificate of deletion a valid pre-image $f^{-1}(y)$ of this key. Now, the commitment itself must encode the committed bit $b$ in such a way that the ability to generate $f^{-1}(y)$ given the commitment implies information-theoretic deletion of $b$. This can be enabled by encoding $b$ in the {\em phase} of a state supported on multiple pre-images of $y$. 

Namely, given an appropriate {\em two-to-one} function $f$, 
%\dakshita{switching to encryption because purification is making it look ugly..}
a commitment\footnote{Technically, it is only an appropriate purification of the scheme described here that will satisfy binding; we ignore this detail for the purposes of this overview.} to a bit $b$ can be
\[\mathsf{Com}(b) = \left( y, \ket{0,x_0}_\sA + (-1)^b \ket{1,x_1}_\sA \right) \]
where $(0,x_0), (1,x_1)$ are the two pre-images of (a randomly sampled) image $y$.
%\dakshita{is mixed state notation better?}
%\dakshita{even this is technically not true. any ideas on how to write a simple, true claim about primitives following from one-wayness without trapdoors? or is this okay?}

Given an image $y$ and a state on register $\sA$, a valid certificate of deletion of the underlying bit could be any pre-image of $y$, which for a well-formed commitment will be obtained by measuring the $\sA$ register in the computational basis. It is easy to see that an immediate {\em honest} measurement of the $\sA$ register implies information-theoretic erasure of the phase $b$. 
But a malicious adversary holding the commitment may decide to perform arbitrary operations on this state in an attempt to find a pre-image $y$ without erasing $b$.
\iffalse{
\footnote{
This template is inspired by, and at the same time significantly generalizes, candidate encryption schemes with publicly-verifiable deletion suggested in recent work~\cite{Poremba22}. %\dakshita{we compare later}
%\iffalse{
Specifically,~\cite{Poremba22} suggested constructions for encryption with publicly-verifiable deletion assuming certain hash functions satisfy a conjectured {\em strong collapsing} property for more complex superpositions of preimages.
%{\em strong collapsing} property for {\em Gaussian superpositions of preimages}. 
In addition to proving conjectures from~\cite{Poremba22} under the (quantum) hardness of the Short Integer Solutions (SIS) problem, a key goal in our work is to base $\PVD$ on even weaker cryptographic assumptions.}
}\fi

In this work, we analyze  (minimal) requirements on the cryptographic hardness of $f$ in the template above, so that the ability to computationally find any preimage of $y$ given the commitment necessarily implies {\em information-theoretic} erasure of $b$. A useful starting point, inspired by recent conjectures in~\cite{Poremba22}, is the {\em collapsing} property of hash functions. This property was first introduced in~\cite{10.1007/978-3-662-49896-5_18} as a quantum strengthening of collision-resistance. 
%This is inspired by a recent work~\cite{Poremba22} that conjectured 

%In addition to proving conjectures in~\cite{Poremba22}, our work %}\fi
%Our key technical result is a proof that properties sufficient for deletion follow from specific standard cryptographic assumptions such as LWE, thereby proving conjectures in~\cite{Poremba22}. We also go a step further and 

%\dakshita{technically, our work has the same template as Alex's, so maybe should reword the previous sentence.}

\paragraph{Collapsing Functions.}
%{\em Collapsing} hash functions were introduced in~\cite{EC:Unruh16} as a quantum strengthening of collision-resistance. 
The notion of \emph{collapsing} considers an experiment where a computationally bounded adversary prepares an arbitrary superposition of preimages of $f$ on a register $\mathsf{A}$, after which the challenger tosses a random coin $c$. If $c = 0$, the challenger measures register $\mathsf{A}$, otherwise it measures a register containing the hash $y$ of the value on register $\mathsf{A}$, thus leaving $\mathsf{A}$ holding a superposition of preimages of $y$. The register $\mathsf{A}$ is returned to the adversary, and we say that $f$ is collapsing if the adversary cannot guess $c$ with better than negligible advantage.
%a preimage register containing a superposition of (two or more) preimages of a hash outcome $y$ cannot be computationally distinguished from a register containing a mixture of singleton preimages of $y$, 
%a computational basis measurement of the hash of a quantum superposition of messages is quantum computationally
%indistinguishable from measuring the message superposition itself
%despite both distributions being information-theoretically distinct.
Constructions of collapsing hash functions are known based on LWE~\cite{10.1007/978-3-662-53890-6_6}, low-noise LPN~\cite{crypto-2022-32202}, and more generally on special types of collision-resistant hashes. They have played a key role in the design of post-quantum protocols, especially in settings where proofs of security of these protocols rely on {\em rewinding} an adversary.

It is easy to see that 
\[\mathsf{Com}(b) = \left( y, \ket{0,x_0} + (-1)^b \ket{1,x_1} \right) \]
computationally hides the bit $b$ as long as the function $f$ used to build the commitment above is {\em collapsing}. Indeed, collapsing implies that the superposition $\ket{0,x_0} + (-1)^b \ket{1,x_1}$ is computationally indistinguishable from the result of measurement in the computational basis, and the latter perfectly erases the phase $b$. 
However, $\PVD$ requires  something stronger: we must show that any adversary that generates a valid pre-image of $y$ given the superposition $\ket{0,x_0} + (-1)^b \ket{1,x_1}$, must have {\em information-theoretically} deleted $b$ from its view, despite $b$ being information-theoretically present in the adversary's view before generating the certificate. We show via a careful proof that this is indeed the case for collapsing $f$.
Proving this turns out to be non-trivial. Indeed, a similar construction in~\cite{Poremba22} based on the Ajtai hash function \cite{DBLP:conf/stoc/Ajtai96} relied on an unproven conjecture, which we prove in this work by developing new techniques. 

In addition, we show how $f$ in the template above can be replaced with functions that satisfy weaker properties than collapsing, yielding $\PVD$ from regular variants of one-way functions. We discuss these results below.
%make many other definitional and conceptual contributions.
%, and introduce new proof techniques to obtain the results summarized below.

%\dakshita{language can be refined}
%then the bit $b$ is computationally hidden from the 

\subsection{Our Results}

We introduce new properties of (hash) functions, namely target-collapsing, generalized target-collision-resistance.
We will show that hash functions satisfying these properties (1) can be based on (regular) variants of one-way functions and (2) imply publicly-verifiable deletion in many settings. Our results also use an intermediate notion, a variant of target-collapsing that satisfies certified everlasting security. Before discussing our results, we motivate and discuss these new definitions informally below.

\subsubsection{Definitions}
%We also introduce a certified everlasting variant of target-collapsing,

\paragraph{Target-Collapsing and Generalized Target-Collision-Resistant Functions.}
Towards better understanding the computational assumptions required for $\PVD$, we observe that in the deletion experiment for the commitment above, 
%it is the receiver of the commitment that is adversarial, and 
the superposition $\ket{x_0} + (-1)^b \ket{x_1}$ is prepared by an {\em honest committer}. 
This indicates that the collapsing requirement, where  security is required to hold even for an adversarial choice of superposition over preimages, may be overkill.

Inspired by this, 
%To this end, we show that a cryptographic property that we call {\em target-collapsing} suffices to provably instantiate the framework above. In defining this property, we move beyond two-to-one functions, and consider target-collapsing properties of arbitrary compressing functions.
we consider a natural weakening called {\em target-collapsing}, where {\em the challenger (as opposed to the adversary)} prepares a superposition of preimages of a random image $y$ of $f$ 
%according to some well-defined distribution $D_{\lambda}$ 
on register $\mathsf{A}$. After this, the challenger tosses a random coin $c$. If $c=0$, it does nothing to $\mathsf{A}$, otherwise it measures $\mathsf{A}$ in the computational basis.
%performs a {\em generalized measurement} -- i.e., measures a register containing a general hash-dependent function $M_f$ of the value on register $\mathsf{A}$, thus leaving $\mathsf{A}$ holding a superposition of preimages of $$. 
The register $\mathsf{A}$ is returned to the adversary, and we say that a hash function is target-collapsing if a computationally bounded adversary cannot guess $c$ with better than negligible advantage.

%to general superpositions and measurements.

As highlighted above, this definition weakens collapsing to allow the challenger (instead of the adversary) to prepare the preimage register. 
The weakening turns out to be significant because we show that target-collapsing functions are realizable from relatively weak cryptographic assumptions -- namely variants of one-way functions -- which are unlikely to imply (standard) collapsing or collision-resistant hash functions due to known black-box separations~\cite{10.1007/BFb0054137}.

To enable these instantiations from weaker assumptions, we first further generalize target-collapsing so that 
%the initial superposition of preimages on register $\mathsf{A}$ is sampled by the challenger according to an arbitrary distribution, and so that 
when $c=1$, the challenger applies a {\em binary-outcome measurement} $M$ %\dakshita{trying to keep it simple, hopefully this does not oversimplify} 
to $\mathsf{A}$ (as opposed to performing a computational basis measurement resulting in a singleton preimage). 
Thus, a template commitment with $\PVD$ from generalized target-collapsing hashes has the form:
\[\mathsf{Com}(b) = \left( y, \sum_{x:f(x) = y, M(x) 
= 0} \ket{x} + (-1)^b \sum_{x:f(x) = y, M(x) 
= 1} \ket{x} \right). \]
We show that this commitment satisfies $\PVD$ as long as $f$ is target-collapsing w.r.t. the measurement $M$, and satisfies an additional property of ``generalized'' target-collision-resistance (TCR), that we discuss next.

%In what follows, we summarize all our results and their implications.

Generalized target-collision-resistance is a quantum generalization of the (standard) cryptographic property of second pre-image resistance/target-collision-resistance.
Very roughly, this considers an experiment where the challenger first prepares a superposition of preimages of a random image $y$ of $f$ on register $\mathsf{A}$. After this, the challenger applies a measurement (e.g., a binary-outcome measurement) $M$ on $\mathsf{A}$ to obtain outcome $\mu$ and sends $\mathsf{A}$ to the adversary. We require that no polynomially-bounded adversary given register $\mathsf{A}$ can output {\em any} preimage $x'$ of $y$ such that $M(x') \neq M(\mu)$ (except with negligible probability)\footnote{We remark that this notion can also be seen as a generalization of ``conversion hardness'' defined in  \cite{HMY}.}.
%\dakshita{this previous sentence is not making sense, why measure $y$? I think I probably meant something like $\mathsf{A}$ shouldn't output a preimage that is not in $\mu$?}

%requires that no (polynomially-bounded) adversary, given a superposition over a set containing {\em some} (say half) of the preimages of a random image $y$, is able to find a preimage that is not contained in this set. In more detail, we define $(\cD, \cM)$ target-collision resistant functions as those where the initial distribution 
%For the purposes of this overview, we will in fact restrict ourselves to superpositions 

%\subsection{Our Results}
%\dakshita{smoother transition to the next paragraph}

\paragraph{Certified Everlasting Target-Collapsing.}
In order to show $\PVD$, instead of directly relying on target-collapsing (which only considers computationally bounded adversaries), %Instead, 
we introduce a stronger notion that we call \emph{certified everlasting} target-collapsing. This considers the following experiment: as before, the challenger prepares a superposition of preimages of a random image $y$ of $f$ on register $\mathsf{A}$. After this, the challenger tosses a random coin $c$. 
If $c = 0$, it does nothing to $\mathsf{A}$, otherwise it applies measurement $M$ to $\mathsf{A}$. The register $\mathsf{A}$ is returned to the adversary, after which the adversary is required to return a pre-image of $y$ as its ``deletion certificate''. While such a certificate can be obtained via an honest measurement of the register $\mathsf{A}$, the {\em certified everlasting target-collapsing} property requires that the following {\em everlasting} security guarantee hold. As long as the adversary is computationally bounded at the time of generating a valid deletion certificate, verification of this certificate implies that the bit $c$ is {\em information-theoretically} erased from the adversary's view, and cannot be recovered even given unbounded resources. That is, if the adversary indeed returns a valid pre-image, they will never be able to guess whether or not the challenger applied measurement $M$.
%\dakshita{motivate everlasting some more here or earlier.}

\subsubsection{New Constructions and Theorems}

\paragraph{Main Theorem.} Now, we are ready to state the main theorem of our paper.
In a nutshell, this says that any (hash) function $f$ that satisfies both target-collapsing and (generalized) target-collision resistance also satisfies {\em certified everlasting} target-collapsing.

\begin{theorem}(Informal).
    If $f$ satisfies target-collapsing and generalized target-collision-resistance with respect to measurement $M$, then $f$ satisfies {\em certified everlasting target-collapsing} with respect to the measurement $M$.
\end{theorem}

We also extend recent results from the collapsing literature \cite{cryptoeprint:2022/786,crypto-2022-32202,crypto-2022-32124} to show that for the case of binary-outcome (in fact, polynomial-outcome) measurements $M$, generalized TCR with respect to $M$ actually implies target-collapsing with respect to $M$. Thus, we obtain the following corollary.

%\begin{lemma}(Informal).
    %Any function that satisfies generalized TCR w.r.t. a {\em binary outcome} measurement $M$, also satisfies {\em target-collapsing} w.r.t. the measurement $M$.
%\end{lemma}

\begin{corollary}(Informal).
    If $f$ satisfies generalized target-collision-resistance with respect to a {\em binary-outcome} measurement $M$, then $f$ satisfies {\em certified everlasting target-collapsing} with respect to  the measurement $M$.
\end{corollary}

\paragraph{Resolving the Strong Gaussian Collapsing Conjecture~\cite{Poremba22}.} 
We now apply the main theorem and its corollary to build various cryptographic primitives with $\PVD$. First, we immediately {\bf prove} the following ``strong Gaussian-collapsing''\footnote{Here, ``Gaussian'' refers to a quantum superposition of Gaussian-weighted vectors, where the distribution assigns probability proportional to
$\rho_\sigma(\vec x) = \exp(-\pi \|\vec x \|^2/ \sigma^2)$ for vectors $\vec x \in \Z^m$ and parameter $\sigma >0$.} conjecture from \cite{Poremba22}, which essentially conjectures that the Ajtai hash function (based on the hardness of SIS) satisfies a certain form of key-leakage security after deletion.
This follows from our main theorem because the Ajtai hash function is known to be collapsing \cite{10.1007/978-3-030-26951-7_12,Poremba22} and collision-resistant (which implies that it is target-collapsing and target-collision-resistant when preimages are sampled from the Gaussian distribution).
%the notion of certified everlasting target-collapsing introduced in this work.

\begin{conjecture}[Strong Gaussian-Collapsing Conjecture,~\cite{Poremba22}] 
There exist $n,m,q \in \N$ with $m \geq 2$ and $\sigma > 0$ such that,
for every efficient quantum algorithm $\algo A$,
$$
\Big|\Pr[\mathsf{StrongGaussCollapseExp}_{\algo A,n,m,q,\sigma}(0)=1] - \Pr[\mathsf{StrongGaussCollapseExp}_{\algo A,n,m,q,\sigma}(1)=1]\Big| \leq \negl(\lambda)
$$
with respect to the experiment defined in \Cref{fig:SGC}.
\end{conjecture}

\begin{figure}[!htb]
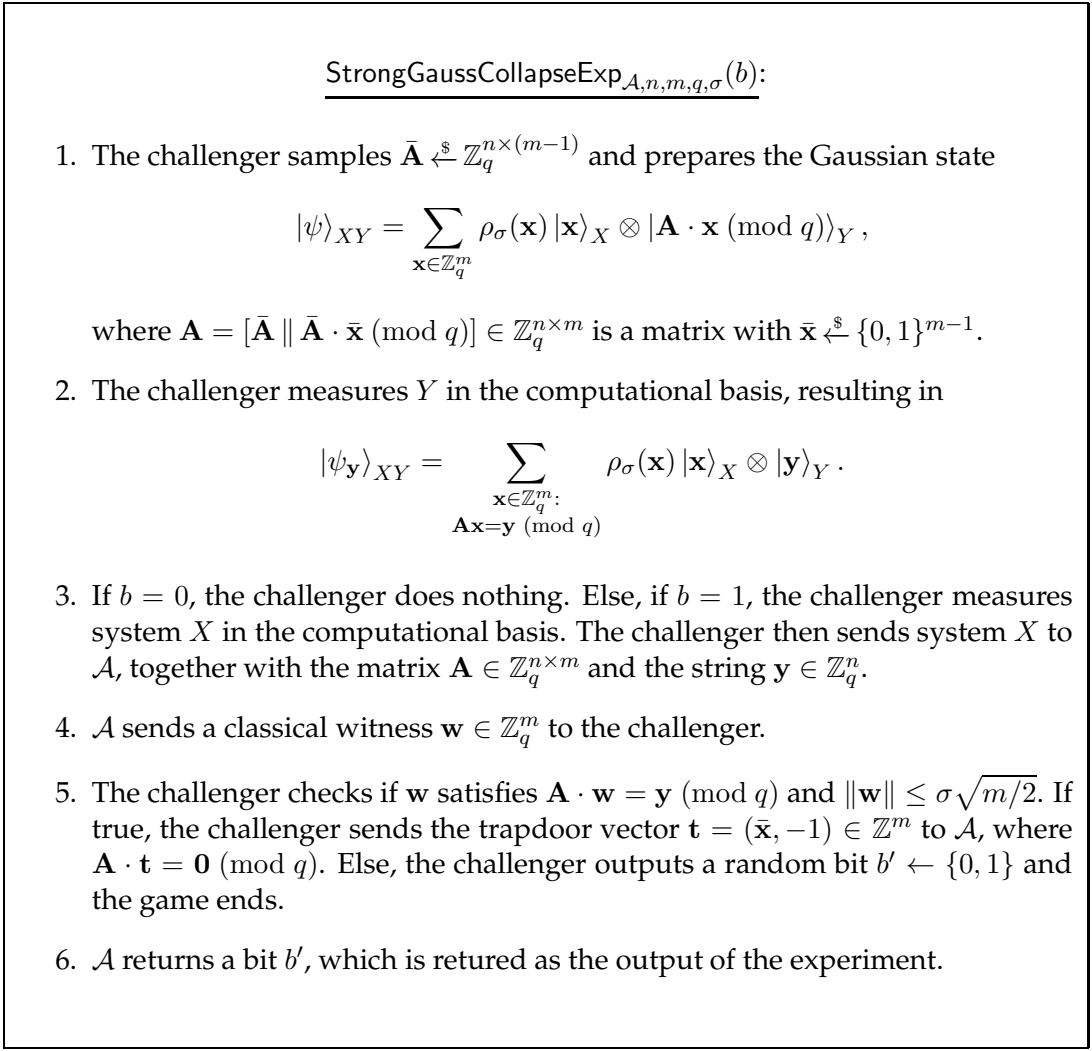

   \begin{center} 
   \begin{tabular}{|p{14cm}|}
    \hline 
\begin{center}
\underline{$\mathsf{StrongGaussCollapseExp}_{\algo A,n,m,q,\sigma}(b)$}: 
\end{center}
\begin{enumerate}
    \item The challenger samples $ \bar{\vec A} \rand \Z_q^{n \times (m-1)}$
    and prepares the Gaussian state
    $$
    \ket{\psi}_{XY} = \sum_{\vec x \in \Z_q^m} \rho_\sigma(\vec x) \ket{\vec x}_X \otimes \ket{\vec A \cdot \vec x \Mod{q}}_Y,
    $$
    where $\vec A = [\bar{\vec A} \, \| \, \bar{\vec A} \cdot \bar{\vec x} \Mod{q}] \in \Z_q^{n \times m}$ is a matrix with $\bar{\vec x} \rand \bit^{m-1}$.

    \item The challenger measures $Y$ in the computational basis, resulting in
    $$
    \ket{\psi_{\vec y}}_{XY} = \sum_{\substack{\vec x \in \Z_q^m:\\ \vec A \vec x= \vec y \Mod{q}}} \rho_\sigma(\vec x) \ket{\vec x}_X \otimes \ket{\vec y}_Y.
    $$
\item If $b=0$, the challenger does nothing. Else, if $b=1$, the challenger measures system $X$ in the computational basis. The challenger then sends system $X$ to $\algo A$, together with the matrix $\vec A \in \Z_q^{n \times m}$ and the string $\vec y \in \Z_q^n$.

\item $\algo A$ sends a classical witness $\vec w \in \Z_q^m$ to the challenger.

\item The challenger checks if $\vec w$ satisfies $\vec A \cdot \vec w = \vec y \Mod{q}$ and $\| \vec w \| \leq \sigma \sqrt{m/2}$. If true, the challenger sends the trapdoor vector $\vec t = (\bar{\vec x},-1) \in \Z^m$ to $\algo A$, where $\vec A \cdot \vec t = \vec 0 \Mod{q}$. Else, the challenger outputs a random bit $b' \gets \{0,1\}$ and the game ends.
 
\item $\algo A$ returns a bit $b'$, which is retured as the output of the experiment.
\end{enumerate}
\ \\
\hline
\end{tabular}
    \caption{The strong Gaussian-collapsing experiment~\cite{Poremba22}.}
    \label{fig:SGC}
    \end{center}
\end{figure}

This conjecture, from~\cite{Poremba22} considers a slightly weaker notion of certified collapsing which resembles the notion of certified deletion first proposed by Broadbent and Islam~\cite{Broadbent_2020}. Here, the adversary is not computationally unbounded once a valid deletion certificate is produced; instead, the challenger simply reveals some additional secret information (in the case of the strong Gaussian-collapsing experiment, the challenger reveals a short trapdoor vector for the Ajtai hash function\footnote{In the strong Gaussian-collapsing experiment it is crucial that the trapdoor is only revealed after a valid certificate is presented; otherwise, the adversary can easily distinguish the collapsed from the non-collapsed world by applying the Fourier transform and using the trapdoor to distinguish $\LWE$ samples from uniformly random vectors~\cite{Poremba22}.}). 

Following results from \cite{Poremba22}, we obtain the following cryptosystems with $\PVD$, for the first time from standard cryptographic assumptions.

\begin{theorem}(Informal)
     Assuming the hardness of $\LWE$ and $\SIS$ with appropriate parameters, there exists public-key encryption and (leveled) fully-homomorphic encryption with $\PVD$.
\end{theorem}

Next, we ask whether one necessarily needs to rely on concrete, highly structured assumptions such as LWE in order to achieve publicly-verifiable deletion, or whether weaker generic assumptions suffice. We present a more general approach to building primitives with $\PVD$ from weaker, generic assumptions.
%We obtain the following results.

\paragraph{Commitments with $\PVD$ from Regular One-Way Functions.} 
We first formulate the notion of a \emph{balanced binary-measurement} TCR hash, which is any function that is TCR with respect to some appropriately balanced binary-outcome measurement. By balanced, we mean that the set of preimages of a random image will have significant weight on preimages that correspond to both measurement outcomes (this will roughly be required to guarantee the binding property of our commitment/correctness properties of our encryption schemes). By roughly following the template described above, we show that such hashes generically imply commitments with $\PVD$. Next, we show that such ``balanced'' functions can be based on (almost-)regular one-way functions\footnote{This is a generalization of regular one-way functions where preimage sets for different images should be polynomially related in size.} By carefully instantiating this outline, we obtain the following results.

\begin{theorem}(Informal).
    Assuming the existence of almost-regular one-way functions, there exists a balanced binary-outcome TCR hash, and consequently there exist commitments with $\PVD$.
\end{theorem}

\iffalse{
\begin{corollary} (Informal).
    Assuming the existence of almost-regular one-way functions, there exist commitments with $\PVD$.
\end{corollary}
}\fi

\paragraph{Public-Key Encryption with $\PVD$ from Regular Trapdoor Functions.}
Next, we take this framework to the public-key setting, showing that any balanced binary-outcome TCR hash with an additional ``trapdoor'' property %of ``trapdoor phase-recoverability'' 
generically implies a public-key encryption scheme with $\PVD$. The additional property roughly requires the existence of a trapdoor for $f$ that enables recovering the phase term from the quantum commitments discussed above: we call this {\em trapdoor phase-recoverability}. We show that balanced binary-outcome TCR, with trapdoor phase-recoverability, can be based on injective trapdoor one-way functions or pseudorandom group actions (the latter builds on~\cite{HMY}). 
%where the latter 
%Using this, we obtain the following results.
%, where the construction from pseudorandom group actions 
%follows immediately from a recent work of~\cite{HMY}.

\begin{theorem} (Informal).
    Assuming the existence of injective trapdoor one-way functions or pseudorandom group actions, there exists a balanced binary-outcome TCR hash with trapdoor phase-recoverability, and consequently there exists public-key encryption with $\PVD$.
\end{theorem}

We also show that injectivity requirement on the trapdoor function can be further relaxed to a notion of ``superposition-invertible'' trapdoor regular one-way function for the results above. Informally, this is a regular one-way function, where a trapdoor allows one to obtain a uniform superposition over all preimages of a given image. This is an example of a {\em generic assumption} that is not known to, and perhaps is unlikely to, imply classical public-key encryption -- but does imply PKE with quantum ciphertexts, and in fact even one that supports $\PVD$. The only other assumption in this category is the concrete assumption that pseudorandom group actions exist~\cite{HMY}.

\begin{theorem} (Informal).
    Assuming the existence of superposition-invertiable regular trapdoor functions, there exists a balanced binary-outcome TCR hash with trapdoor phase-recoverability and consequently, there exists public-key encryption with $\PVD$.
\end{theorem}

\paragraph{Advanced Encryption with $\PVD$ from Weak Assumptions}
Finally, we show that hybrid encryption gives rise to a generic compiler for encryption with $\PVD$, obtaining the following results. 

%approach for building encryption schemes preserves the publicly-verifiable deletion property. 

%Namely, given an arbitrary encryption scheme $X$ and any commitment with $\PVD$, 
%we can obtain an encryption under $X$ of bit $b$ in two parts:
%\begin{itemize}
    %\item First, generate $\mathsf{Com}(b)$ on commit and decommit registers $\mathsf{C}, \mathsf{R}$.
    %\item Output the contents of $\mathsf{C}$, along with an encryption under $X$ of the contents of $\mathsf{R}$.
%\end{itemize}

\begin{theorem}(Informal).
    Assuming the existence of injective trapdoor one-way functions or pseudorandom group actions, and $X \in \{$attribute-based encryption, quantum fully-homomorphic encryption, witness encryption, timed-release encryption$\}$,
    there exists $X$ with $\PVD$.
\end{theorem}

Prior to this work, while there existed encryption schemes with $\PVD$ from non-standard assumptions such as one-shot signatures~\cite{hiroka2021quantum}, conjectured strong collapsing~\cite{Poremba22} or post-quantum indistinguishability obfuscation~\cite{BGGKMRR}, no basic or advanced cryptosystems supporting $\PVD$ were known from standard assumptions.
%none of the basic or advanced primitives above with $\PVD$ were known from standard assumptions. 
We provide a more detailed overview of prior work below.

\subsection{Prior work}

The first notion resembling \emph{certified deletion} was introduced by
Unruh \cite{Unruh2013} who proposed a (private-key) quantum timed-release
encryption scheme that is \emph{revocable}, i.e. it allows a user to \emph{return} the ciphertext of a quantum timed-release encryption scheme, thereby losing all access to the data. Unruh's scheme uses conjugate coding~\cite{Wiesner83,BB84} and relies on the \emph{monogamy of entanglement} in order to guarantee that revocation necessarily erases information about the plaintext. 
Broadbent and Islam~\cite{Broadbent_2020} introduced the notion of \emph{certified deletion} and constructed a private-key quantum encryption scheme with the aforementioned feature which is inspired by the quantum key distribution protocol~\cite{BB84,Tomamichel2017largelyself}. 
In contrast with Unruh's~\cite{Unruh2013} notion of revocable quantum ciphertexts which are eventually returned and verified, Broadbent and Islam~\cite{Broadbent_2020} consider certificates which are entirely classical. %\dakshita{I thought certificates could be classical in Unruh's scheme also? sorry if i'm mistaken}\alex{As far as I know, Unruh's reovcation uses a projective measurement that checks if the returned state is maximally entangled with another target system . Let me check it again}. 
Moreover, the security definition requires that, once the certificate is successfully verified, the plaintext remains hidden even if the secret key is later revealed.
Inspired by the notion of quantum copy-protection~\cite{Aar09},
Ananth and La Placa~\cite{ananth2020secure} defined a form of quantum software protection called \emph{secure software leasing} whose anti-piracy notion requires that the encoded program is returned and verified.

Using a hybrid encryption scheme, Hiroka, Morimae, Nishimaki and Yamakawa~\cite{hiroka2021quantum} extended the scheme in~\cite{Broadbent_2020} to both public-key and attribute-based encryption with privately-verifiable certified deletion via \emph{receiver non-committing} encryption~\cite{10.5555/1756169.1756191,10.1145/237814.238015}. 
Hiroka, Morimae, Nishimaki and Yamakawa~\cite{hiroka2021certified} considered \emph{certified everlasting zero-knowledge proofs} for $\mathsf{QMA}$ via the notion of \emph{everlasting security} which was first formalized by M\"{u}ller-Quade and Unruh~\cite{10.1007/978-3-540-70936-7_3}. Bartusek and Khurana~\cite{cryptoeprint:2022/1178} revisited the notion of certified deletion and presented a unified approach for how to generically convert any public-key, attribute-based, fully-homomorphic, timed-release or witness encryption scheme into an equivalent quantum encryption scheme with certified deletion. In particular, they considered a stronger notion called \emph{certified everlasting security} which allows the adversary to be computationally unbounded once a valid deletion certificate is submitted. This is also the definition we consider in this work.
%\alex{Would you like to mention any other contributions from your paper which might be relevant in this work?} 
In the same spirit, Hiroka, Morimae, Nishimaki and Yamakawa~\cite{cryptoeprint:2022/969} gave a \emph{certified everlasting} functional encryption scheme which allows the receiver of the ciphertext to obtain the outcome specific function applied the plaintext, but nothing else.
In other very recent work, Ananth, Poremba and Vaikuntanathan~\cite{cryptoeprint:2023/325} used Gaussian superpositions to construct (key)-revocable cryptosystems, such as public-key encryption, fully homomorphic encryption and pseudorandom functions assuming the hardness of $\LWE$, and Agarwal et al. \cite{AKNYY} introduced a generic compiler for adding key-revocability to a variety of cryptosystems. In these systems, the cryptographic key consists of a quantum state which can later be \emph{certifiably revoked} via a quantum channel -- in contrast with the classical deletion certificates for ciphertexts considered in this work.

\paragraph{Cryptosystems with Publicly Verifiable Deletion.}
First, in addition to their results in the setting of private verification,~\cite{hiroka2021quantum} also gave a public-key encryption scheme with certified deletion which is \emph{publicly verifiable} assuming the
existence of one-shot signatures (which rely on strong black-box notions of obfucation) and extractable witness encryption. Using 
\emph{Gaussian superpositions},
Poremba~\cite{Poremba22} proposed \emph{Dual-Regev}-based public-key and fully homomorphic encryption schemes with certified deletion which are publicly verifiable and proven secure assuming the (then unproven) \emph{strong Gaussian-collapsing conjecture} --- a strengthening of the collapsing property of the Ajtai hash.
Finally, a recent work~\cite{BGGKMRR} relies on post-quantum indistinguishability obfuscation (iO) to obtain both publicly verifiable deletion and publicly verifiable key revocation. This is a strong assumption for which we have candidates, but no constructions based on standard (post-quantum) assumptions at this time. 
%On the other hand, our work relies
%than the ones we use in this work, and at this time, we only have candidate post-quantum iO schemes under plausible assumptions on leaky variants of LWE.

\section{Technical Overview}

In this overview, we begin by discussing the key ideas involved in proving our main theorem. We show how to prove publicly verifiable deletion for a toy protocol that relies on stronger assumptions than the ones that we actually rely on in our actual technical sections. 

Next, we progressively relax these assumptions to instantiate broader frameworks, including the one from~\cite{Poremba22}, obtaining public-key encryption and fully-homomorphic encryption with $\PVD$ from LWE/SIS.

Finally, we further generalize this to enable constructions from weak cryptographic assumptions -- including commitments with $\PVD$ from variants of one-way functions and PKE with $\PVD$ from trapdoored variants of the same assumption. We also discuss a hybrid approach that enables a variety of advanced encryption schemes supporting $\PVD$.

\subsection{Proving Our Main Theorem}
\label{sec:overreq}
Consider the toy commitment
\[\mathsf{Com}(b) = \left( y, \ket{0, x_0} + (-1)^b \ket{1, x_1} \right) \]
where $(0, x_0), (1, x_1) $ are preimages of $y$ under a structured two-to-one function $f$, where every image has a preimage that begins with a $0$ and another that begins with a $1$.  
We note that this commitment can be efficiently prepared by first preparing a superposition over all preimages 
\[\sum_{b \in \{0,1\}, x \in \{0,1\}^\secp} \ket{b,x}\] 
on a register $\mathsf{X}$, then writing the output of $f$ applied on $X$ to register $\mathsf{Y}$, and finally measuring the contents of register $\mathsf{Y}$ to obtain image $y$. The register $\mathsf{X}$ contains $\ket{0, x_0} + \ket{1, x_1}$, which can be converted to $\ket{0, x_0} + (-1)^b \ket{1, x_1}$ via (standard) phase kickback.

%is target collapsing and target collision-resistance w.r.t. a computational basis measurement of the pre-image register.

%We will show that this commitment satisfies publicly-verifiable deletion as long as $f$ is target collapsing and target collision-resistant w.r.t. a computational basis measurement of the pre-image register. 
To show  that the commitment satisfies publicly-verifiable deletion, 
%we must demonstrate that the advantage of 
we consider an adversary $\cA = (\cA_1, \cA_2)$ where $\cA_1$ is (quantum) polynomial time and $\cA_2$ is unbounded, participating in the following experiment.
%is negligible in following experiment.
%$c$ with probability negligibly better than $\frac{1}{2}$ in the following experiment:

\begin{itemize}
\item The challenger samples $b \leftarrow \{0,1\}$ and runs $\mathsf{Expmt}_0(b)$, described below.\\
\noindent{\underline{$\mathsf{Exmpt}_0(b):$}}
\begin{enumerate}
    \item Prepare $\left( \ket{0, x_0} + (-1)^b \ket{1, x_1}, y \right)$ on registers $\mathsf{A}, \mathsf{B}$ and send them to $\cA_1$.
    \item $\cA_1$ outputs a (classical) deletion certificate $\gamma$,\footnote{If the $\cA_1$ outputs a quantum state as their certificate, the state is measured in the computational basis to obtain a classical certificate $\gamma$.} and left-over state $\rho$.
    \item If $f(\gamma) \neq y$, output a uniformly random bit $b' \gets \{0,1\}$, otherwise output $b' = \cA_2(\rho)$.
\end{enumerate}
\item The advantage of $\cA$ is defined to be $\mathsf{Adv}_{\cA}^{\mathsf{Expmt}_0} = \big| \Pr[b' = b]  - \frac{1}{2} \big|$.
\end{itemize}

We discuss how to prove the following.
\begin{claim} (Informal).
\label{clm:overclm}
For every $\cA = (\cA_1, \cA_2)$ where $\cA_1$ is (quantum) computationally bounded, \[\mathsf{Adv}_{\cA}^{\mathsf{Expmt}_0} = \negl(\secp), \]
as long as $f$ is target collapsing and target collision-resistant w.r.t. a computational basis measurement of the pre-image register.
\end{claim}

\iffalse{
We split the proof of this claim into two parts: (1)  show that the function $f$ is also {\em certified everlasting} target-collapsing, and (2) show that $f$ being {\em certified everlasting} target-collapsing implies that $\mathsf{Com}$ satisfies publicly verifiable deletion.

\paragraph{Part 1: $f$ is Certified Everlasting Target-Collapsing.}
}\fi

\paragraph{Overview of the Proof of Claim \ref{clm:overclm}.}
To prove this claim, we must  show that $b$ %indicating whether or not the superposition of preimages was measured, 
is information-theoretically {\em removed} from the leftover state of any $\cA_1$ that generates a valid pre-image of $y$, despite the fact that the adversary's view contains $b$ at the beginning of the experiment.

Proof techniques for this type of experiment were recently introduced in~\cite{cryptoeprint:2022/1178} in the context of {\em privately verifiable deletion} via BB84 states. Inspired by their method, our first step is to defer the dependence of the experiment on the bit $b$. 
In more detail, we will instead imagine sampling
the distribution by guessing a uniformly random $c \leftarrow \{0,1\}$, and initializing the adversary with $\left( \ket{x_0} + (-1)^c \ket{x_1}, y \right)$.
%when $c = 0$, or a uniform mixture of $\{\ket{x_0}, \ket{x_1}\}$ when $c = 1$. 
The challenger later obtains input $b$ and aborts the experiment (outputs $\bot$) if $c \neq b$.
Since $c$ was a uniformly random guess, the trace distance between the $b = 0$ and $b = 1$ outputs of this modified experiment is at least half the trace distance between the outputs of the original experiment.
%Now, the bit $b$ is only used by the experiment to determine whether or not to output $\bot$. 
Moreover, we can further delay the process of obtaining input $b$, and then abort or not until {\em after} the adversary outputs a certificate of deletion.
That is, we can consider a {\em purification} where a register $\mathsf{C}$ contains a superposition $\ket{0} + \ket{1}$ of two choices for $c$, and is later measured to determine bit $c$. This experiment is discussed in detail below.\\

\noindent 
\underline{$\mathsf{Expmt}_1(b):$}
The experiment proceeds as follows.
\begin{enumerate}
\item Prepare the $\ket{+}$ state on an ancilla register $\mathsf{C}$, and a superposition of preimages $\ket{x_0} + \ket{x_1}$ of a random $y$ on register $\mathsf{A}$.
%\dakshita{say somewhere in the beginning how to prepare this.}
\item 
Then, controlled on the contents of register $\mathsf{C}$, do the following:
if the control bit is $0$, do nothing, and otherwise flip the phase on $x_1$ (via phase kickback), changing the contents of $\mathsf{A}$ to $\ket{x_0} - \ket{x_1}$. 
This means that the overall state is
\[
\frac{1}{\sqrt{2}} \sum_{c \in \{0,1\}} \ket{c}_{\mathsf{C}} \otimes \ket{0, x_0}_{\mathsf{A}} + (-1)^c \ket{1, x_1}_{\mathsf{A}}
\]
Send $\mathsf{A}$ to $\cA_1$.  
\item Obtain from $\cA_1$ a purported certificate of deletion $\gamma$.
\item If $f(\gamma) \neq y$, abort, and otherwise measure register $\mathsf{C}$ to obtain output $c$, and abort if $c \neq b$. In the case of abort, output a uniformly random bit $b' \gets \{0,1\}$.
\item If no aborts occurred, output $b' = \cA_2(\rho)$.
\end{enumerate}
We note that the event $c = b$ occurs with probability exactly $\frac{1}{2}$, and since measurements on separate subsystems commute, we have that 
\begin{equation}
\label{eq:over1}
\mathsf{Adv}_{\cA}^{\mathsf{Expmt}_1} \geq \frac{1}{2} \mathsf{Adv}_{\cA}^{\mathsf{Expmt}_0}.
\end{equation}
where $\mathsf{Adv}_{\cA}^{\mathsf{Expmt}_1} = \big| \Pr[\mathsf{Expmt}_1(b) = b]  - \frac{1}{2} \big|$ for $b \leftarrow \{0,1\}$.
%\dakshita{clean or define expmt 0, and the advantage notatiob.}

Once the dependence of the experiment on $b$ has been deferred, as above, we can consider another experiment (described below) where the challenger measures the contents of register $\mathsf{A}$ {\em before} sending it to $\cA_1$. Intuitively, performing this measurement {\em removes} information about $b$ from $\cA_1$'s view in a manner that is computationally undetectable by $\cA_1$ (due to the target-collapsing property of $f$). \\

\noindent \underline{$\mathsf{Expmt}_2(b):$}
The experiment proceeds as follows.
\begin{itemize}
\item Prepare the $\ket{+}$ state on an ancilla register $\mathsf{C}$, and a superposition of preimages $\ket{x_0} + \ket{x_1}$ of a random $y$ on register $\mathsf{A}$.
\emph{Next, measure register $\mathsf{A}$ in the computational basis.}

Then, controlled on the contents of register $\mathsf{C}$, do the following:
if the control bit is $0$, do nothing, and otherwise flip the phase on $x_1$. 
This means that the overall state is a uniform mixture of the states
\[\frac{1}{\sqrt{2}} \sum_{c \in \{0,1\}} \ket{c}_{\mathsf{C}} \otimes \ket{0, x_0}_{\mathsf{A}} \text{ and }
\frac{1}{\sqrt{2}} \sum_{c \in \{0,1\}} (-1)^c\ket{c}_{\mathsf{C}} \otimes \ket{1, x_1}_{\mathsf{A}}
\]
Finally, send $\mathsf{A}$ to $\cA_1$. 
%controlled on the contents of register $\mathsf{C}$, do the following:
%if the control is $0$, send $\mathsf{A}$ to the adversary, and otherwise flip the phase on the contents of $\mathsf{A}$ to obtain $\ket{x_0} - \ket{x_1}$ before sending it to $\cA_1$. 
\item Obtain from $\cA_1$ a purported certificate of deletion $\gamma$.
\item If $f(\gamma) \neq y$, abort, otherwise measure register $\mathsf{C}$ to obtain output $c$, and abort if $c \neq b$. In the case of abort, output a uniformly random bit $b' \gets \{0,1\}$.
\item If no aborts occurred, output $b' = \cA_2(\rho)$.
\end{itemize}

%\paragraph{A Failed Argument.} 
As described above, the target-collapsing property of $f$ implies that $\cA_1$ cannot (computationally) distinguish the register $\mathsf{A}$ obtained in $\mathsf{Expmt}_2(b)$ from the one obtained in $\mathsf{Expmt}_1(b)$. However, this is not immediately helpful:  information about which experiment $\cA_1$ participated in could potentially be encoded into $\cA_1$'s left-over state $\rho$, so that it remains computationally hidden from $\cA_1$ but can be extracted by (unbounded) $\cA_2$. And it is after all the output of $\cA_2$ that determines the advantage of $\cA$. 
Because of $\cA_2$ being unbounded and the experiments only being {\em computationally} indistinguishable, even if we could show that $\mathsf{Adv}_{\cA}^{\mathsf{Expmt}_2} = \negl(\secp)$, it is unclear how to use this to show our desired claim, i.e., $\mathsf{Adv}_{\cA}^{\mathsf{Expmt}_0} = \negl(\secp)$.
It may appear that the proof is stuck.

To overcome this issue, we will aim to identify an {\em efficiently computable predicate} of the challenger's system, which will {\em imply} the following (inefficient) property: when $\cA_1$ outputs a valid deletion certificate, even an unbounded $\cA_2$ cannot determine whether it participated in $\mathsf{Expmt}_1(b)$ or $\mathsf{Expmt}_2(b)$, i.e.,  $\cA_1$'s left-over state is information-theoretically independent of $b$.

\paragraph{Identifying an Efficiently Computable Predicate.} 
Observe that in $\mathsf{Expmt}_2(b)$, the ancilla register $\mathsf{C}$ is {\em unentangled} with the rest of the experiment. 
In fact, the ancilla register is exactly $\ket{+}$ when we give the adversary $\ket{0,x_0}$ on register $\mathsf{A}$, and $\ket{-}$ when we give the adversary $\ket{1,x_1}$ on register $\mathsf{A}$.
Moreover, in $\mathsf{Expmt}_2(b)$, the \underline{target-collision-resistance} of $f$ implies that the computationally-bounded $\cA_1$ given $x_0$ cannot output $x_1$ as their deletion certificate (and vice-versa).

This, along with the fact that the certificate {\em must} be a pre-image of $y$ means that the following guarantee holds   (except with negligible probability) in $\mathsf{Expmt}_2(b)$:
\begin{center}{\em When the adversary outputs a valid certificate $\gamma$, a projection of the pre-image register onto $\ket{+}$ succeeds if $\gamma = (0, x_0)$ and a projection of the pre-image register onto $\ket{-}$ succeeds if $\gamma = (1, x_1)$.}\end{center}

%In other words, given a valid deletion certificate beginning with bit $\beta$, projects register $\mathsf{C}$ onto $\ket{0} + (-1)^b \ket{1}$.
%That is, we can define an efficient project $\Pi_b$  is indeed an efficient projection that succeeds except with negligible probability in $\mathsf{Expmt}_2^b$, when the adversary generates a valid deletion certificate.

At this point, we can rely on the \underline{target-collapsing} property of $f$ to prove the following claim: the {\em efficient projection} described above also succeeds except with negligible probability in $\mathsf{Expmt}_1(b)$, when the adversary generates a valid deletion certificate. If this claim is not true, then since the experiments (including $\mathsf{\cA}_1$) run in quantum polynomial time until the point that the deletion certificate is generated, and the projection is efficient, one can build a reduction that contradicts target-collapsing of $f$. This reduction obtains a challenge (which is either a superposition when the challenger did not measure, or a mixture if the challenger did measure) on register $\mathsf{A}$, prepares ancilla $\mathsf{C}$ as in $\mathsf{Expmt}_1(b)$, then follows steps 2, 3 identically to $\mathsf{Expmt}_1(b)$. 
Next, given a deletion certificate $(\beta,x_\beta)$, the reduction projects $\mathsf{C}$ onto $\ket{0} + (-1)^\beta \ket{1}$, outputting $1$ if the projection succeeds and $0$ otherwise. 

%Having established that when the adversary outputs a valid deletion certificate $(\beta||x_\beta)$ in $\mathsf{Expmt}_1(b)$, then the projection $\Pi_\beta$ of $\mathsf{C}$ onto $\ket{0} + (-1)^\beta \ket{1}$ almost always succeeds, 
\paragraph{Introducing an Alternative Experiment.}
Having established that the projection above must succeed in $\mathsf{Expmt}_1(b)$ except with negligible probability,
we can now consider an alternative experiment $\mathsf{Expmt}_{\mathsf{alt}}(b)$. This is identical to $\mathsf{Expmt}_1(b)$, except that the challenger {\em additionally} projects register $\mathsf{C}$ onto $\ket{0} + (-1)^{\beta} \ket{1}$ when the adversary generates a valid certificate $(\beta,x_\beta)$. We established above that the projection is successful in $\mathsf{Expmt}_1(b)$ except with negligible probability, and this implies that 
\begin{equation} 
\label{eq:over2}
\mathsf{Adv}_{\cA}^{\mathsf{Expmt}_{\mathsf{alt}}} \geq \mathsf{Adv}_{\cA}^{\mathsf{Expmt}_{1}} - \negl(\secp)
\end{equation}
where as before, $\mathsf{Adv}_{\cA}^{\mathsf{Expmt}_{\mathsf{alt}}} = \big| \Pr[\mathsf{Expmt}_{\mathsf{alt}}(b) = b]  - \frac{1}{2} \big|$ for $b \leftarrow \{0,1\}$.

Crucially, in $\mathsf{Expmt}_{\mathsf{alt}}(b)$, the bit $c$ is determined by a measurement on register $\mathsf{C}$ which is {\em unentangled} with the system and in either the $\ket{+}$ or $\ket{-}$ state (due to the projective measurement that we just applied). Thus, measuring $\mathsf{C}$ in the computational basis results in a uniformly random and independent $c$. By definition of the experiment (abort when $b \neq c$, continue otherwise) -- this implies that the bit $b$ is set in a way that is uniformly random and independent of the adversary's view, and thus  
\begin{equation}
\label{eq:over3}
\mathsf{Adv}_{\cA}^{\mathsf{Expmt}_{\mathsf{alt}}} = 0
\end{equation}
Now, equations~(\ref{eq:over1}, \ref{eq:over2}, \ref{eq:over3}) together yield the desired claim, that is, $\mathsf{Adv}_{\cA}^{\mathsf{Expmt}_{0}} = \negl(\secp)$.

This completes a simplified overview of our key ideas, assuming the existence of a perfectly $2$-to-$1$ function $f$ where every image $y$ has preimages $\left( (0, x_0), (1, x_1) \right)$, and where $f$ satisfies both target-collapsing and target-collision-resistance.
Unfortunately, we do not know how to build functions satisfying these clean properties from simple generic assumptions. Instead, we will generalize the template above, where the first generalization will no longer require $f$ be $2$-to-$1$.

\paragraph{Generalizing the Template.}
First, note that we can replace $\ket{0, x_0}$ and $\ket{1, x_1}$ with superpositions over two disjoint sets of preimages of $y$ separated via an efficient binary-outcome measurement, namely
\[
\mathsf{Com}(b) = 
\sum_{x: f(x) = y, M(x) = 0} \ket{x} + (-1)^b \sum_{x: f(x) = y, M(x) = 1} \ket{x}
\]
We can even consider measurements $M$ that have arbitrarily many outcomes.
Proof ideas described above also generalize almost immediately to show that for any $M$, $\mathsf{Com}$ satisfies $\PVD$ as long as $f$ is target-collapsing and target-collision resistant w.r.t. $M$.
In fact, we can generalize this even further (see our main results in Section~\ref{sec:maintheorem},~\ref{sec:mainaux}) to consider arbitrary (as opposed to uniform) distributions over pre-images, as well as to account for any auxiliary information that may be sampled together with the description of the hash function.

\paragraph{Certified Everlasting Target-Collapsing.}
As discussed in the results section, our actual technical proofs proceed in two parts. (1) Show that for any $M$, a function $f$ that is target-collapsing and target-collision resistant w.r.t. $M$ is also {\em certified everlasting} target-collapsing w.r.t. $M$, and (2) show that $f$ being {\em certified everlasting} target-collapsing implies that $\mathsf{Com}$ satisfies publicly verifiable deletion.

Recall that {\em certified} everlasting target collapsing requires that an adversary that outputs a valid deletion certificate information-theoretically loses the bit $b$ determining whether they received a superposition or a mixture of preimages. Our proof of certified everlasting target-collapsing follows analogously to the proof sketched above. In short, we defer measurement of a bit $b$ which decides whether the adversary is given a superposition or a mixture, and then rely on target-collapsing and target-collision-resistance to argue that an efficient projection on the challenger's state (almost) always succeeds when the adversary outputs a valid certificate. We finally show that success of this projection implies that the adversary's state is information-theoretically independent of $b$.

The certified everlasting target-collapsing property almost immediately implies certified deletion security of $\mathsf{Com}$ via a hybrid argument: 
\begin{itemize}
\item In $\mathsf{Hyb}_0$, the adversary obtains register $\mathsf{A}$ containing
\[
\mathsf{Com}(0) = 
\sum_{x: f(x) = y, M(x) = 0} \ket{x} + \sum_{x: f(x) = y, M(x) = 1} \ket{x}
\]
\item In $\mathsf{Hyb}_1$, the measurement $M$ is applied to $\mathsf{A}$ before sending it to the adversary.
\item In $\mathsf{Hyb}_2$, the adversary obtains register $\mathsf{A}$ containing
\[
\mathsf{Com}(1) = 
\sum_{x: f(x) = y, M(x) = 0} \ket{x} - \sum_{x: f(x) = y, M(x) = 1} \ket{x}
\]
\end{itemize}
The certified everlasting hiding property of $f$ guarantees that all hybrids are statistically close when the adversary outputs a valid deletion certificate. Moreover, these experiments abort and output a random bit when the adversary does not output a valid certificate, and it is easy to show (by computational indistinguishability) that the probability of generating a valid certificate remains negligibly close between experiments.

\paragraph{TCR Implies Target-Collapsing for Polynomial-Outcome Measurements}
We also show that when $M$ has polynomially many possible outcomes, then TCR implies target-collapsing w.r.t. $M$. This follows from techniques that were recently developed in the literature on collapsing versus collision resistant hash functions~\cite{cryptoeprint:2022/786,crypto-2022-32202,crypto-2022-32124}. In a nutshell, these works showed that any distinguisher that distinguishes mixtures from superpositions over preimages for an {\em adversarially chosen} image $y$, can be used to swap between pre-images, and therefore find a collision for $y$. We observe that their technique is agnostic to whether the image $y$ is chosen randomly (in the targeted setting) or adversarially. Furthermore, it also extends to swapping superpositions over sets of pre-images to superpositions over other sets. These allow us to prove (Section~\ref{sec:tcr-implies}) that TCR w.r.t. any polynomial-outcome measurement $M$ implies target-collapsing w.r.t. $M$.

\subsection{Publicly-Verifiable Deletion via Gaussian Superpositions}
In \Cref{sec:Dual-Regev}, we revisit the \emph{Dual-Regev} public-key and (leveled) fully homomorphic encryption schemes with publicly-verifiable deletion which were proposed by Poremba~\cite{Poremba22}
and were conjectured to be secure under the \emph{strong Gaussian-collapsing property}. 
By applying our main 
%result from \Cref{thm:CETC-generalization} 
theorem to the Ajtai hash function, we obtain a proof of the conjecture, which allows us to show the certified everlasting security of the aforementioned schemes assuming the hardness of the $\LWE$ assumption.

The constructions introduced in~\cite{Poremba22} exploit the the duality between $\LWE$ and $\SIS$~\cite{cryptoeprint:2009/285}, and rely on the fact that one encode Dual-Regev ciphertexts via Gaussian superpositions. Below, we give a high-level sketch of the basic public-key construction.

\begin{itemize}
\item To generate a pair of keys $(\sk,\pk)$, sample a random $\vec A \in \Z_q^{n \times (m+1)}$ together with a particular short trapdoor vector $\vec t \in \Z^{m+1}$ such that $\vec A \cdot \vec t = \vec 0 \Mod{q}$. Let $\pk = \vec A$ and $\sk = \vec t$.

\item To encrypt $b \in \bit$ using $\pk=\vec A$, generate the following for a random $\vec y \in \Z_q^n$:
$$
\vk \leftarrow (\vec A,\vec y), \quad\quad
\ket{\ct} \leftarrow \sum_{\vec s \in \Z_q^n} \sum_{\vec e \in \Z_q^{m+1}} \rho_{q/\sigma}(\vec e) \, \omega_q^{-\langle\vec s,\vec y \rangle} \ket{\vec s^\intercal \vec A + \vec e^\intercal +b \cdot (0,\dots,0, \lfloor\frac{q}{2} \rfloor)},
$$
where $\vk$ is a public verification key and $\ket{\ct}$ is the quantum ciphertext for $\sigma >0$.

\item To decrypt $\ket{\ct}$ using $\sk$, measure in the computational basis to obtain $\vec c \in \Z_q^{m+1}$, and output $0$, if $\vec c^\intercal \cdot \sk\in \Z_q$
is closer to $0$ than to $\lfloor\frac{q}{2}\rfloor$,
and output $1$, otherwise. Here $\sk = \vec t$ is chosen such that $\vec c^\intercal \cdot \sk$ yields an approximation of $b \cdot \lfloor \frac{q}{2} \rfloor$ from which we can recover $b$.
\end{itemize}
To delete the ciphertext $\ket{\ct}$, perform a measurement in the Fourier basis. Poremba~\cite{Poremba22} showed that the Fourier transform of $\ket{\ct}$ results in the \emph{dual} quantum state given by
$$
\ket{\widehat{\ct}}=\sum_{\substack{\vec x \in \Z_q^{m+1}:\\ \vec A \vec x = \vec y \Mod{q}}}\rho_{\sigma}(\vec x) \, \omega_q^{\langle \vec x,b \cdot (0,\dots,0,  \lfloor\frac{q}{2} \rfloor)\rangle} \,\ket{\vec x}.
$$
In other words, a Fourier basis measurement of $\ket{\ct}$ will necessarily erase all information about the plaintext $b \in \bit$ and results in a \emph{short} vector $\pi \in \Z_q^{m+1}$ such that $\vec A \cdot \pi = \vec y \Mod{q}$. To publicly verify a deletion certificate, simply check whether a certificate $\pi$ is a solution to the (inhomogenous) $\SIS$ problem specified by $\vk=(\vec A,\vec y)$. Due to the hardness of the $\SIS$ problem, it is computationally difficult to produce a valid deletion certificate from $(\vec A,\vec y)$ alone.

Our approach to proving certified everlasting security of the Dual-Regev public-key and fully-homomorphic encryption schemes with publicly-verifiable deletion in~\cite{Poremba22} is as follows.
First, we observe that the Ajtai hash function is both target-collapsing and target-collision-resistant with respect to the 
 discrete Gaussian distribution. Here, the former follows from $\LWE$ as a simple consequence of the \emph{Gaussian-collapsing property} previously shown by Poremba~\cite{10.1007/978-3-030-26951-7_12,Poremba22}, whereas the latter follows immediately from the quantum hardness of $\SIS$. 
 %By invoking our main result in \Cref{thm:CETC-generalization}, this 
 Thus, our main theorem implies that the Ajtai hash function is certified-everlasting target-collapsing (see \Cref{thm:ajtai-certified-everlasting}). Finally, as a simple corollary of our theorem, we obtain a proof of the \emph{strong Gaussian-collapsing conjecture} in \cite{Poremba22}, which we state in \Cref{SGC}. We also note that the aforementioned conjecture considers a weaker notion of certified collapsing which resembles the notion of certified deletion first proposed by Broadbent and Islam~\cite{Broadbent_2020}. Here, the adversary is not computationally unbounded once a valid deletion certificate is produced; instead, the challenger simply reveals additional secret information (in the case of the strong Gaussian-collapsing experiment, this is a short trapdoor vector for the Ajtai hash function). Our notion of certified everlasting target-collapsing is significantly stronger; in particular, it implies the weaker collapsing scenario considered by Poremba~\cite{Poremba22}. This follows from the fact that the security reduction can simply brute-force search for a short trapdoor solution for the Ajtai hash once it enters the phase in which it is allowed to be computationally unbounded. We exploit this fact in the proof of \Cref{SGC}.

\subsection{Weakening Assumptions for Publicly-Verifiable Deletion}
Next, we look for instantiations of the above template from {\em generic} cryptographic assumptions, as opposed to structured specific assumptions such as LWE.
Here, all of our instantiations only require us to consider functions that are target-collision-resistant and target-collapsing w.r.t. binary-outcome measurements (and as discussed above, TCR implies certified-everlasting target-collapsing in this setting).
%Techniques discussed above can already show that such functions satisfy certified everlasting target-collapsing. %However, to build commitments and other primitives such as encryption from such functions, we also need the bit $b$ to be (approximately) recoverable, with at least some inverse polynomial probability. 
In addition, for the case of commitments, in order for the commitment to satisfy binding\footnote{We  actually prove that a purification of the template commitment described above satisfies honest-binding~\cite{Yan}.
Namely, the committer generates the state above but leaves registers containing the image $y$ (and the key, if $f$ is a keyed function) unmeasured, and holds on to these registers for the opening phase. It can later either open the commitment by sending these registers to a receiver, or request deletion, by measuring them and publishing $y$ (and any keys for the function).
}, we require that there is a measurement that can distinguish 
\[\sum_{x: f(x) = y, M(x) = 0} \ket{x} + \sum_{x: f(x) = y, M(x) = 1} \ket{x}\] from \[\sum_{x: f(x) = y, M(x) = 0} \ket{x} - \sum_{x: f(x) = y, M(x) = 1} \ket{x}\] with probability $\delta$ for any constant $0 < \delta \leq 1$.
For the case of public-key encryption, we similarly require that a trapdoor be able to recover the phase with constant probability. We then resort to standard amplification techniques to boost correctness error from constant to (negligibly close to) $0$. We note that this amplification would also work if the phase was recoverable with inverse-polynomial $\delta$ (as opposed to constant); however, we focus on constant $\delta$ because of simplicity, and because it suffices for our instantiations.

In the template above, we observe that a measurement can find the phase with inverse polynomial probability whenever the sets 
\[ \sum_{x: f(x) = y, M(x) = 0} \ket{x} \text{ and } \sum_{x: f(x) = y, M(x) = 1} \ket{x} \]
are somewhat ``balanced'', i.e. for a random image $y$, for sets $S_0 = \{x: f(x) = y, M(x) = 0\}$  and $S_1 = \{x: f(x) = y, M(x) = 1\}$, we have that $\frac{|S_0|}{|S_1|}$ is a fixed constant.
We show in \cref{sec:com} and \cref{sec:pke} that commitments and PKE with $\PVD$ can be obtained from appropriate variants of TCR functions following this template.
%\dakshita{technically also possible with less balanced sets}

Now, our goal is to build such TCR functions from generic assumptions. A natural idea would be to start with any one-way function $f$ and compose it with a random two-to-one hash $h$ defined on its range\footnote{The {\em co-domain} of a function $f:\{0,1\}^n \rightarrow \{0,1\}^m$ is $\{0,1\}^m$, and we will also refer to this as the {\em range} of the function in this paper. The {\em image} is the set of all actual output values of $f$, i.e. the set $\{y: \exists x \text{ such that } f(x) = y\}$. The co-domain/range may in general be a superset of the image of a function.}.
%\footnote{By range (or codomain), we mean a set of elements in which the output of $f$ is guaranteed to fall, which may be a strict superset of its image.}
Then, any output $y$ of the composed function $(h \circ f)$ is associated with two elements $\{z_0,z_1\} = h^{-1}(y)$ in the range of $f$, and the binary-outcome measurement would measure one of $z_0$ or $z_1$. Recalling that we eventually want to prove target-collision-resistance, the hope would be that just given a superposition over the preimages of, say, $z_0$, the one-wayness of $f$ would imply the difficulty of finding a preimage of $z_1$\footnote{More concretely, a purported reduction to one-wayness when given challenge image $z_1$, can sample a random image $z_0$ with its preimages, then find $h$ s.t. $h(z_0) = h(z_1)$, thereby using a TCR adversary to find a preimage of the given challenge $z_1$.}. This could give the type of TCR property we need.

\paragraph{Technical Bottlenecks, and a Resolution.} Unfortunately, there are two issues with the approach proposed above. 
First, $f$ may be extremely unbalanced, so that the relative sizes of the sets of preimages of two random points $y_1, y_2$, i.e. $|\{x: f(x) = y_1\}|$ and $|\{x: f(x) = y_2\}|$ in its image may have very different sizes, that are not polynomially related with each other. 
%\dakshita{I had some trouble parsing the previous sentences - James, what did you mean?}
There may even be many points in the co-domain/range that have \emph{zero} preimages (for a general OWF, we cannot guarantee that its image is equal to its range). A second related issue is that the above sketched reduction to one-wayness may not work. Let's say we choose $h$ to be a two-to-one function defined by a random shift $\Delta$, i.e. $h(x) = h(x \oplus \Delta)$. Then we are essentially asking that it be hard to invert a random \emph{range} element of $f$, as opposed to $f(x)$ for a random \emph{domain} element $x$, which is the standard one-wayness assumption.

We don't know how to make this approach work from arbitrary one-way functions, which we leave as an open question. Instead, we appeal to a result of~\cite{balancedOWF}, who in the classical context of building statistically hiding commitments, show the following result. 
%Starting with an (almost)-\emph{regular} one-way function and hashing its range down with an appropriate universal hash function results in a function with exactly the properties we need: 
By appropriately combining an (almost)-\emph{regular}\footnote{An almost regular one-way function generalizes regular one-way functions to require only that for any two images $y_1, y_2$ of the function, the sizes of preimage sets of $y_1, y_2$ are polynomially related. In particular, injective functions, and (standard) regular functions also satisfy almost-regularity.} one-way function with universal hash functions, it {\em is} possible to obtain a function $f$ with exactly the required properties:
sufficiently balanced, and one-way over its \emph{range}. The former property means that an overwhelming fraction of range elements have similar-sized preimage sets, while the latter property says that an element $y$ sampled randomly from the range of the function cannot be inverted except with negligible probability. This resolves both the difficulties above.

Given such a balanced function $f$, we apply a random two-to-one hash $h$ defined by a shift $\Delta$ to the range of this $f$. We prove in Section~\ref{sec:almostreg} that this implies the flavor of target-collision-restistant hash that we need to construct commitments with $\PVD$.

%Next, we rely on prior work~\cite{balancedOWF} which shows that 
%\dakshita{@james, I was wondering if you could fill this part in}
%$\frac{(1-\delta)}{(1+\delta)} \leq \frac{|S_0|}{|S_1|} \leq \frac{(1+\delta)}{\frac(1-\delta)}$ for some constant $\delta \in [0,1)$.

\paragraph{Public-Key Encryption with $\PVD$.}
Next, we note that the construction above {\em also} yields a public-key encryption scheme, as long as there is a trapdoor that allows recovery of the phase $b$ given the state 
\[ y, \sum_{x: f(x) = y, M(x) = 0} \ket{x} + (-1)^b \sum_{x: f(x) = y, M(x) = 1} \ket{x} \]
We call this property ``trapdoor phase-recoverability''. We show that this property is achievable from generic assumptions, even those that are not known to imply classical PKE.  
\begin{itemize}
    \item Specifically, trapdoor phase-recoverability is implied by a trapdoored variant of (almost) regular one-way functions, for which a trapdoor to the function allows recovery of a uniform superposition over all preimages of any given image $y$. This then allows efficient projection onto $\sum_{x: f(x) = y, M(x) = 0} \ket{x} + (-1)^b \sum_{x: f(x) = y, M(x) = 1} \ket{x}$ for any efficient $M$. 
    We also note that this property is satisfied by any (standard) trapdoored injective function.
    But it is also satisfied by functions such as the Ajtai function that are not necessarily injective. Indeed, it is unclear how to build classical public-key encryption, or even PKE with classical ciphertexts, given a general trapdoor phase-recoverable function. Nevertheless, we formalize the above ideas in \cref{sec:pke} and \cref{sec:almostreg} to build PKE schemes with quantum ciphertexts, that also support $\PVD$.
    \item Additionally, we show in \cref{sec:hmy} that a recent public-key encryption scheme of~\cite{HMY} from pseudorandom group actions also satisfies trapdoor phase-recoverability: in fact, the decryption algorithm in~\cite{HMY} relies on recovering the phase from a similar superposition, given a trapdoor.
\end{itemize}

\paragraph{Hybrid Encryption with PVD.} Finally, we observe that we can use any encryption scheme $\Enc$ to encrypt the trapdoor $\td$ associated with the above construction, and security will still hold. That is, if $\Enc$ is semantically-secure, then our techniques extend to show that a ciphertext of the form
\[ y, \sum_{x: f(x) = y, M(x) = 0} \ket{x} + (-1)^b \sum_{x: f(x) = y, M(x) = 1} \ket{x} , \Enc(\td)\]
where $\mathsf{td}$ is the trapdoor for $f$, still supports publicly-verifiable deletion of the bit $b$. Thus, our approach can be seen as a way to \emph{upgrade} cryptographic schemes $\Enc$ with special properties to satisfy $\PVD$. In particular, we prove in \cref{sec:generic} that instantiating $\mathsf{Enc}$ appropriately with attribute-based encryption, fully-homomorphic encryption, witness encryption, or timed-release encryption gives us the same scheme supporting $\PVD$.

\subsection{Discussion and Directions for Future Work}
Our work demonstrates a strong relationship between weak security properties of (trapdoored) one-way functions and publicly-verifiable deletion. In particular, previous work~\cite{Poremba22} conjectured that collapsing functions, which are a quantum strengthening of collision-resistant hashes, lead to cryptosystems with publicly-verifiable deletion. 
%But obtaining publicly-verifiable deletion from standards assumption like LWE remained an open problem. 
Besides proving this conjecture, we also show that collapsing/collision-resistance, which are considered stronger assumptions than one-wayness, are actually not necessary for PVD. 

Indeed, weakenings that we call target-collapsing and generalized-target-collision-resistance, can be obtained from (regular) variants of one-way functions, and do suffice for publicly-verifiable deletion.
Analogously to their classical counterparts, we believe that these primitives will be of independent interest. 
Indeed, a natural question that this work leaves open is whether variants of these primitives that suffice for publicly-verifiable deletion can be based on {\em one-way functions} without the regularity constraint. It is also interesting to further understand relationships and implications between target-collision-resistance and target-collapsing, including when these properties may or may not imply each other. It may also be useful to understand if these weaker properties can suffice in place of stronger properties such as collapsing and collision-resistance in other contexts, including the design of post-quantum protocols.

Finally, note that we rely on trapdoored variants of these primitives to build public-key encryption schemes. Here too, in addition to obtaining PKE with $\PVD$ from any injective trapdoor one-way function (TDF), it becomes possible to relax assumptions to only require (almost)-regularity and trapdoor phase-recoverability -- properties that can plausibly be achieved from weaker concrete assumptions than injective TDFs.
%, such as pseudorandom group actions~\cite{HMY}. 
These are new examples of complexity assumptions that yield public-key encryption with quantum ciphertexts, but may be too weak to obtain PKE with classical ciphertexts.
It is an interesting question to further investigate the weakest complexity assumptions that may imply public-key encryption, with or without $\PVD$.
%\dakshita{out of steam for now, @james if you have energy please feel free to add, o/w I will add in the morning}

\section*{Acknowledgements}
D.K. was supported in part by NSF CAREER CNS-2238718, NSF CNS-2247727 and DARPA SIEVE. This material is based upon work supported by
the Defense Advanced Research Projects Agency through Award HR00112020024.

A.P. is partially supported by AFOSR YIP (award number FA9550-16-1-0495), the Institute for Quantum Information and Matter (an NSF Physics Frontiers Center; NSF Grant PHY-1733907) and by a grant from the Simons 
Foundation (828076, TV).

\section{Preliminaries}

In this section, we review basic concepts from quantum computing and cryptography.

\subsection{Quantum Computing}

We refer to \cite{NielsenChuang11,Wilde13} for a comprehensive background on quantum computation. 

A finite-dimensional complex Hilbert space is denoted by $\mathcal{H}$, and we use subscripts to distinguish between different systems (or registers); for example, we let $\mathcal{H}_{\mathsf{A}}$ be the Hilbert space corresponding to a system $\mathsf{A}$. 
The tensor product of two Hilbert spaces $\algo H_{\mathsf{A}}$ and $\algo H_{\mathsf{B}}$ is another Hilbert space denoted by $\algo H_{\mathsf{AB}} = \algo H_{\mathsf{A}} \otimes \algo H_{\mathsf{B}}$.  We let $\algo L(\algo H)$
denote the set of linear operators over $\algo H$.
A quantum system over the $2$-dimensional Hilbert space $\mathcal{H} = \mathbb{C}^2$ is called a \emph{qubit}. For $n \in \mathbb{N}$, we refer to quantum registers over the Hilbert space $\mathcal{H} = \big(\mathbb{C}^2\big)^{\otimes n}$ as $n$-qubit states. We use the word \emph{quantum state} to refer to both pure states (unit vectors $\ket{\psi} \in \mathcal{H}$) and density matrices $\rho \in \mathcal{D}(\mathcal{H)}$, where we use the notation $\mathcal{D}(\mathcal{H)}$ to refer to the space of positive semidefinite linear operators of unit trace acting on $\algo H$. 
Occasionally, we consider \emph{subnormalized states}, i.e. states in the space of positive semidefinite operators over $\algo H$ with trace norm not exceeding $1$.

The \emph{trace distance} of two density matrices $\rho,\sigma \in \mathcal{D}(\mathcal{H)}$ is given by
$$
\TD(\rho,\sigma) = \frac{1}{2} \Tr\left[ \sqrt{ (\rho - \sigma)^\dag (\rho - \sigma)}\right].
$$

A quantum channel $\Phi:  \algo L(\algo H_{\mathsf{A}}) \rightarrow \algo L(\algo H_{\mathsf{B}})$ is a linear map between linear operators over the Hilbert spaces $\algo H_{\mathsf{A}}$ and $\algo H_{\mathsf{B}}$. 
We say that a channel $\Phi$ is \emph{completely positive} if, for a reference system $R$ of arbitrary size, the induced map $I_R \otimes \Phi$ is positive, and we call it \emph{trace-preserving} if $\Tr[\Phi(X)] = \Tr[X]$, for all $X\in \algo L(\algo H)$. A quantum channel that is both completely positive and trace-preserving is called a quantum $\CPTP$ channel. 

A polynomial-time \emph{uniform} quantum algorithm (or $\QPT$ algorithm) is a polynomial-time family of quantum circuits given by $\algo C = \{C_\lambda\}_{\lambda \in \N}$, where each circuit $C \in \algo C$ is described by a sequence of unitary gates and measurements; moreover, for each $\lambda \in \N$, there exists a deterministic polynomial-time Turing machine that, on input $1^\lambda$, outputs a circuit description of $C_\lambda$. Similarly, we also define (classical) probabilistic polynomial-time $(\PPT)$ algorithms. A quantum algorithm may, in general, receive (mixed) quantum states as inputs and produce (mixed) quantum states as outputs. Occasionally, we restrict $\QPT$ algorithms implicitly; for example, if we write $\Pr[\mathcal{A}(1^{\lambda}) = 1]$ for a $\QPT$ algorithm $\mathcal{A}$, it is implicit that $\mathcal{A}$ is a $\QPT$ algorithm that outputs a single classical bit.

\paragraph{Quantum Fourier transform.} Let $q \geq 2$ be a modulus and $n \in \N$ and let $\omega_q = e^{ \frac{2 \pi i}{q}} \in \mathbb{C}$ denote the primitive $q$-th root of unity.
The $m$-qudit \emph{$q$-ary quantum Fourier transform} over the ring $\Z_q^m$ is defined by the operation,
$$
\FT_q : \quad \ket{\vec x} \quad \mapsto \quad \sqrt{q^{-m}} \displaystyle\sum_{\vec y \in \Z_q^m} \omega_q^{\langle \vec x,\vec y\rangle} \ket{\vec y}, \quad\quad \forall \vec x \in \Z_q^m.
$$
The $q$-ary quantum Fourier transform is \emph{unitary} and can be efficiently implemented on a quantum computer for any integer modulus $q \geq 2$~\cite{892139}.

\paragraph{Pauli Twirling.}

We use the following unitary operators:
\begin{itemize}
    \item Pauli-$\mathsf{Z}$ operator:

    $$
\mathsf{Z}^z = \sum_{x \in \bit} (-1)^{x \cdot z} \proj{x}, \quad \text{ for } z \in \bit.
$$

\item Multi-qubit Pauli-$\mathsf{Z}$ operator:

    $$
\mathsf{Z}^z = \mathsf{Z}^{z_1} \otimes \dots \otimes \mathsf{Z}^{z_m}, \quad \text{ for } z \in \bit^m.
$$

\item Controlled-$\mathsf{Z}$ operator:

    $$
\mathsf{C}\mathsf{Z}^z = \sum_{c \in \bit} \proj{c} \otimes \mathsf{Z}^{c \cdot z}, \quad \text{ for } z \in \bit^m.
$$

\end{itemize}
Here, we use the notation $\mathsf{Z}^0 = I$ and $\mathsf{Z}^1 = \mathsf{Z}$, as well as $c \cdot z = (c \cdot z_1,\dots,c \cdot z_m)$ for $z \in \bit^m$.

We use the following well-known property of the Pauli-$\mathsf{Z}$ dephasing channel which says that, on average, a random Pauli-Z twirl induces a measurement in the computational basis.

\begin{lemma}[Pauli-$\mathsf{Z}$ Twirl]\label{lem:random-Z}
The Pauli-$\mathsf{Z}$ dephasing channel applied to an $m$-qubit state $\rho$ satsifies
$$
\algo Z(\rho) \,\overset{\mathrm{def}}{=} \,\ 2^{-m} \sum_{z \in \bit^m}  \mathsf{Z}^{z} \rho \left(\mathsf{Z}^{z}\right)^\dag= \sum_{x \in \bit^m} \Tr[\ketbra{x}{x} \rho] \,\ketbra{x}{x}.
$$
\end{lemma}

\subsection{Cryptography}

Throughout this work, wet $\lambda\in \N$ denote the security parameter. We assume that the reader is familiar with the fundamental cryptographic concepts. 

\paragraph{The Short Integer Solution problem.}
%\label{sec:sis}

The (inhomogenous) $\SIS$ problem was introduced by Ajtai~\cite{DBLP:conf/stoc/Ajtai96} in his seminal work on average-case lattice problems. The problem is defined as follows. 

\begin{definition}[Inhomogenous SIS problem,\cite{DBLP:conf/stoc/Ajtai96}]\label{def:ISIS} Let $n,m \in \N$ be integers, let $q\geq 2$ be a modulus and let $\beta >0$ be a parameter. The Inhomogenous Short Integer Solution problem $(\ISIS)$ problem is to find a short solution $\vec x \in \Z^m$ with $\|\vec x\| \leq \beta$ such that $\vec A \cdot \vec x = \vec y \Mod{q}$ given as input a tuple $(\vec A \rand \Z_q^{n \times m},\vec y \rand \Z_q^n)$.
The Short Integer Solution $(\SIS)$ problem is a homogenous variant of the $\ISIS$ problem with input $(\vec A \rand \Z_q^{n \times m},\vec 0 \in\Z_q^n)$.
\end{definition}

Micciancio and Regev~\cite{DBLP:journals/siamcomp/MicciancioR07} showed that the $\SIS$ problem is, on the average, as hard as approximating worst-case lattice problems to within small factors. Subsequently, Gentry, Peikert and Vaikuntanathan~\cite{cryptoeprint:2007:432} gave an improved reduction showing that, for parameters $m=\poly(n)$, $\beta=\poly(n)$ and prime $q \geq \beta \cdot \omega(\sqrt{n \log q})$, the average-case $\SIS_{n,q,\beta}^m$ problem is as hard as approximating the shortest independent vector problem $(\mathsf{SIVP})$ problem in the
worst case to within a factor $\gamma = \beta \cdot \tilde{O}(\sqrt{n})$.
We assume that $\SIS_{n,q,\beta}^m$, for $m=\Omega(n \log q)$, $\beta = 2^{o(n)}$ and $q=2^{o(n)}$, is hard against polynomial-time quantum adversaries.

\paragraph{The Learning with Errors problem.}
%\label{sec:lwe}

The \emph{Learning with Errors} problem serves as the primary basis of hardness of post-quantum cryptosystems and was introduced by Regev~\cite{Regev05}. The problem is defined as follows.

\begin{definition}[Learning with Errors problem, \cite{Regev05}]\label{def:decisional-lwe} Let $n,m \in \N$ be integers, let $q\geq 2$ be a modulus and let $\alpha \in (0,1)$ be a noise ratio parameter. The (decisional) Learning with Errors $(\LWE_{n,q,\alpha q}^m)$ problem is to distinguish between the following samples
$$
(\vec A \rand \Z_q^{n \times m},\vec s^\intercal \vec A+ \vec e^\intercal \Mod{q}) \quad \text{ and } \quad (\vec A \rand \Z_q^{n \times m},\vec u \rand \Z_q^m),\,\,
$$
where $\vec s \rand  \Z_q^n$ is a uniformly random vector and where $\vec e \sim D_{\Z^m,\sigma}$ is a discrete Gaussian error vector, where $D_{\Z^m,\sigma}$
assigns probability proportional to
$\rho_\sigma(\vec x) = \exp(-\pi \|\vec x \|^2/ \sigma^2)$ to each $\vec x \in \Z^m$, for $\sigma = \alpha q>0$.

We rely on the quantum $\LWE_{n,q,\alpha q}^m$ assumption which states that the samples above are computationally indistinguishable for any $\QPT$ algorithm.

\end{definition}

It was shown in~\cite{Regev05,cryptoeprint:2017/258} that the $\LWE_{n,q,\alpha q}^m$ problem with parameter $\alpha q \geq 2 \sqrt{n}$ is at least as hard as approximating the shortest independent vector problem $(\mathsf{SIVP})$ to within a factor of $\gamma = \widetilde{O}(n / \alpha)$ in worst case lattices of dimension $n$. In this work we assume the subexponential hardness of $\LWE_{n,q,\alpha q}^m$ which relies on the worst case hardness of approximating short vector problems in
lattices to within a subexponential factor. 
We assume that $\LWE_{n,q,\alpha q}^m$, for $m=\Omega(n \log q)$, $q=2^{o(n)}$, $\alpha=1/2^{o(n)}$,  is hard against polynomial-time quantum adversaries. 
%We note that this assumption implies  $\SIS_{n,q,\beta}^m$ for the parameters described in~\Cref{sec:sis}. 

\section{Main Theorem: Certified Everlasting Target-Collapsing}

%The notion of a \emph{collapsing} hash function was introduced by Unruh~\cite{cryptoeprint:2015/361} as a quantum strengthening of collision resistance. In the classical setting, 

%Collapsing can be considered a quantum analogue of classical collision-resistance. In the classical setting, a weaker security notion has also been studied, called \emph{target}-collision-resistance \cite{10.1145/73007.73011}, which requires that for any \emph{fixed} input $x$, no polynomial-time adversary can find a collision $x' \neq x$ such that $h(x) = h(x')$. We now formalize this definition, and then introduce a quantum analogue which we call \emph{target-collapsing}.

\subsection{Definitions}

In this section, we present our definitions of target-collapsing and (generalized) target-collision-resistance. We parameterize our definitions by a distribution $\cD$ over preimages and a measurement function $\cM$. Note that when $\cM$ is the identity function, the notion of $(\cD,\cM)$-target-collapsing corresponds to a notion where the entire preimage register is measured in the computational basis. In this case we drop parameterization by $\cM$ and just say $\cD$-target-collapsing. Also, when $\cD$ is the uniform distribution, we drop parameterization by $\cD$ and just say $\cM$-target-collapsing.

\begin{definition}[$(\cD,\cM)$-Target-Collapsing Hash Function]\label{def:target-collapsing}
Let $\lambda \in \N$ be the security parameter.
A hash function family given by $\cH = \{H_\secp : \{0,1\}^{m(\secp)} \to \{0,1\}^{n(\secp)}\}_{\secp \in \bbN}$ is $(\cD,\cM)$-target-collapsing for some distribution $\cD = \{D_\secp\}_{\secp \in \bbN}$ over $\{\{0,1\}^{m(\secp)}\}_{\secp \in \bbN}$ and family of functions $\cM = \{\{M[h] : \{0,1\}^{m(\secp)} \to \{0,1\}^{k(\secp)}\}_{h \in H_\secp}\}_{\secp \in \bbN}$ if, for every QPT adversary $\cA = \{\cA_\secp\}_{\secp \in \bbN}$,
$$
|
\Pr[ \mathsf{TargetCollapseExp}_{\algo H,\algo A,\algo D,\algo M,\lambda}(0)=1] - \Pr[ \mathsf{TargetCollapseExp}_{\algo H,\algo A,\algo D,\algo M,\lambda}(1)=1]
| \leq \negl(\lambda).
$$
Here, the experiment $\mathsf{TargetCollapseExp}_{\algo H,\algo A,\algo D,\algo M,\lambda}(b)$ is defined as follows:

\begin{enumerate}

    \item The challenger prepares the state \[\sum_{x \in \{0,1\}^{m(\secp)}}\sqrt{D_\secp(x)}\ket{x}\] on register $X$, and samples a random hash function $h \rand H_\lambda$. Then, it coherently computes $h$ on $X$ (into a fresh $n(\secp)$-qubit register $Y$) and measures system $Y$ in the computational basis, which results in an outcome $y \in \bit^{n(\lambda)}$.
    \item If $b=0$, the challenger does nothing. Else, if $b=1$, the challenger coherently computes $M[h]$ on $X$ (into a fresh $k(\secp)$-qubit register $V$) and measures system $V$ in the computational basis. Finally, the challenger sends the outcome state in system $X$ to $\algo A_\secp$, together with the string $y \in \bit^{n(\lambda)}$ and a description of the hash function $h$.
    \item $\algo A_\secp$ returns a bit $b'$, which we define as the output of the experiment.
\end{enumerate}
\end{definition}

We also define an analogous notion of $(\cD,\cM)$-target-collision-resistance, as follows. Similarly to above, we drop the parameterization by $\cM$ in the case that it is the identity function, and we drop the parameterization by $\cD$ in the case that it is the uniform distribution. Notice that target-collision-resistance (without parameterization) then coincides with the classical notion where a uniformly random input is sampled, and the adversary must find a collision with respect to this input (this is also sometimes called second-preimage resistance, or weak collision-resistance).

%\dakshita{check that this does not contradict classical definitions}\james{how does the above sound?}

\begin{definition}[$(\cD,\cM)$-Target-Collision-Resistant Hash Function]\label{def:target-CR}
A hash function family $\cH = \{H_\secp : \{0,1\}^{m(\secp)} \to \{0,1\}^{n(\secp)}\}_{\secp \in \bbN}$ is $(\cD,\cM)$-target-collision-resistant for some distribution $\cD = \{D_\secp\}_{\secp \in \bbN}$ over $\{\{0,1\}^{m(\secp)}\}_{\secp \in \bbN}$ and family of functions $\cM = \{\{M[h] : \{0,1\}^{m(\secp)} \to \{0,1\}^{k(\secp)}\}_{h \in H_\secp}\}_{\secp \in \bbN}$ if, for every QPT adversary $\cA = \{\cA_\secp\}_{\secp \in \bbN}$,
$$
|
\Pr[ \mathsf{TargetCollRes}_{\algo H,\algo A,\algo D,\algo M,\lambda}=1]| \leq \negl(\lambda).
$$
Here, the experiment $\mathsf{TargetCollRes}_{\algo H,\algo A,\algo D,\algo M,\lambda}$ is defined as follows:
\begin{enumerate}
    \item The challenger prepares the state \[\sum_{x \in \{0,1\}^{m(\secp)}}\sqrt{D_\secp(x)}\ket{x}\] on register $X$, and samples a random hash function $h \rand H_\lambda$. Next, it coherently computes $h$ on $X$ (into a fresh $n(\secp)$-qubit system $Y$) and measures system $Y$ in the computational basis, which results in an outcome $y \in \bit^{n(\lambda)}$. Next, it coherently computes $M[h]$ on $X$ (into a fresh $k(\secp)$-qubit register $V$) and measures system $V$ in the computational basis, which results in an outcome $v$. Finally, its sends the outcome state in system $X$ to $\algo A_\secp$, together with the string $y \in \{0,1\}^{n(\secp)}$ and a description of the hash function $h$.
    \item $\algo A_\secp$ responds with a string $x \in \{0,1\}^{m(\secp)}$.
    \item The experiment outputs 1 if $h(x) = y$ and $M[h](x) \neq v$.
\end{enumerate}

\end{definition}

Finally, we define the notion of a \emph{certified everlasting} target-collapsing hash.

\begin{definition}\label{def:ev-target-collapsing}
A hash function family $\cH = \{H_\secp : \{0,1\}^{m(\secp)} \to \{0,1\}^{n(\secp)}\}_{\secp \in \bbN}$ is certified everlasting $(\cD,\cM)$-target-collapsing for some distribution $\cD = \{D_\secp\}_{\secp \in \bbN}$ over $\{\{0,1\}^{m(\secp)}\}_{\secp \in \bbN}$ and family of functions $\cM = \{\{M[h] : \{0,1\}^{m(\secp)} \to \{0,1\}^{k(\secp)}\}_{h \in H_\secp}\}_{\secp \in \bbN}$ if for every two-part adversary $\algo A = \{\algo A_{0,\secp},\algo A_{1,\secp}\}_{\secp \in \bbN}$, where $\{\algo A_{0,\secp}\}_{\secp \in \bbN}$ is QPT and $\{\algo A_{1,\secp}\}_{\secp \in \bbN}$ is unbounded, it holds that 

%adversary $\algo A = \{(\algo A_{\secp,0},\algo A_{\secp,1})\}_{\secp \in \bbN}$ consisting of a  and an unbounded algorithm $\algo A_{\secp,1}$, and any bit $b' \in \{0,1\}$, it holds that
$$
|\Pr\left[\mathsf{EvTargetCollapseExp}_{\algo H,\algo A,\algo D,\algo M,\lambda}(0) = 1 \right] - \Pr\left[\mathsf{EvTargetCollapseExp}_{\algo H,\algo A,\algo D,\algo M,\lambda}(1) = 1\right]| \leq \negl(\lambda).
$$
Here, the experiment $\mathsf{EvTargetCollapseExp}_{\algo H,\algo A,\algo D,\algo M,\lambda}(b)$ is defined as follows:

\begin{enumerate}

    \item The challenger prepares the state \[\sum_{x \in \{0,1\}^{m(\secp)}}\sqrt{D_\secp(x)}\ket{x}\] on register $X$, and samples a random hash function $h \rand H_\lambda$. Then, it coherently computes $h$ on $X$ (into a fresh $n(\secp)$-qubit system $Y$) and measures system $Y$ in the computational basis, which results in an outcome $y \in \bit^{n(\lambda)}$.
    \item If $b=0$, the challenger does nothing. Else, if $b=1$, the challenger coherently computes $M[h]$ on $X$ (into an auxiliary $k(\secp)$-qubit system $V$) and measures system $V$ in the computational basis. Finally, the challenger sends the outcome state in system $X$ to $\algo A_{0,\secp}$, together with the string $y \in \bit^{n(\lambda)}$ and a description of the hash function $h$.
    \item $\algo A_{0,\secp}$ sends a classical certificate $\pi \in \bit^{m(\lambda)}$ to the challenger and initializes $\algo A_{1,\secp}$ with its residual state.
    \item The challenger checks if $h(\pi)=y$. If true, $\algo A_{1,\secp}$ is run until it outputs a bit $b'$. Otherwise, $b' \gets \{0,1\}$ is sampled uniformly at random. The output of the experiment is $b'$.

    %the challenger sends $\top$ to $\cA_{\secp,1}$ and the game continues; else, the game ends and the output of the experiment is $\bot$.
    %\item $\algo A_{\secp,1}$ outputs a bit $b'$, which we define as the output of the experiment.
    
\end{enumerate}
\end{definition}

\subsection{Main Theorem}
\label{sec:maintheorem}
Our main theorem is the following.

\begin{theorem}\label{thm:CETC-generalization}
Let $\cH = \{H_\secp\}_{\secp \in \bbN}$ be a hash function family that is both $(\cD,\cM)$-target-collapsing and $(\cD,\cM)$-target-collision-resistant, for some distribution $\cD$ and efficiently computable family of functions $\cM$. Then, $\cH$ is certified everlasting $(\cD,\cM)$-target-collapsing.
\end{theorem}

\begin{proof}

Throughout the proof, we will leave the security parameter implicit, defining $H \coloneqq H_\secp, D \coloneqq D_\secp, m \coloneqq m(\secp), n \coloneqq n(\secp)$, $k \coloneqq k(\secp)$, $\cA_0 \coloneqq \cA_{0,\secp}$, and $\cA_1 \coloneqq \cA_{1,\secp}$. Next, we define

\[\ket{\psi}_X \coloneqq \sum_{x \in \{0,1\}^m}\sqrt{D(x)}\ket{x}.\] For $h \in H, y \in \{0,1\}^m$, we define a unit vector \[\ket{\psi_{h,y}}_X \propto  \sum_{x \in \{0,1\}^m : h(x)=y}\sqrt{D(x)}\ket{x}.\] Finally, for $h \in H, y \in \{0,1\}^m, v \in \{0,1\}^k$ we define a unit vector \[\ket{\psi_{h,y,v}}_X \propto  \sum_{x \in \{0,1\}^m : h(x)=y, M[h](x)=v}\sqrt{D(x)}\ket{x}.\]

% ~~ \ket{\psi_{h,y}'}_{X,V} \propto \sum_{x \in \{0,1\}^m : h(x)=y}\sqrt{D(x)}\ket{x}\ket{F[h](x)},

We consider the following hybrids. 

\begin{itemize}
    \item $\mathsf{Exp}_0(b)$:
    \begin{enumerate}
    %\item The adversary sends an $m(\lambda)$-qubit quantum state $\rho$ in a system $X$ to the challenger.
    \item The challenger prepares $\ket{\psi}_X$, samples a random hash function $h \rand H_\lambda$, coherently computes $h$ on $X$ into a fresh $n$-qubit register $Y$, and measures $Y$ in the computational basis to obtain $y \in \bit^{n}$ and a left-over state $\ket{\psi_{h,y}}_X$.
    \item If $b=0$, the challenger does nothing. Else, if $b=1$, the challenger computes $M[h]$ on $X$ into a fresh $k$-qubit register $V$, and measures $V$ in the computational basis. Finally, the challenger sends the left-over state in system $X$ to $\algo A_0$, together with the string $y \in \bit^{n}$ and a classical description of $h$.

    \item $\algo A_0$ sends a classical certificate $\pi \in \bit^m$ to the challenger and initializes $\algo A_1$ with its residual state.

    \item The challenger checks if $h(\pi)=y$. If true, $\algo A_1$ is run until it outputs a bit $b'$. Otherwise, $b' \gets \{0,1\}$ is sampled uniformly at random. The output of the experiment is $b'$.

    %If true, the challenger sends $\top$ to $\cA_1$ and the game continues; else, the game ends and the output of the experiment is $\bot$.

    %\item $\algo A_1$ outputs a bit $b'$, which we define as the output of the experiment.
\end{enumerate}

    \item $\mathsf{Exp}_1(b)$:
    \begin{enumerate}
    \item The challenger prepares $\ket{\psi}_X$, samples a random hash function $h \rand H_\lambda$, coherently computes $h$ on $X$ into a fresh $n$-qubit register $Y$, and measures $Y$ in the computational basis to obtain $y \in \bit^{n}$ and a left-over state $\ket{\psi_{h,y}}_X$.
    \item The challenger computes $M[h]$ on $X$ into a fresh $k$-qubit register $V$ to obtain a state
    
    \[\propto \sum_{x \in \{0,1\}^m: h(x)=y} \sqrt{D(x)}\ket{x}_X\ket{M[h](x)}_V.\]
    
    Then, the challenger samples a random string $z \rand \bit^k$, prepares a $\ket{+}$ state in system $C$, and applies a controlled-$\mathsf{Z}^{z}$ operation from $C$ to $V$, which results in a state
    
    \begin{align*}
        &\propto \sum_{c \in \{0,1\}} \ket{c}_C \otimes \sum_{x \in \{0,1\}^m: h(x)=y} \sqrt{D(x)}\ket{x}_X\mathsf{Z}^{c \cdot z}\ket{M[h](x)}_V\\ &= \sum_{c \in \{0,1\}} \ket{c}_C \otimes \sum_{x \in \{0,1\}^m: h(x)=y} \sqrt{D(x)}(-1)^{c \cdot \langle M[h](x),z\rangle}\ket{x}_X\ket{M[h](x)}_V.
    \end{align*}
    
      %\[\propto \sum_{c \in \{0,1\}}\ket{c}_C \otimes \sum_{x \in \{0,1\}^m : h(x)=y}(-1)^{c \cdot \langle F[h](x),z\rangle}\ket{x}_X.\]

    Finally, the challenger uncomputes the $V$ register by again computing $M[h]$ from $X$ to $V$, and sends system $X$ to $\algo A_0$, together with $y \in \bit^n$ and a classical description of $h$.

    \item $\algo A_0$ sends a classical certificate $\pi \in \bit^{m}$ to the challenger and initializes $\algo A_1$ with its residual state.
    \item The challenger checks if $h(\pi)=y$. Then, the challenger measures system $C$ to obtain $c' \in \{0,1\}$ and checks that $c' = b$. If both checks are true, $\algo A_1$ is run until it outputs a bit $b'$. Otherwise, $b' \gets \{0,1\}$ is sampled uniformly at random. The output of the experiment is $b'$.

    %If false, the game ends and the output of the experiment is $\bot$. Then, the challenger measures system $C$ to obtain $c' \in \{0,1\}$. If $c'\neq b$, the game ends and the output of the experiment is $\bot$. Otherwise, if both checks passed, the output of the experiment is $\rho$.
    %\item $\algo A_1$ outputs a bit $b'$, which we define as the output of the experiment.
    \end{enumerate}

\item $\mathsf{Exp}_2(b)$:
    \begin{enumerate}
    \item The challenger prepares $\ket{\psi}_X$, samples a random hash function $h \rand H_\lambda$, coherently computes $h$ on $X$ into a fresh $n$-qubit register $Y$, and measures $Y$ in the computational basis to obtain $y \in \bit^{n}$ and a left-over state $\ket{\psi_{h,y}}_X$.
    \item The challenger computes $M[h]$ on $X$ into a fresh $k$-qubit register $V$. Then, the challenger samples a random string $z \rand \bit^k$, prepares a $\ket{+}$ state in system $C$, applies a controlled-$\mathsf{Z}^{z}$ operation from $C$ to $V$, and finally uncomputes the $V$ register by again computing $M[h]$ from $X$ to $V$. Note that this results in a state
    
    \[\propto \sum_{c \in \{0,1\}}\ket{c}_C \otimes \sum_{x \in \{0,1\}^m : h(x)=y}(-1)^{c \cdot \langle M[h](x),z\rangle}\ket{x}_X.\]
    
    Finally, it sends system $X$ to $\algo A_0$, together with $y \in \bit^n$ and a classical description of $h$.

    \item $\algo A_0$ sends a classical certificate $\pi \in \bit^{m}$ and initializes $\algo A_1$ with its residual state.

    \item The challenger checks if $h(\pi)=y$. Then, the challenger applies the following projective measurement to system $C$:
    \[\Big\{\proj{\phi_\pi^{ z}},I - \proj{\phi_\pi^{ z}}\Big\} \,\quad \text{ where } \quad \ket{\phi_\pi^{ z}} \coloneqq \frac{1}{\sqrt{2}}  \left( \ket{0} +  (-1)^{\langle M[h](\pi), z \rangle } \ket{1}\right),\] and checks that the first outcome is observed. Finally, the challenger measures system $C$ to obtain $c' \in \{0,1\}$ and checks that $c'=b$. If all three checks are true, $\algo A_1$ is run until it outputs a bit $b'$. Otherwise, $b' \gets \{0,1\}$ is sampled uniformly at random. The output of the experiment is $b'$.
    
    %Finally, the challenger measures system $C$ to obtain $c' \in \{0,1\}$. If $c'\neq b$, the game ends and the output of the experiment is $\bot$. Otherwise, if both checks passed, the challenger sends $\top$ to $\cA_1$ and the game continues.
    %\item $\algo A_1$ outputs a bit $b'$, which we define as the output of the experiment.
    \end{enumerate}

Finally, we also use the following hybrid which is convenient for the sake of the proof.

\item $\mathsf{Exp}_3(b)$:
     \begin{enumerate}
     \item The challenger prepares $\ket{\psi}_X$, samples a random hash function $h \rand H_\lambda$, coherently computes $h$ on $X$ into a fresh $n$-qubit register $Y$, and measures $Y$ in the computational basis to obtain $y \in \bit^{n}$ and a left-over state $\ket{\psi_{h,y}}_X$.
     \item The challenger computes $M[h]$ on $X$ into a fresh $k$-qubit register $V$. Then, the challenger measures $V$ in the computational basis to obtain $v \in \{0,1\}^k$. Next, the challenger samples a random string $z \rand \bit^k$, prepares a $\ket{+}$ state in system $C$, applies a controlled-$\mathsf{Z}^{z}$ operation from $C$ to $V$, and finally uncomputes the $V$ register by again computing $M[h]$ from $X$ to $V$. Note that this results in the state
    
    \[\frac{1}{\sqrt{2}}\left(\ket{0}_C + (-1)^{\langle v,z\rangle}\ket{1}_C\right) \otimes \ket{\psi_{h,y,v}}_X.\]
    
    Finally, the challenger sends system $X$ to $\algo A_0$, together with $y \in \bit^n$ and a classical description of $h$.
    
    \item $\algo A_0$ sends a classical certificate $\pi \in \bit^{m}$ to the challenger and initializes $\algo A_1$ with its residual state.

    \item The challenger checks if $h(\pi)=y$. Then, the challenger applies the following projective measurement to system $C$:
    \[\Big\{\proj{\phi_\pi^{ z}},I - \proj{\phi_\pi^{ z}}\Big\} \,\quad \text{ where } \quad \ket{\phi_\pi^{ z}} \coloneqq \frac{1}{\sqrt{2}}  \left( \ket{0} +  (-1)^{\langle M[h](\pi), z \rangle } \ket{1}\right),\] and checks that the first outcome is observed. Finally, the challenger measures system $C$ to obtain $c' \in \{0,1\}$ and checks that $c'=b$. If all three checks are true, $\algo A_1$ is run until it outputs a bit $b'$. Otherwise, $b' \gets \{0,1\}$ is sampled uniformly at random. The output of the experiment is $b'$.
    \end{enumerate}
\end{itemize}

Before we analyze the probability of distinguishing between the consecutive hybrids, we first show that the following statements hold for the final experiment $\mathsf{Exp}_3(b)$.

\begin{claim}\label{claim:identical-certificate}
The probability that the challenger accepts the deletion certificate $\pi$ in Step 4 of $\mathsf{Exp}_3(b)$ and $M[h](\pi) \neq v$ is negligible. That is,

\[\Pr_{h,y,v} \left[
 h(\pi) = y
\,\,\, \wedge \,\,\,
M[h](\pi) \,\neq\, v
 \,\, : \,\, \pi \gets \algo A_0(h,y,\ket{\psi_{h,y,v}})\right] \leq \negl(\lambda),\] where the probability is over the challenger preparing $\ket{\psi}$, sampling $h$, and measuring $y$ and $v$ as described in $\Exp_3(b)$ to produce the left-over state $\ket{\psi_{h,y,v}}$.

%where $y = h(x_0) \in \bit^{n}$ is the image and where $\sigma_X^{z}$ is the reduced state %\james{Isn't the reduced state just $\ket{\vec x_0}$? If so, we wouldn't have to sample $\vec z$ in the above game, and we can just give $\cA_0$ the vector $\vec x_0$ (which directly corresponds to targeted collision-resistance)} 
%with respect to
%$$
%\sigma_{CX} = \frac{1}{2} \sum_{c,c' \in \bit} \ketbra{c}{c'}_C \otimes \mathsf{Z}^{c \cdot z}{\proj{x_0}}_X \left(\mathsf{Z}^{c' \cdot z}\right)^\dag.
%$$
\end{claim}

\begin{proof}
This follows directly from the assumed $(\cD,\cM)$-target-collision resistance of $\cH$, since the above probability is exactly $\Pr[\mathsf{TargetCollRes}_{\algo H,\algo A,\algo D,\algo M,\lambda}=1]$.
\end{proof}

%\begin{proof}
%Suppose for the sake of contraction that the probability is at least $1/\poly(\lambda)$. We now show that we can use $\algo A_0$ to break the target-collision-resistance of the hash family $\algo H = \{H_\lambda\}_{\lambda \in \N}$. 

%Our reduction proceeds as follows:
%\begin{enumerate}
 %\item Run $\algo A_0$ to obtain an $m(\lambda)$-qubit quantum state $\rho$ in system $X$.

%\item Measure system $X$ and send the outcome $x_0$ to the challenger.
    
%\item Once the challenger replies with a description of a hash function $h \rand H_\lambda$, send the register $\ket{x_0}$, the image $y = h(x_0)$ as well as a description of $h$ to $\algo A_0$.

%\item When $\algo A_0$ outputs $(\pi, \rho_{\aux})$, discard $\rho_{\aux}$ and output $(x_0,\pi)$.
%\end{enumerate}
%Notice that the state $\ket{x_0}$ which is sent to $\algo A_0$ is identical to the reduced state $\sigma_X$ with respect to $\sigma_{CX}$.
%By assumption, $\algo A_0$ outputs a valid certificate $\pi \neq x_0$ such that $h(\pi) = h(x_0)$ with probability at least $1/\poly(\lambda)$. Thus, we have broken the target-collision-resistance of $\algo H$.
%\end{proof}

\begin{claim}\label{claim:measurement-succeeds-wp-1}
The probability that the challenger accepts the deletion certificate $\pi$ in Step $4$ of $\mathsf{Exp}_3(b)$ and
the subsequent projective measurement on system $C$ fails (returns the second outcome) is negligible.
\end{claim}
\begin{proof}
This follows directly from \Cref{claim:identical-certificate}, which implies that except with negligible probability, the register $C$ is in the state
\[\frac{1}{\sqrt{2}}\left(\ket{0} + (-1)^{\langle v,z\rangle}\ket{1}\right)\] at the time the challenger applies the projective measurement.

%which implies that in the case the certificate $\pi$ returned by the adversary in
%$\mathsf{Exp}_3$ is identical to the pre-image produced by the challenger with all but negligible probability.
%Therefore, the projective measurement must also succeed with overwhelming probability.
\end{proof}

For any experiment $\Exp_i(b)$, we define the advantage \[\mathsf{Adv}(\mathsf{\Exp}_i) \coloneqq |\Pr\left[\Exp_i(0) = 1\right]-\Pr\left[\Exp_i(1)=1\right]|.\]

\begin{claim}
$$
\mathsf{Adv}(\mathsf{Exp}_2) = 0.
$$
\end{claim}
\begin{proof}
First note that in the case that the challenger rejects because either the deletion certificate is invalid or their projection fails, the experiment does not involve $b$, and thus the advantage of the adversary is $0$. Second, in the case that the challenger's projection succeeds, the register $C$ is either in the state
$$
\frac{1}{\sqrt{2}}  ( \ket{0} +  (-1)^{\langle \pi, z \rangle } \ket{1}) \quad\,\, \text{ or } \quad\,\,  \frac{1}{\sqrt{2}}  ( \ket{0} -  (-1)^{\langle \pi, z \rangle } \ket{1}) 
$$
for some $z \in \bit^k$, and thereby completely
unentangled from the rest of the system. Notice that the challenger's measurement of system $C$ with outcome $c'$ results in a uniformly random bit, which completely masks $b$. Therefore, the experiment is also independent of $b$ in this case, and thus the adversary's overall advantage in $\mathsf{Exp}_2$ is $0$. 
\end{proof}
Next, we argue the following.

\begin{claim}
$$
|\mathsf{Adv}(\mathsf{Exp}_2) - \mathsf{Adv}(\mathsf{Exp}_1) | \,\leq \, \negl(\lambda).
$$    
\end{claim}
\begin{proof}
Recall that \Cref{claim:measurement-succeeds-wp-1} shows that the projective measurement performed by the challenger in Step $4$ of $\mathsf{Exp}_3$ succeeds with overwhelming probability. We now argue that the same is also true in $\mathsf{Exp}_2$. Suppose for the sake of contradiction that there is a non-negligible difference between the success probabilities of the measurement. 
We now show that this implies the existence of an efficient distinguisher $\algo A'$ that breaks the $(\cD,\cM)$-target-collapsing property of the hash family $\algo H = \{H_\lambda\}_{\lambda \in \N}$. 

$\algo A'$ receives $(y,h)$ and a state on register $X$ from its challenger. Next, it computes $M[h]$ on $X$ into a fresh $k$-qubit register $V$, samples a random string $z \rand \{0,1\}^k$, prepares a $\ket{+}$ state in system $C$, applies a controlled-$\mathsf{Z}^z$ operation from $C$ to $V$, and then uncomputes register $V$ by again applying $M[h]$ from $X$ to $V$. Then, it runs $\algo A$ on $(y,h,X)$, which outputs a certificate $\pi$.

%Our reduction proceeds as follows:
%The distinguisher $\algo A'$ runs $\algo A_0$ to obtain an $m$-qubit state $\rho$ in system $X$,
%and forwards it to the challenger who responds
%with a description of a hash function $h \rand H_\lambda$, an image $y \in \bit^m$ and a state $\rho_y$ in system $X$
%which is either a partially measured state (consisting of a superposition of pre-images)
%or a single measured pre-image $\proj{x_0}$ such that $x_0 \in \bit^m$ and $h(x_0)=y$.
%Next, $\algo D$ samples a random string $z \rand \bit^m$ and runs $\algo A_0$ given as input system $X$ of the state
%$$
%\sigma_{CX} = \frac{1}{2} \sum_{c,c' \in \bit} \ketbra{c}{c'}_C \otimes \mathsf{Z}^{c \cdot z}{\rho_y}_X \left(\mathsf{Z}^{c' \cdot z}\right)^\dag.
%$$

Finally, $\algo A'$ applies the following projective measurement to system $C$:
$$
\Big\{\proj{\phi_\pi^{z}},I - \proj{\phi_\pi^{z}}\Big\} \,\quad \text{ where } \quad
\ket{\phi_\pi^{z}} \coloneqq \frac{1}{\sqrt{2}}  \left( \ket{0} +  (-1)^{\langle \pi, z \rangle } \ket{1}\right),
$$
and outputs $1$ if the measurement succeeds and $0$ otherwise.
If there is a non-negligible difference in success probabilities of this measurement between $\Exp_3(b)$ and $\Exp_2(b)$ (for any $b \in \{0,1\}$), then $\cA'$ breaks $(\cD,\cM)$-target-collapsing of $\cH$.

%between the case when $\rho_y$ is a superposition of pre-images, or $\rho_y = \proj{x_0}$ is a single pre-image of $y$, this immediately breaks the targeted-collapsing property of $\algo H$. Therefore, the projective measurement in Step $5$ of $\mathsf{Exp}_2$ must also succeed with overwhelming probability.

Now, recall that $\mathsf{Exp}_2(b)$ is identical to $\mathsf{Exp}_1(b)$, except that the challenger applies an additional a measurement in Step 4. Because the measurement succeeds with overwhelming probability, it follows from Gentle Measurement that the advantage of the adversary must remain the same up to a negligible amount. This proves the claim.
\end{proof}

%Suppose that $\algo A = (\algo A_0, \algo A_1)$ wins at $\mathsf{Exp}_0$ with probability $\epsilon >0$. We now show the following:

%\james{need to update this claim still}

\begin{claim}\label{claim:exp_0-exp_1}
\[\mathsf{Adv}(\Exp_1) = \mathsf{Adv}(\Exp_0)/2.\]   
\end{claim}
\begin{proof}

First note that in $\Exp_1(b)$, we can imagine measuring register $C$ to obtain $c'$ and aborting if $c' \neq b$ before the challenger sends any information to the adversary. This follows because register $C$ is disjoint from the adversary's registers. Next, by \Cref{lem:random-Z}, we have the following guarantees about the state on system $X$ given to the adversary in $\Exp_1(b)$. 

\begin{itemize}
    \item In the case $c' = b = 0$, the reduced state on register $X$ is $\ket{\psi_{h,y}}$.
    \item In the case that $c' = b = 1$, the reduced state on register $X$ is a mixture over $\ket{\psi_{h,y,v}}$ where $v$ is the result of measuring register $V$ in the computational basis.
\end{itemize}

Thus, this experiment is identical to $\Exp_0(b)$, except that we decide to abort and output a uniformly random bit $b'$ with probability 1/2 at the beginning of the experiment.

\end{proof}

Putting everything together, we have that $\mathsf{Adv}(\Exp_0) \leq \negl(\secp)$, which completes the proof.

\end{proof}

\subsection{Auxiliary Information}
\label{sec:mainaux}

Next, we generalize the above theorem statement to handle hash functions that are sampled with some auxiliary information. That is, there is an algorithm $(h,\aux) \gets \Samp(1^\secp)$ that samples the description of a hash function $h$ along with some auxiliary information $\aux$. We will want to allow the adversary to potentially see information about $\aux$ (but not necessarily all of it), so we define a family $\cZ = \{Z_\secp(\aux)\}_{\secp \in \bbN}$ that specifies what information the adversary sees about $\aux$. In the most straightforward case, $\cZ$ could be some distribution over classical or quantum states, parameterized by $\aux$. However, we also consider an \emph{interactive} $Z_\secp(\aux)$. That is, $Z_\secp$ is the description of an interactive machine that is initialized with $\aux$ and interacts with the adversary $\cA_\secp$.

\begin{definition}
A hash function family $\cH = \{H_\secp : \{0,1\}^{m(\secp)} \to \{0,1\}^{n(\secp)}\}_{\secp \in \bbN}$ with an associated sampling algorithm $\Samp$ is $(\cD,\cM,\cZ)$-target-collapsing for some distribution $\cD = \{D_\secp\}_{\secp \in \bbN}$ over $\{\{0,1\}^{m(\secp)}\}_{\secp \in \bbN}$, family of functions $\cM = \{\{M[h] : \{0,1\}^{m(\secp)} \to \{0,1\}^{k(\secp)}\}_{h \in H_\secp}\}_{\secp \in \bbN}$, and family of (static or interactive) distributions $\cZ = \{Z_\secp(\aux)\}_{(\cdot,\aux) \in \Samp(1^\secp), \secp \in \bbN}$ if, for every QPT adversary $\cA = \{\cA_\secp\}_{\secp \in \bbN}$,
$$
|
\Pr[ \mathsf{TargetCollapseExp}_{\algo H,\algo A,\algo D,\algo M,\algo Z,\lambda}(0)=1] - \Pr[ \mathsf{TargetCollapseExp}_{\algo H,\algo A,\algo D,\algo M,\algo Z,\lambda}(1)=1]
| \leq \negl(\lambda),
$$
where the experiment $\mathsf{TargetCollapseExp}_{\algo H,\algo A,\algo D,\algo M,\algo Z,\lambda}(b)$ is defined as in \cref{def:target-collapsing} except that $h$ is sampled by $(h,\aux) \gets \Samp(1^\secp)$, and the adversary is given (or interacts with) $Z_\secp(\aux)$ along with $(X,y,h)$.

\end{definition}

\begin{definition}
A hash function family $\cH = \{H_\secp : \{0,1\}^{m(\secp)} \to \{0,1\}^{n(\secp)}\}_{\secp \in \bbN}$ with an associated sampling algorithm $\Samp$ is $(\cD,\cM,\cZ)$-target-collision-resistant for some distribution $\cD = \{D_\secp\}_{\secp \in \bbN}$ over $\{\{0,1\}^{m(\secp)}\}_{\secp \in \bbN}$, family of functions $\cM = \{\{M[h] : \{0,1\}^{m(\secp)} \to \{0,1\}^{k(\secp)}\}_{h \in H_\secp}\}_{\secp \in \bbN}$, and family of (static or interactive) distributions $\cZ = \{Z_\secp(\aux)\}_{(\cdot,\aux) \in \Samp(1^\secp), \secp \in \bbN}$ if, for every QPT adversary $\cA = \{\cA_\secp\}_{\secp \in \bbN}$,
$$
\Pr[ \mathsf{TargetCollRes}_{\algo H,\algo A,\algo D,\algo M,\algo Z,\lambda}(0)=1] \leq \negl(\secp),
$$
where the experiment $\mathsf{TargetCollRes}_{\algo H,\algo A,\algo D,\algo M,\algo Z,\lambda}(b)$ is defined as in \cref{def:target-CR} except that $h$ is sampled by $(h,\aux) \gets \Samp(1^\secp)$, and the adversary is given (or interacts with) $Z_\secp(\aux)$ along with $(X,y,h)$.

\end{definition}

\begin{definition}
A hash function family $\cH = \{H_\secp : \{0,1\}^{m(\secp)} \to \{0,1\}^{n(\secp)}\}_{\secp \in \bbN}$ with an associated sampling algorithm $\Samp$ is certified everlasting $(\cD,\cM,\cZ)$-target-collapsing for some distribution $\cD = \{D_\secp\}_{\secp \in \bbN}$ over $\{\{0,1\}^{m(\secp)}\}_{\secp \in \bbN}$, family of functions $\cM = \{\{M[h] : \{0,1\}^{m(\secp)} \to \{0,1\}^{k(\secp)}\}_{h \in H_\secp}\}_{\secp \in \bbN}$, and family of (static or interactive) distributions $\cZ = \{Z_\secp(\aux)\}_{(\cdot,\aux) \in \Samp(1^\secp), \secp \in \bbN}$ if, for every two-part adversary $\algo A = \{\algo A_{0,\secp},\algo A_{1,\secp}\}_{\secp \in \bbN}$, where $\{\algo A_{0,\secp}\}_{\secp \in \bbN}$ is QPT and $\{\algo A_{1,\secp}\}_{\secp \in \bbN}$ is unbounded, it holds that 
$$
|\Pr[ \mathsf{EvTargetCollapseExp}_{\algo H,\algo A,\algo D,\algo M,\algo Z,\lambda}(0)=1]-\Pr[ \mathsf{EvTargetCollapseExp}_{\algo H,\algo A,\algo D,\algo M,\algo Z,\lambda}(1)=1]| \leq \negl(\secp),
$$
where the experiment $\mathsf{EvTargetCollapseExp}_{\algo H,\algo A,\algo D,\algo M,\algo Z,\lambda}(b)$ is defined as in \cref{def:target-CR} except that $h$ is sampled by $(h,\aux) \gets \Samp(1^\secp)$, and the first part of the adversary $\algo A_{0,\secp}$ is given (or interacts with) $Z_\secp(\aux)$ along with $(X,y,h)$.

\end{definition}

Now, the following generalization of \cref{thm:CETC-generalization} follows immediately from the proof of \cref{thm:CETC-generalization}, by additionally giving $Z_\secp(\aux)$ to the adversary in each of the experiments.

\begin{theorem}
Let $\cH = \{H_\secp\}_{\secp \in \bbN}$ be a hash function family that is both $(\cD,\cM,\cZ)$-target-collapsing and $(\cD,\cM,\cZ)$-target-collision-resistant, for some distribution $\cD$, efficiently computable family of functions $\cM$, and (static or interactive) distribution $\cZ$. Then, $\cH$ is certified everlasting $(\cD,\cM,\cZ)$-target-collapsing.
\end{theorem}

\subsection{Target-Collision-Resistance implies Target-Collapsing for Polynomial-Outcome Measurements}
\label{sec:tcr-implies}

In this section, we show that recent techniques from the collapsing hash function / collapsing commitment literature \cite{cryptoeprint:2022/786,crypto-2022-32202,crypto-2022-32124} imply that when $\cM$ is a function with polynomial number of outcomes, then $(\cD,\cM,\cZ)$-target-collision-resistance implies $(\cD,\cM,\cZ)$-target-collapsing. In this paper, we will only need to use this claim for \emph{two-outcome} measurements, but we show it for the more general case of polynomial-outcome measurements.

\begin{lemma}\label{thm:TC-from-TCR}
Let $\cH = \{H_\secp : \{0,1\}^{m(\secp)}\to \{0,1\}^{n(\secp)}\}_{\secp \in \bbN}$ be a hash function family that is $(\cD,\cM,\cZ)$-target-collision-resistant for some distribution $\cD = \{D_\secp\}_{\secp \in \bbN}$ over $\{\{0,1\}^{m(\secp)}\}_{\secp \in \bbN}$, family of functions $\cM = \{\{M[h] : \{0,1\}^{m(\secp)} \to \{0,1\}^{k(\secp)}\}_{h \in H_\secp}\}_{\secp \in \bbN}$ for $k(\secp)=O(\log \secp)$, and family of (static or interactive) distributions $\cZ = \{Z_\secp(\aux)\}_{(\cdot,\aux) \in \Samp(1^\secp), \secp \in \bbN}$. Then, $\cH$ is $(\cD,\cM,\cZ)$-target-collapsing.
\end{lemma}

\begin{proof}
We will make use of the following fact \cite[Claim 3.5]{cryptoeprint:2022/786}.

\begin{fact}\label{fact:distinguish-map}
Let $D$ be a projector, $\{\Pi_i\}_{i \in [N]}$ be pairwise orthogonal projectors, and $\ket{\psi}$ be any state such that $\ket{\psi} \in \mathsf{im}(\sum_{i \in [N]}\Pi_i)$. Then,
    \[\sum_{i \in [N]}\bigg\| \left(\sum_{j \neq i}\Pi_j\right)D\Pi_i\ket{\psi}\bigg\|^2  \geq \frac{1}{N}\left(\|D\ket{\psi}\|^2-\left(\sum_{i \in [N]}\|D\Pi_i\ket{\psi}\|^2\right)\right)^2.\]
\end{fact}

Now, suppose there exists an adversary $\{\cA_\secp\}_{\secp \in \bbN}$ that breaks the $(\cD,\cM,\cZ)$-target-collapsing of $\cH$. Dropping parameterization by $\secp$ for convenience, we can write such an adversary as a binary outcome projective measurement $(D,I-D)$ applied to a state received from the challenger. For any $h \in H_\secp, y \in \{0,1\}^n$, let $\ket{\psi_{h,y}}$ be the normalized state such that \[\ket{\psi_{h,y}} \propto \ket{h,y}\otimes\sum_{x \in \{0,1\}^m: h(x)=y}\sqrt{D(x)}\ket{x},\] and for $i \in \{0,1\}^k$, let \[\Pi_{i,h} \coloneqq \sum_{x \in \{0,1\}^m : M[h](x) = i}\dyad{x}{x}.\]

Then, the adversary's advantage in the $(\cD,\cM,\cZ)$-target-collapsing game can be written as 

\[\E_{h,y}\left[\|D\ket{\psi_{h,y}}\|^2 - \sum_{i \in \{0,1\}^k}\|D\Pi_{i,h}\ket{\psi_{h,y}}\|^2 \right] = \nonnegl(\secp),\] where the expectation is over the sampling of $h \gets H_\secp$ and the challenger's measurement of $y$.

Thus, by \cref{fact:distinguish-map}, it follows that

\[\E_{h,y}\left[\sum_{i \in \{0,1\}^k}\bigg\|\left(\sum_{j \neq i}\Pi_{j,h}\right)D\Pi_{i,h}\ket{\psi_{h,y}}\bigg\|^2\right] = \nonnegl(\secp),\] since $2^k = 2^{O(\log \secp)} = \poly(\secp)$. This completes the proof, as this expression exactly corresponds to the adversary's probability of winning the $(\cD,\cM,\cZ)$-target-collision-resistance game by applying $D$ and then measuring in the computational basis.

\end{proof}

%\ifsubmission
%\paragraph{Conclusion.}
%We refer the reader to Section \ref{sec:overreq} for an overview of how target-collapsing implies publicly-verifiable deletion.
%Due to lack of space, we defer the formal proof of this fact, and also our instantiations of target-collapsing functions from LWE/SIS following~\cite{Poremba22}, from almost-regular one-way functions as well as from pseudorandom group actions~\cite{HMY} to the appendices. 
%\else \fi

\section{Publicly-Verifiable Deletion from Dual-Regev Encryption}\label{sec:Dual-Regev}

In this section, we recall the constructions of Dual-Regev public-key encryption as well as fully homomorphic encryption with publicly-verifiable deletion introduced by Poremba~\cite{Poremba22}. Using our main result on certified-everlasting target-collapsing hashes in \Cref{thm:CETC-generalization}, we prove the \emph{strong Gaussian-collapsing conjecture} in~\cite{Poremba22}, and then conclude that the aforementioned constructions achieve certified-everlasting security assuming the quantum hardness of $\LWE$ and $\SIS$.

First, let us recall the definition of public-key encryption with publicly-verifiable deletion. 

\subsection{Definition: Encryption with Publicly-Verifiable Deletion}

A public-key encryption (PKE) scheme with publicly-verifiable deletion (PVD) has the following syntax.

\begin{itemize}
    \item $\KeyGen(1^\secp) \to (\pk,\sk)$: the key generation algorithm takes as input the security parameter $\secp$ and outputs a public key $\pk$ and secret key $\sk$.
    \item $\Enc(\pk,m) \to (\vk,\ket{\ct})$: the encryption algorithm takes as input the public key $\pk$ and a plaintext $m$, and outputs a (public) verification key $\vk$ and a ciphertext $\ket{\ct}$.
    \item $\Dec(\sk,\ket{\ct}) \to m$: the decryption algorithm takes as input the secret key $\sk$ and a ciphertext $\ket{\ct}$ and outputs a plaintext $m$.
    \item $\Del(\ket{\ct}) \to \pi$: the deletion algorithm takes as input a ciphertext $\ket{\ct}$ and outputs a deletion certificate $\pi$.
    \item $\Vrfy(\vk,\pi) \to \{\top,\bot\}$: the verify algorithm takes as input a (public) verification key $\vk$ and a proof $\pi$, and outputs $\top$ or $\bot$.
\end{itemize}

\begin{definition}[Correctness of deletion]\label{def:correctness-deletion}
A PKE scheme with PVD satisfies \emph{correctness of deletion} if for any $m$, it holds with $1-\negl(\secp)$ probability over $(\pk,\sk) \gets \Gen(1^\secp), (\vk,\ket{\ct}) \gets \Enc(\pk,m),\pi \gets \Del(\ket{\ct}),\mu \gets \Vrfy(\vk,\pi)$ that $\mu = \top$.
\end{definition}

\begin{definition}[Certified deletion security]\label{def:security-deletion}
A PKE scheme with PVD satisfies \emph{certified deletion security} if it satisfies standard semantic security, and moreover, for any QPT adversary $\{\cA_\secp\}_{\secp \in \bbN}$, it holds that 
\[\TD\left(\mathsf{EvPKE}_{\cA,\secp}(0),\mathsf{EvPKE}_{\cA,\secp}(1)\right) = \negl(\secp),\] where the experiment $\mathsf{EvPKE}_{\cA,\secp}(b)$ is defined as follows.
\begin{itemize}
    \item Sample $(\pk,\sk) \gets \Gen(1^\secp)$ and $(\vk,\ket{\ct}) \gets \Enc(\pk,b)$.
    \item Run $\cA_\secp(\pk,\vk,\ket{\ct})$, and parse their output as a deletion certificate $\pi$ and a left-over quantum state $\rho$.
    \item If $\Vrfy(\vk,\pi) = \top$, output $\rho$, and otherwise output $\bot$.
\end{itemize}
\end{definition}

\ \\
Before we introduce the
Dual-Regev public-key schemes proposed by Poremba~\cite{Poremba22}, let us first recall some basic facts about
Gaussian superpositions.

\subsection{Gaussian 
Superpositions}

Let $m \in \N$. The \emph{Gaussian measure} $\rho_\sigma$ with parameter $\sigma > 0$ is defined as
\begin{align*}
\rho_\sigma(\vec x) = \exp(-\pi \|\vec x \|^2/ \sigma^2), \quad \,\, \forall \vec x \in \mathbb{R}^m.   
\end{align*}
Given a modulus $q \in \N$ and $\sigma \in (\sqrt{2m},q/\sqrt{2m})$, the \emph{truncated} discrete Gaussian distribution $D_{\Z_q^m,\sigma}$ over the finite set $\Z^m \cap (-\frac{q}{2},\frac{q}{2}]^m$ with support $\{\vec x \in \Z_q^m : \|\vec x\| \leq \sigma \sqrt{m}\}$ is defined as
$$
D_{\Z_q^m,\sigma}(\vec x) = \frac{\rho_\sigma(\vec x)}{\displaystyle\sum_{\vec z \in \Z_q^m,\|\vec z\| \leq \sigma\sqrt{m} } \rho_\sigma(\vec z)}.
$$
In this section, we consider
Gaussian superposition states over $\Z^m \cap (-\frac{q}{2},\frac{q}{2}]^m$ of the form
    $$
 \ket{\psi} =    \sum_{\vec x \in \Z_q^m} \rho_\sigma(\vec x) \ket{\vec x}.
    $$
The state $\ket{\psi}$ is not normalized for convenience. A standard tail bound~\cite[Lemma 1.5 (ii)]{Banaszczyk1993} implies that (the normalized variant of) $\ket{\psi}$ is within negligible trace distance of a \emph{truncated} discrete Gaussian superposition $ \ket{\tilde{\psi}}$
with support $\{\vec x \in \Z_q^m : \|\vec x\| \leq \sigma \sqrt{\frac{m}{2}}\}$, where
$$
\ket{\tilde{\psi}}
= 
\left(\sum_{\vec z \in \Z_q^m,\|\vec z\| \leq \sigma \sqrt{\frac{m}{2}} } \rho_{\frac{\sigma}{\sqrt{2}}}(\vec z) \right)^{-\frac{1}{2}}\sum_{\vec x \in \Z_q^m : \|\vec x\| \leq \sigma \sqrt{\frac{m}{2}}}
\rho_\sigma(\vec x) \ket{\vec x}.
$$
Note that a measurement of $\ket{\tilde{\psi}}$ results in a sample from the truncated discrete Gaussian distribution $D_{\Z_q^m,\frac{\sigma}{\sqrt{2}}}$.
We remark that Gaussian superpositions with parameter $\sigma = \Omega(\sqrt{m})$ can be efficiently implemented using standard quantum state preparation techniques; for example using \emph{quantum rejection sampling} and the \emph{Grover-Rudolph algorithm}~\cite{Grover2002CreatingST,Regev05,Brakerski18,brakerski2021cryptographic}.

Let $\vec A \in \Z_q^{n \times m}$. We use the following algorithm, denoted by $\mathsf{GenGauss}(\vec A,\sigma)$ which prepares a partially measured Gaussian superposition of pre-images of a randomly generated image.
\begin{enumerate}
\item Prepare a Gaussian superposition in system $X$ with parameter $\sigma > 0$:
    $$
 \ket{\psi} =    \sum_{\vec x \in \Z_q^m} \rho_\sigma(\vec x) \ket{\vec x} \otimes \ket{\vec 0}.
    $$
\item Apply the unitary $U_{\vec A}: \ket{\vec x}\ket{\vec 0} \rightarrow \ket{\vec x} \ket{\vec A \cdot \vec x \Mod{q}}$, which results in the state
$$
 \ket{ \psi} =   \sum_{\vec x \in \Z_q^m} \rho_\sigma(\vec x) \ket{\vec x} \otimes \ket{\vec A \cdot \vec x \Mod{q}}.
  $$
  \item Measure the second register in the computational basis, which results in $\vec y \in \Z_q^n$ and a state
    $$
    \ket{\psi_{\vec y}} = \sum_{\substack{\vec x \in \Z_q^m:\\ \vec A \vec x= \vec y \Mod{q}}} \rho_\sigma(\vec x) \ket{\vec x}.
    $$
\end{enumerate}

Finally, we use the following lemma which characterizes the Fourier transform of a partially measured Gaussian superposition.

\begin{lemma}[\cite{Poremba22}, Lemma 16]\label{lem:duality}
Let $m \in \N$, $q \geq 2$ be a prime and $\sigma \in (\sqrt{8m},q/\sqrt{8m})$.
Let $\vec A \in \Z_q^{n \times m}$ be a matrix whose columns generate $\Z_q^n$ and let $\vec y \in \Z_q^n$ be arbitrary. Then, the $q$-ary quantum Fourier transform of the (normalized variant of the) Gaussian coset state
$$
 \ket{\psi_{\vec y}} = \sum_{\substack{\vec x \in \Z_q^m\\ \vec A \vec x = \vec y \Mod{q}}}\rho_{\sigma}(\vec x) \ket{\vec x}
$$
is within negligible (in $m \in \N$) trace distance of the (normalized variant of the) Gaussian state
$$
 \ket{\hat\psi_{\vec y}} = \sum_{\vec s \in \Z_q^n} \sum_{\vec e \in \Z_q^m} \rho_{q/\sigma}(\vec e) \, \omega_q^{-\langle\vec s,\vec y \rangle} \ket{\vec s^\intercal \vec A + \vec e^\intercal \Mod{q}}.
$$
\end{lemma}

\subsection{(Strong) Gaussian-Collapsing Property.}

We use the following result which says that the Ajtai hash function is target-collapsing with respect to the truncated discrete Gaussian distribution.

\begin{theorem}[Gaussian-collapsing property, \cite{Poremba22}, Theorem 4]\label{thm:Gauss-collapsing}
Let $n\in \N$ and $q$ be a prime with $m \geq 2n \log q$, each parameterized by $\lambda \in \N$. Let  $\sigma \in (\sqrt{8m},q/\sqrt{8m})$.
Then,
the following samples are computationallyindistinguishable assuming the quantum hardness of decisional $\mathsf{LWE}_{n,q,\alpha q}^m$, for any noise ratio $\alpha \in (0,1)$ with relative noise magnitude $1/\alpha= \sigma \cdot 2^{o(n)}:$
$$
\Bigg(\vec  A \rand \Z_q^{n \times m},\,\, \ket{\psi_{\vec y}}=\sum_{\substack{\vec x \in \Z_q^m\\ \vec A \vec x = \vec y}}\rho_{\sigma}(\vec x) \,\ket{\vec x}, \,\,\vec y\in \Z_q^n \Bigg)\,\, \approx_c \,\,\,\, \Bigg(\vec  A \rand \Z_q^{n \times m}, \,\,\ket{\vec x_0},\,\, \vec A \cdot \vec x_0 \,\in \Z_q^n\Bigg)
$$
where $(\ket{\psi_{\vec y}},\vec y) \leftarrow \mathsf{GenGauss}(\vec A,\sigma)$ and where $\vec x_0 \sim D_{\Z_q^m,\frac{\sigma}{\sqrt{2}}}$ is a discrete Gaussian error.
\end{theorem}

Using our main theorem on certified-everlasting target-collapsing hashes in \Cref{thm:CETC-generalization}, we can now prove a stronger variant of \Cref{thm:Gauss-collapsing}. We show the following:

\begin{theorem}\label{thm:ajtai-certified-everlasting}
Let $\lambda \in \N$ be the security parameter, $n(\lambda) \in \N$, $q(\lambda) \in \N$ be a modulus, $m \geq 2n \log q$ and $\sigma \in (\sqrt{2m},q/\sqrt{2m})$. Then, the Ajtai hash function family
$\algo H = \{H_\lambda\}_{\lambda \in \N}$ with
$$
H_\lambda = \left\{ h_{\vec A}: \big\{\vec x \in \Z_q^m : \|\vec x\| \leq \sigma \sqrt{m/2}\big\} \rightarrow \Z_q^n \, \text{ s.t. } \, h_{\vec A}(\vec x) = \vec A \cdot \vec x \Mod{q}; \, \vec A \in \Z_q^{n \times m} \right\}.
$$
is certified everlasting $D_{\Z_q^m,\frac{\sigma}{\sqrt{2}}}$-target-collapsing assuming the quantum hardness of $\SIS_{n,q,\sigma\sqrt{2m}}^m$ and $\mathsf{LWE}_{n,q,\alpha q}^m$, for any noise ratio $\alpha \in (0,1)$ with relative noise magnitude $1/\alpha= \sigma \cdot 2^{o(n)}$.
\end{theorem}
\begin{proof}
By the Gaussian-collapsing property in \Cref{thm:Gauss-collapsing}, it follows that $\algo H$ is $D_{\Z_q^m,\frac{\sigma}{\sqrt{2}}}$-target-collapsing assuming the quantum hardness of $\mathsf{LWE}_{n,q,\alpha q}^m$, for any noise ratio $\alpha \in (0,1)$ with relative noise magnitude $1/\alpha= \sigma \cdot 2^{o(n)}$. Moreover, from the quantum hardness of $\SIS_{n,q,\sigma\sqrt{2m}}^m$ it follows that $\algo H$ is $D_{\Z_q^m,\frac{\sigma}{\sqrt{2}}}$-target-collision-resistant. Therefore, the claim follows from \Cref{thm:CETC-generalization}.
\end{proof}

As a corollary, we immediately recover the so-called strong Gaussian-collapsing property of the Ajtai hash function which was previously stated as a conjecture by Poremba~\cite{Poremba22}.

\begin{corollary}[Strong Gaussian-collapsing property]\label{SGC}\ \\
Let $\lambda \in \N$ be the security parameter, $n(\lambda) \in \N$, $q(\lambda) \in \N$ be a modulus and $m > 2n \log q$. Let $\sigma = \Omega(\sqrt{m})$ be a parameter. Then, the Ajtai hash function satisfies the strong Gaussian-collapsing property assuming the quantum hardness of $\SIS_{n,q,\sigma\sqrt{2m}}^m$ and $\mathsf{LWE}_{n,q,\alpha q}^m$, for any noise ratio $\alpha \in (0,1)$ with relative noise magnitude $1/\alpha= \sigma \cdot 2^{o(n)}$. In other words, for every $\QPT$ adversary $\algo A$,
$$
|\Pr[\mathsf{StrongGaussCollapseExp}_{\algo A,n,m,q,\sigma}(0)=1] - \Pr[\mathsf{StrongGaussCollapseExp}_{\algo A,n,m,q,\sigma}(1)=1] \leq \negl(\lambda)
$$
where $\mathsf{StrongGaussCollapseExp}_{\algo A,n,m,q,\sigma}(b)$ is the experiment from \Cref{fig:SGC}.
\end{corollary}
\begin{proof} 
To prove the statement, we can simply reduce the certified everlasting $D_{\Z_q^m,\frac{\sigma}{\sqrt{2}}}$-target-collapsing security of the Ajtai hash $\vec A = [\bar{\vec A} \, \| \, \bar{\vec A} \cdot \bar{\vec x} \Mod{q}] \in \Z_q^{n \times m}$ with $\bar{\vec x} \rand \bit^{m-1}$ to the
strong Gaussian-collapsing security, and invoke \Cref{thm:ajtai-certified-everlasting}. Here we rely on the fact that the distribution of $\vec A$ is statistically close to uniform by the leftover hash lemma whenever $m > 2n \log q$. Now consider the unbounded reduction that given $\vec A \in \Z_q^{n \times m}$, samples a uniformly random vector $\vec t = (\vec x,-1) \in \Z^{m}$ with $\vec x \in \bit^{m-1}$ such that $\bar{\vec A} \vec x = \bar{\vec A} \bar{\vec x} \Mod{q}$, and then runs the second part of the strong Gaussian-collapsing adversary on input $\vec t$ in order to predict the challenger's bit. Note that such vectors $\vec t$ exist because of how the matrix $\vec A$ is constructed in the experiment. If the strong Gaussian-collapsing adversary has noticeable advantage, then so does the reduction, which would break the certified everlasting $D_{\Z_q^m,\frac{\sigma}{\sqrt{2}}}$-target-collapsing security of the Ajtai hash. 
\end{proof}

\subsection{Dual-Regev Public-Key Encryption with Publicly-Verifiable Deletion}

We now consider the following Dual-Regev encryption scheme introduced by Poremba~\cite{Poremba22}.

\begin{construction}[Dual-Regev $\mathsf{PKE}$ with Publicly-Verifiable Deletion]\label{cons:dual-regev-cd}
Let $n \in \N$ be the security parameter, $m \in \N$ and $q$ be a prime. Let $\alpha \in (0,1)$ and $\sigma = 1/\alpha$ be parameters.
The Dual-Regev $\mathsf{PKE}$ scheme $\mathsf{DualPKECD} = (\KeyGen,\Enc,\Dec,\Del,\Vrfy)$ with certified deletion is defined as follows:
\begin{description}
\item $\KeyGen(1^\lambda) \rightarrow (\pk,\sk):$ sample a random matrix $\bar{\vec A} \rand \Z_q^{n\times m}$ and a vector $\bar{\vec x} \rand \bit^{m}$
and choose $\vec A = [\bar{\vec A} \| \bar{\vec A} \cdot \bar{\vec x} \Mod{q}]$.
Output $(\pk,\sk)$, where $\pk=\vec A \in \Z_q^{n \times (m+1)}$ and $\sk = (-\bar{\vec x}, 1) \in \Z_q^{m+1}$.
\item $\Enc(\pk,b) \rightarrow (\vk,\ket{\ct})$: parse the public key as $\vec A \leftarrow \pk$. To encrypt a single bit $b \in \bit$, generate the following pair for a random $\vec y \in \Z_q^n$:
$$
\vk \leftarrow (\vec A, \vec y), \quad \ket{\ct} \leftarrow \sum_{\vec s \in \Z_q^n} \sum_{\vec e \in \Z_q^{m+1}} \rho_{q/\sigma}(\vec e) \, \omega_q^{-\langle\vec s,\vec y \rangle}\ket{\vec s^\intercal \vec A + \vec e^\intercal +b \cdot (0,\dots,0, \lfloor\frac{q}{2} \rfloor)},
$$
where $\vk$ is the public verification key and $\ket{\ct}$ is an $(m+1)$-qudit quantum ciphertext.

\item $\Dec(\sk,\ket{\ct}) \rightarrow \bit:$ to decrypt, measure $\ket{\ct}$ in the computational basis with outcome $\vec c \in \Z_q^m$. Compute $\vec c^\intercal \cdot \sk \in \Z_q$ and output $0$, if it is closer to $0$ than to $\lfloor\frac{q}{2}\rfloor$, and output $1$, otherwise.

\item $\Del(\ket{\ct}) \rightarrow \pi:$ Measure $\ket{\ct}$ in the Fourier basis and output the measurement outcome $\pi \in \Z_q^{m+1}$.
\item $\Vrfy(\vk,\pi) \rightarrow \{\top,\bot\}:$ to verify a deletion certificate $\pi \in \Z_q^{m+1}$, parse $(\vec A,\vec y) \leftarrow \vk$ and output $\top$, if $\vec A \cdot \pi = \vec y \Mod{q}$ and $\| \pi \| \leq \sqrt{m+1}/\sqrt{2}\alpha$, and output $\bot$, otherwise.
\end{description}
\end{construction}

Let us now illustrate how the deletion procedure takes place. Recall from \Cref{lem:duality} that the Fourier transform of the ciphertext $\ket{\ct}$ results in the \emph{dual} quantum state
\begin{align}\label{eq:dual-with-phase}
\ket{\widehat{\ct}}=\sum_{\substack{\vec x \in \Z_q^{m+1}:\\ \vec A \vec x = \vec y \Mod{q}}}\rho_{\sigma}(\vec x) \, \omega_q^{\langle\vec x,b \cdot (0,\dots,0,  \lfloor\frac{q}{2} \rfloor)\rangle} \,\ket{\vec x}.
\end{align}
In other words, a Fourier basis measurement of $\ket{\ct}$ necessarily erases all information about the plaintext $b \in \bit$ and results in a \emph{short} vector $\pi \in \Z_q^{m+1}$ such that $\vec A \cdot \pi = \vec y \Mod{q}$. Hence, to publicly verify a deletion certificate we can simply check whether it is a solution to the $\ISIS$ problem specified by the verification key $\vk=(\vec A,\vec y)$. Using \Cref{thm:ajtai-certified-everlasting}, we obtain the following:

\begin{theorem}
Let $n\in \N$ and let $q \geq 2$ be a prime modulus such that $q=2^{o(n)}$ and $m \geq 2n \log q$. Let $\sigma \in (\sqrt{8m},q/\sqrt{8m})$ and $\alpha \in (0,1)$ be a noise ratio with $1/\alpha= 2^{o(n)} \cdot \sigma$.
Then, the Dual-Regev public-key encryption scheme in \Cref{cons:dual-regev-cd} has everlasting certified deletion security assuming the quantum (subexponential) hardness of $\mathsf{LWE}_{n,q,\alpha q}^m$ and $\SIS_{n,q,\sigma\sqrt{2m}}^m$.
\end{theorem}
\begin{proof}
The proof is identical to the template used in~\cite[Theorem 7]{Poremba22}, except that the adversary is allowed to be computationally unbounded once the deletion certificate is submitted. This is in contrast with the original proof who considered forwarding the \emph{secret key} during the security experiment. We remark that we do not invoke the strong Gaussian-collapsing property to prove the indistinguishability of the hybrids; instead we use the (stronger) notion of certified everlasting $D_{\Z_q^m,\frac{\sigma}{\sqrt{2}}}$-target-collapsing property of the Ajtai hash shown in \Cref{thm:ajtai-certified-everlasting}. This results in the stronger notion of everlasting certified
deletion security.
\end{proof}

\subsection{Dual-Regev (Leveled) Fully Homomorphic Encryption with Publicly-Verifiable Deletion}\label{sec:dualfhe-pvd}

A homomorphic encryption scheme with certified deletion~\cite{Broadbent_2020,Poremba22,BBK22} is a scheme that supports both homomorphic operations as well as certified deletion of quantum ciphertexts. Here, the two properties are thought of as \emph{separate} features that may or may not be mutually exclusive. Several works~\cite{Poremba22,BBK22,BGGKMRR} have also considered the possibility of realizing both homomorphic evaluation and certified deletion \emph{simultaneously} within a single (possibly interactive) protocol. For example, Poremba~\cite{Poremba22} proposed a four-message protocol that allows a client to learn the outcome of a homomorphic evaluation performed by an untrusted quantum server, while simultaneously certifying that the server has subsequently deleted all data. Bartusek and Khurana~\cite{BBK22} subsequently defined the notion of a four-message protocol for \emph{blind delegation with certified deletion}, which can be instantiated using any $\FHE$ scheme with certified deletion. Crucially, both of the aforementioned four-message protocols require that the server is \emph{honest} during the evaluation phase of the protocol. Finally, in a subsequent follow-up work, Bartusek et al.~\cite{BGGKMRR} constructed a \emph{maliciously} secure bind delegation protocol with certified deletion which relied on succinct non-interactive arguments (SNARGs) for polynomial-time computation.

In this section, we recall the Dual-Regev (leveled) fully homomorphic encryption scheme with publicly-verifiable deletion introduced by Poremba~\cite{Poremba22}. The scheme is based on the \emph{dual variant} of of the (leveled) $\FHE$ scheme by Gentry, Sahai and Waters~\cite{GSW2013,mahadev2018classical}. 
Using our main result on certified-everlasting target-collapsing hashes in \Cref{thm:CETC-generalization}, we then prove the scheme achieves certified-everlasting security assuming the quantum hardness of $\LWE$ and $\SIS$.
Contrary to related works~\cite{Poremba22,BBK22,BGGKMRR}, we only consider the basic notion of $\FHE$ with publicly-verifiable deletion which treats both properties as separate features.

\paragraph{Parameters.} Let $\lambda \in \N$ be the security parameter and let $n \in \N$. Let $L$ be an upper bound on the depth of the polynomial-sized Boolean circuit which is to be evaluated. We choose the following set of parameters for the Dual-Regev leveled $\FHE$ scheme (each parameterized by $\lambda$).
\begin{itemize}
    \item a prime modulus $q \geq 2$.
    \item an integer $m \geq 2n \log q$.
    \item an integer $N = (m+1) \cdot \lceil \log q \rceil$.
     
     \item a noise ratio $\alpha\in (0,1)$ such that
$$
\sqrt{8(m+1)}\leq \alpha q \leq \frac{q}{\sqrt{8}(m+1)\cdot (N+1)^L}.
$$
\end{itemize}

\begin{construction}[Dual-Regev leveled $\mathsf{FHE}$ scheme with certified deletion]\label{cons:FHE-cd}
Let $\lambda \in \N$ be the security parameter.
The Dual-Regev (leveled) $\mathsf{FHE}$ scheme $\mathsf{DualFHECD} = (\KeyGen,\Enc,\Dec,\Eval,\Del,\Vrfy)$ with certified deletion consists of the following algorithms.
\begin{description}
\item $\KeyGen(1^\lambda,1^L) \rightarrow (\pk,\sk):$ sample $\bar{\vec A} \rand \Z_q^{n\times m}$ and vector $\bar{\vec x} \rand \bit^{m}$
and let $\vec A = [\bar{\vec A} \| \bar{\vec A} \cdot \bar{\vec x} \Mod{q}]^\intercal$.
Output $(\pk,\sk)$, where $\pk=\vec A \in \Z_q^{(m+1) \times n}$ and $\sk = (-\bar{\vec x}, 1) \in \Z_q^{m+1}$.
\item $\Enc(\pk,x) \rightarrow (\vk,\ket{\ct}):$ to encrypt a bit $x\in \bit$, parse the public key as $\vec A \in \Z_q^{(m+1) \times n} \leftarrow \pk$ and generate the following pair consisting of a verification key and ciphertext for a random $\vec Y \in \Z_q^{n \times N}$ with columns $\vec y_1,\dots,\vec y_N \in \Z_q^{n}$:
$$
\vk \leftarrow (\vec A,\vec Y), \quad\,\,
\ket{\ct} \leftarrow \sum_{\vec S \in \Z_q^{n \times N}} \sum_{\vec E \in \Z_q^{(m+1)\times N}} \rho_{q/\sigma}(\vec E) \, \omega_q^{-\Tr[\vec S^\intercal \vec Y]} \ket{\vec A\cdot \vec S + \vec E + x \cdot \vec G},
$$
where $\vec G = [\vec I \, \| \, 2 \vec I \, \| \, \dots \, \| \, 2^{\lceil \log q \rceil -1} \vec I] \in \Z_q^{(m+1) \times N}$  denotes the \emph{gadget matrix} and where $\sigma = 1/\alpha$.

\item $\Eval(\mathsf{C}_0,\mathsf{C}_1) \rightarrow \mathsf{C}_0 \mathsf{C}_1 \mathsf{C}$: to apply a $\mathsf{NAND}$ gate onto two registers $\mathsf{C}_0$ and $\mathsf{C}_1$ (possibly part of a larger ciphertext), append an ancilla system $\ket{\vec 0}_{\mathsf{C}}$, and apply the unitary $U_{\mathsf{NAND}}$, defined by
$$
U_{\mathsf{NAND}}: \quad \ket{\vec X}_{\mathsf{C}_0} \otimes \ket{\vec Y}_{\mathsf{C}_1} \otimes \ket{\vec Z}_\mathsf{C} \quad \rightarrow \quad  \ket{\vec X}_{\mathsf{C}_0} \otimes \ket{\vec Y}_{\mathsf{C}_1} \otimes \ket{\vec Z + \vec G - \vec X \cdot \vec G^{-1}(\vec Y) \Mod{q}}_{\mathsf{C}},
$$
where $\vec X,\vec Y,\vec Z \in \Z_q^{(m+1)\times N}$
and $\vec G^{-1}$ is the (non-linear) inverse operation such that $\vec G \circ \vec G^{-1} = \vec I$.
Output the resulting registers $\mathsf{C}_0 \mathsf{C}_1 \mathsf{C}$.

\item $\Dec(\sk,\mathsf{C}) \rightarrow \bit \, \mathbf{or} \, \bot:$ measure the register $\mathsf{C}$ in the computational basis to obtain $\vec C \in \Z_q^{(m+1)\times N}$ and compute $c = \sk^\intercal \cdot \vec c_N \in \Z \cap (-\frac{q}{2},\frac{q}{2}]$, where $\vec c_N \in \Z_q^{m+1}$ is the $N$-th column of $\vec C$; output $0$, if $c$
is closer to $0$ than to $\lfloor\frac{q}{2}\rfloor$,
and output $1$, otherwise.

\item $\Del(\ket{\ct}) \rightarrow \pi:$ measure $\ket{\ct}$ in the Fourier basis with outcomes $\pi = (\pi_1|\dots|\pi_N) \in \Z_q^{(m+1)\times N}$.

\item $\Vrfy(\vk,\pk,\pi) \rightarrow \bit:$ to verify the deletion certificate $\pi = (\pi_1\|\dots\|\pi_N) \in \Z_q^{(m+1)\times N}$, parse $(\vec A \in \Z_q^{(m+1) \times n},(\vec y_1 \|\dots \|\vec y_N)  \in \Z_q^{n \times N}) \leftarrow \vk$ and output $\top$, if both $\vec A^\intercal \cdot \pi_i = \vec y_i \Mod{q}$ and $\| \pi_i \| \leq \sqrt{m+1}/\sqrt{2}\alpha$ for every $i \in [N]$, and output $\bot$, otherwise.
\end{description}
\end{construction}

For additional details on the correctness of the scheme, we refer to Section $9$ of \cite{Poremba22}.

\begin{theorem}\label{thm:FHE-CD-security} Let $\lambda \in \N$ be the security parameter and let $L$ be an upper bound on the size of the Boolean circuit which is to be evaluated. Let $n \in \N$, let $q\geq 2$ be a prime modulus and let $m \geq 2 n \log q$. Let $N = (m+1) \cdot \lceil \log q \rceil$. Let $\alpha \in (0,1)$ be a noise ratio such that$$
\sqrt{8(m+1)N}\leq \alpha q \leq \frac{q}{\sqrt{8}(m+1)\cdot (N+1)^L}.
$$
Then, $\mathsf{DualFHECD}$ in \Cref{cons:FHE-cd} has everlasting certified deletion security assuming the quantum (subexponential) hardness of $\SIS_{n,q,\sigma\sqrt{2m}}^m$ and $\mathsf{LWE}_{n,q,\alpha q}^m$.
\end{theorem}
\begin{proof}
The proof is identical to the template in~\cite[Theorem 10]{Poremba22}, except that the adversary is allowed to be computationally unbounded once the deletion certificate is submitted. This is in contrast with the original proof who considered forwarding the \emph{secret key} during the security experiment. We remark that we do not invoke the strong Gaussian-collapsing property to prove the indistinguishability of the hybrids; instead we use the (stronger) notion of certified everlasting $D_{\Z_q^m,\frac{\sigma}{\sqrt{2}}}$-target-collapsing property of the Ajtai hash function shown in \Cref{thm:ajtai-certified-everlasting}. This results in the stronger notion of everlasting certified
deletion security.
\end{proof}

\section{Publicly-Verifiable Deletion from Balanced Binary-Measurement TCR}\label{sec:regularOWF}

In this section, we show how to build a variety of cryptographic primitives with $\PVD$ from a specific type of hash function that we call \emph{balanced binary-measurement target-collision-resistant}.

\begin{definition}[Balanced Binary-Measurement TCR Hash]\label{def:BBMhash}
A hash function family $\cH = \{H_\secp : \{0,1\}^{m(\secp)} \to \{0,1\}^{n(\secp)}\}_{\secp \in \bbN}$ is \emph{balanced binary-measurement target-collision-resistant} if:
\begin{enumerate}
    \item There exists a family of efficiently computable \emph{single-output-bit} measurement functions $\cM = \{\{M[h] : \{0,1\}^{m(\secp)} \to \{0,1\}\}_{h \in H_\secp}\}_{\secp \in \bbN}$ such that $\cH$ is $\cM$-target-collision-resistant (\cref{def:target-CR}).
    \item There exists a constant $\delta > 0$ such that\footnote{It is also straightforward to generalize our results to any $\delta(\secp) = 1/\poly(\secp)$.} \[\Pr_{h \gets H_\secp, x \gets \{0,1\}^{m(\secp)}}\left[\bigg|\frac{A_{h,x,0} - A_{h,x,1}}{A_{h,x,0} + A_{h,x,1}}\bigg| \leq 1-\delta\right] = 1-\negl(\secp),\]
    
    where $A_{h,x,b} \coloneqq |\{x' \in h^{-1}(h(x)) : M[h](x') = b\}|.$
\end{enumerate}
\end{definition}

\begin{remark}
By \cref{thm:TC-from-TCR} and \cref{thm:CETC-generalization}, any balanced binary-measurement TCR $\cH$ with associated measurement function $\cM$ is also $\cM$-target-collapsing and certified everlasting $\cM$-target-collapsing.
\end{remark}

\subsection{Commitments}
\label{sec:com}

A \emph{canonical quantum bit commitment} \cite{Yan} consists of a family of pairs of unitaries $\{(Q_{\secp,0},Q_{\secp,1})\}_{\secp \in \bbN}$. To commit to a bit $b$, the committer applies $Q_{\secp,b}$ to the all-zeros state $\ket{0}$ to obtain a state on registers $C$ and $R$, and sends register $C$ to the receiver. To open, the committer sends the bit $b$ and the remaining state on register $R$. The receiver applies $Q_{\secp,b}^\dagger$ to registers $(C,R)$, measures the result in the standard basis, and accepts if all zeros are observed.

\begin{definition}[Computational Hiding]
A canonical quantum bit commitment $\{(Q_{\secp,0},Q_{\secp,1})\}_{\secp \in \bbN}$ satisfies \emph{computational hiding} if for any QPT adversary $\{\cA_\secp\}_{\secp \in \bbN}$,

\[\left|\Pr\left[\cA_\secp(\Tr_R\left(Q_{\secp,0}\ket{0}\right)) = 1\right] - \Pr\left[\cA_\secp(\Tr_R\left(Q_{\secp,1}\ket{0}\right)) = 1\right] \right| = \negl(\secp).\]
\end{definition}

\begin{definition}[Honest Binding]
A canonical quantum bit commitment $\{(Q_{\secp,0},Q_{\secp,1})\}_{\secp \in \bbN}$ satisfies \emph{honest binding} if for any auxiliary family of states $\{\ket{\psi_\secp}\}_{\secp \in \bbN}$ on register $Z$ and any family of physically realizable unitaries $\{U_\secp\}_{\secp \in \bbN}$ on registers $R,Z$, 

\[\left\| \left(Q_{\secp,1}\dyad{0}{0}Q_{\secp,1}^\dagger\right)U\left(Q_{\secp,0}\ket{0}\ket{\psi}\right)\right\| = \negl(\secp).\]
\end{definition}

\begin{definition}[Publicly-Verifiable Deletion]
A canonical quantum bit commitment $\{(Q_{\secp,0},Q_{\secp,1})\}_{\secp \in \bbN}$ has \emph{publicly-verifiable deletion} if there exists a measurement $\{V_{\secp}\}_{\secp \in \bbN}$ on register $R$, a measurement $\{D_{\secp}\}_{\secp \in \bbN}$ on register $C$, and a classical predicate $\Ver(\cdot,\cdot) \to \{\top,\bot\}$ that satisfy the following properties.

\begin{itemize}
    \item \textbf{Correctness of deletion.} For any $b \in \{0,1\}$, it holds that
    \[\Pr\left[\Ver(\vk,\pi) = \top : (\vk,\pi) \gets (V_\secp \otimes D_\secp)Q_{\secp,b}\ket{0}\right] = 1-\negl(\secp).\]
    \item \textbf{Certified everlasting hiding.} For any QPT adversary $\cA = \{\cA_\secp\}_{\secp \in \bbN}$, it holds that 
    \[\TD\left(\mathsf{EvExp}_{\cA}(\secp,0), \mathsf{EvExp}_{\cA}(\secp,1)\right) = \negl(\secp),\] where $\mathsf{EvExp}_{\cA}(\secp,b)$ is the following experiment.
    \begin{itemize}
        \item Prepare $Q_{\secp,b}\ket{0}$, measure register $R$ with $V_\secp$ to obtain $\vk$, and send $(\vk,C)$ to $\cA_\secp$.
        \item Parse $\cA_{\secp}$'s output as a deletion certificate $\pi$ and a left-over state $\rho$. If $\Ver(\vk,\pi) = \bot$, output $\bot$, and otherwise output $\rho$.
    \end{itemize}
\end{itemize}
\end{definition}

%\james{Define zero-knowledge proof with certified everlasting zero-knowledge}

\paragraph{Construction.} We construct a quantum canonical bit commitment with $\PVD$ as follows. Let $\cH = \{H_\secp : \{0,1\}^{m(\secp)} \to \{0,1\}^{n(\secp)}\}_{\secp \in \bbN}$ be a balanced binary-measurement TCR hash with associated measurement function $\cM = \{\{M[h]\}_{h \in H_\secp}\}_{\secp \in \bbN}$, and let $m = m(\secp)$, $n = n(\secp)$. For any $h \in H_\secp,y \in \{0,1\}^n$, and $b \in \{0,1\}$, we will define the state 

\[\ket{\psi_{h,y,b}} \coloneqq \frac{1}{\sqrt{|h^{-1}(y)|}}\sum_{x : h(x)= y}(-1)^{M[h](x)}\ket{x}.\]

\begin{itemize}
    \item Consider the following procedure $S_{\secp,b}$. Sample $h \gets H_\secp$ and for $i \in [\secp]$, prepare the state \[\frac{1}{\sqrt{2^m}}\sum_{x \in \{0,1\}^m}(-1)^{b \cdot M[h](x)}\ket{x}\ket{h(x)},\] and measure the second register to obtain $y_i$ and left-over state $\ket{\psi_{h,y_i,b}}$. Then, output \[(h,y_1,\dots,y_\secp),
    \bigotimes_{i \in [\secp]}\ket{\psi_{h,y_i,b}}.\]

    Now, $Q_{\secp,b}$ will be the purification of $S_{\secp,b}$, where the output register is $C$ and the auxiliary register is $R$. That is, $Q_{\secp, b}$ prepares the state
    \[\frac{1}{\sqrt{|H_\secp|2^{\secp m}}}\sum_{h,x_1,\dots,x_\secp}(-1)^{b \cdot \bigoplus_{i \in [\secp]}M[h](x_i)}\ket{h,h(x_1),\dots,h(x_\secp)}_R\ket{h,h(x_1),\dots,h(x_\secp),x_1,\dots,x_\secp}_C.\]
    \item $V_\secp$ measures register $R$ in the standard basis to obtain $\vk = (h,y_1,\dots,y_\secp)$. $D_\secp$ measures register $C$ in the standard basis to obtain $(h,y_1,\dots,y_\secp,x_1,\dots,x_\secp)$, and outputs $\pi = (x_1,\dots,x_\secp)$.
    \item $\Ver((h,y_1,\dots,y_\secp),(x_1,\dots,x_\secp))$ outputs $\top$ iff $h(x_i) = y_i$ for all $i \in [\secp]$.
\end{itemize}

\begin{theorem}\label{thm:commitment}
The above construction satisfies computational hiding, honest binding, and publicly-verifiable deletion. Thus, assuming the existence of a balanced binary-measurement TCR hash, there exists a quantum canonical bit commitment with $\PVD$.
\end{theorem}

\begin{proof}
%First, we note that computational hiding and certified everlasting hiding follow from \cref{corollary:target-collapsing} and a straightforward hybrid argument, so it suffices to show honest binding.\james{TODO: make the hybrid argument explicit}

First we argue computational hiding. On a commitment to $b$, the receiver sees the mixed state 

\[\E_{h,y_1,\dots,y_\secp}\left[\bigotimes_{i \in [\secp]}\ket{\psi_{h,y_i,b}}\right],\] where the expectation is over sampling $h \gets \cH$ and measuring random $y_1,\dots,y_\secp$. Note the following two facts.

\begin{enumerate}
	\item Given any state $\ket{\psi_{h,y,b}}$, let $M[h]\left(\ket{\psi_{h,y,b}}\right)$ be the mixed state that results from measuring the bit $M[h](\cdot)$ on $\ket{\psi_{h,y,b}}$. By the $\cM$-target-collapsing of $\cH$, we have that for any $b \in \{0,1\}$, \[\E_{h,y}\left[\ket{\psi_{h,y,b}}\right] \approx_c \E_{h,y}\left[M[h]\left(\ket{\psi_{h,y,b}}\right)\right],\] where $\approx_c$ denotes computational indistinguishability. The case of $b=0$ follows directly by definition of $\cM$-target-collapsing and the case of $b=1$ follows because a reduction can efficently map $\ket{\psi_{h,y,0}}$ to $\ket{\psi_{h,y,1}}$ using the fact that $M[h]$ is efficiently computable.
	\item For any $h,y$, $M[h]\left(\ket{\psi_{h,y,0}}\right)$ and $M[h]\left(\ket{\psi_{h,y,1}}\right)$ are equivalent states, which follows by definition.
\end{enumerate}

Thus, we can run the following hybrid argument.

\begin{itemize}
	\item $\Hyb_0$: The receiver is given a commitment to 0.
	\item $\Hyb_1\dots\Hyb_\secp$: In $\Hyb_i$, we switch $\ket{\psi_{h,y,0}}$ to $M[h]\left(\ket{\psi_{h,y,0}}\right)$. This is computationally indistinguishable from $\Hyb_{i-1}$ by the first fact above.
	\item $\Hyb_{\secp+1}$: Switch $M[h]\left(\ket{\psi_{h,y_i,0}}\right)$ to $M[h]\left(\ket{\psi_{h,y_i,1}}\right)$ for all $i \in [\secp]$. This is perfectly indistinguishable from $\Hyb_\secp$ by the second fact above.
	\item $\Hyb_{\secp+2}\dots\Hyb_{2\secp + 1}$: In $\Hyb_{i + \secp + 1}$, we switch $M[h]\left(\ket{\psi_{h,y_i,1}}\right)$ to $\ket{\psi_{h,y_i,1}}$. This is computationally indistinguishable from $\Hyb_{i+\secp}$ by the first fact above.
\end{itemize}

This completes the proof of computational hiding. Next, since $\cH$ satisfies \emph{certified everlasting} $\cM$-target-collapsing, we see that each hybrid is \emph{statistically close} when the receiver outputs a valid deletion certificate. Thus, the same proof establishes publicly-verifiable deletion.

%For this, it suffices to show (by Uhlmann's theorem) that the trace distance between the outputs of $S_{\secp,0}$ and $S_{\secp,1}$ is $1-\negl(\secp)$. We lower bound the trace distance by demonstrating a measurement that distinguishes with probability $1-\negl(\secp)$. 

%\james{TODO: Instead of citing Uhlmann's theorem, we can say that it suffices to demonstrate a measurement on just the $R$ register that accepts $Q_{\secp,0}\ket{0}$ with probability 1 and rejects $Q_{\secp,0}\ket{0}$ with probability $1-\negl(\secp)$, and note that any $U$ in the definition of honest-binding cannot change the success probability of this measurement}

Finally, we show honest binding. For this, it suffices to demonstrate a measurement on register $C$ that accepts with probability 1 on the output of $Q_{\secp,0}$ and with probability $\negl(\secp)$ on the output of $Q_{\secp,1}$. This suffices because any $U$ that breaks honest binding must then necessarily affect the result of this measurement by a $\nonnegl(\secp)$ amount, which is impossible since $U$ does not operate on $C$.

The measurement takes the classical part of the output $(h,y_1,\dots,y_\secp)$ and attempts to project the quantum part onto \[\dyad{\psi_{h,y_1,0}}{\psi_{h,y_1,0}} \otimes \dots \otimes \dyad{\psi_{h,y_\secp,0}}{\psi_{h,y_\secp,0}}.\] Clearly this accepts the output of $S_{\secp,0}$ with probability 1, so it suffices to show that the output of $S_{\secp,1}$ is accepted with probability $\negl(\secp)$. To see this, we bound

\begin{align*}
    &\E_{h,y_1,\dots,y_\secp}\left[\prod_{i \in [\secp]}|\braket{\psi_{h,y_i,1}|\psi_{h,y_i,0}}|^2\right] \\
    &\E_{h,y_1,\dots,y_\secp}\left[\prod_{i \in [\secp]}\left(\frac{1}{|h^{-1}(y_i)|}\left(\sum_{x:h(x)=y_i}(-1)^{M[h]}\bra{x}\right)\left(\sum_{x:h(x)=y_i}\ket{x}\right)\right)^2\right] \\
    &= \E_{h,y_1,\dots,y_\secp}\left[\prod_{i \in [\secp]}\left(\frac{1}{|h^{-1}(y_i)|}\left(|\{x : h(x)=y_i,M[h](x)=0\}|-|\{x : h(x)=y_i,M[h](x)=1\}|\right)\right)^2\right] \\
    %&=\E_{h,y_1,\dots,y_\secp}\left[\prod_{i \in [\secp]}\left(\frac{p_{y_i}-p_{y_i \oplus \Delta}}{p_{y_i}+p_{y_i \oplus \Delta}}\right)^2\right]\\
    &\leq \left(1-\delta\right)^{2\secp} + \negl(\secp) \\
    &=\negl(\secp),
\end{align*}

where the inequality follows from property (2) of \cref{def:BBMhash}.
  
%\begin{itemize}
    %\item \james{computational hiding is straightforward from target-collapsing}
    %\item \james{to show honest binding, we can establish that the fidelity between the states on register $C$ when $b = 0$ vs $b=1$ is negligible, which follows from $\delta$-balanced and $\secp$ repetition}
    %\item \james{certified everlasting hiding is straightforward from certified everlasting target-collapsing}
%\end{itemize}
\end{proof}

\subsection{Public-Key Encryption}
\label{sec:pke}

\begin{definition}[Trapdoor Phase-Recoverability]
We say that a balanced binary-measurement TCR hash has \emph{trapdoor phase-recoverability} if there exist algorithms $\Samp,\Recover$ with the following properties.
\begin{itemize}
    \item $\Samp(1^\secp)$: The sampling algorithm samples a uniformly random function $h \in H_\secp$ along with a trapdoor $\td$.
    \item $\Recover(\td,y,X)$: There exist constants $c,\epsilon$ such that with probability $1-\negl(\secp)$ over $(h,\td) \gets \Samp(1^\secp)$, 
    
    \begin{align*}
        &\Pr_{x \gets \{0,1\}^m}\left[\Recover(\td,h(x),\ket{\psi_{h,h(x),0}}) \to 0\right] \geq c + \epsilon,\\
        &\Pr_{x \gets \{0,1\}^m}\left[\Recover(\td,h(x),\ket{\psi_{h,h(x),1}}) \to 0\right] \leq c - \epsilon,
    \end{align*}
    
    %\begin{align*}
    %\Pr_{x \gets \{0,1\}^m}\left[\Recover(\td,h(x),\ket{\psi_{h,h(x),0}}) \to 0\right] \geq c+\epsilon
    
    %&\Pr_{x \gets \{0,1\}^m}\left[\Recover(\td,h(x),\ket{\psi_{h,h(x),1}}) \to 0\right] \leq c-\epsilon,
    %\end{align*}
    
    where \[\ket{\psi_{h,y,b}} \coloneqq \frac{1}{\sqrt{|h^{-1}(y)|}}\sum_{x:h(x) = y}(-1)^{M[h](x)}\ket{x}.\]
\end{itemize}
\end{definition}

\begin{theorem}\label{thm:PKE}
Assuming the existence of a binary-measurement TCR hash $\cH$ with trapdoor phase-recoverability, there exists public-key encryption with $\PVD$.
\end{theorem}

\begin{proof}
This follows from essentially the same construction as commitments. Let $\cM$ be the measurement function associated with $\cH$ and let $(\Samp,\Invert)$ be the associated trapdoor algorithms. Then, the PKE with $\PVD$ is defined as follows.
\begin{itemize}
    \item $\Gen(1^\secp)$: Sample $(h,\td) \gets \Samp(1^\secp)$ and set $\pk \coloneqq h, \sk \coloneqq \td$.
    \item $\Enc(\pk,b)$: For $i \in [\secp]$, prepare the state \[\frac{1}{\sqrt{2^m}}\sum_{x \in \{0,1\}^m}(-1)^{b \cdot M[h](x)}\ket{x}\ket{h(x)},\] and measure the second register to obtain $y_i$ and left-over state $\ket{\psi_{h,y_i,b}}$. Then, set  \[\ket{\ct} \coloneqq \left(y_1,\dots,y_\secp,
    \bigotimes_{i \in [\secp]}\ket{\psi_{h,y_i,b}}\right), ~~ \vk \coloneqq (h,y_1,\dots,y_\secp).\]
    \item $\Dec(\sk,\ket{\ct})$: Parse $\ket{\ct}$ as $(y_1,\dots,y_\secp,X_1,\dots,X_\secp)$, for $i \in [\secp]$ run \[b_i \gets \Recover(\td,y_i,X_i),\]
    and output 0 if $|\{i : b_i = 0\}|/\secp > c$, and output 1 otherwise.

    %perform the measurement \[\left\{\dyad{\psi_{h,y_1,0}}{\psi_{h,y_1,0}} \otimes \dots \otimes \dyad{\psi_{h,y_\secp,0}}{\psi_{h,y_\secp,0}}, \bbI - \dyad{\psi_{h,y_1,0}}{\psi_{h,y_1,0}} \otimes \dots \otimes \dyad{\psi_{h,y_\secp,0}}{\psi_{h,y_\secp,0}}\right\}\] on the second part. Output 0 if the first outcome is observed and 1 otherwise. Note that since $\ket{\psi_{h,y_i,0}}$ can be efficiently prepared (to within $\negl(\secp)$ trace distance) given $\td$ and $y_i$, this measurement can be efficiently implemented.
    
    \item $\Del(\ket{\ct})$: Parse $\ket{\ct}$ as $(y_1,\dots,y_\secp,X_1,\dots,X_\secp)$ and measure $X_i$ in the standard basis to obtain $\pi \coloneqq (x_1,\dots,x_\secp)$.
    \item $\Vrfy(\vk,\pi)$: Output $\top$ iff $h(x_i) = y_i$ for all $i \in [\secp]$.
\end{itemize}

Correctness follows from a standard Hoeffding inequality and correctness of deletion (\cref{def:correctness-deletion}) is immediate. Certified deletion security (\cref{def:security-deletion}) follows from the $\cM$-target-collapsing and certified everlasting $\cM$-target-collapsing of $\cH$, using the same hybrid argument as in the proof of \cref{thm:commitment}.
\end{proof}

\subsection{A Generic Compiler}
\label{sec:generic}
Let $(\Gen,\Enc,\Dec,\Del,\Vrfy)$ be the encryption scheme defined last section, let $\cA = \{\cA_\secp\}_{\secp \in \bbN}$ be an adversary, let $p(\secp)$ be a polynomial, and let $\cZ = \{Z_\secp(\aux)\}_{\aux \in \{0,1\}^{p(\secp)}, \secp \in \bbN}$ be a (static or interactive) family of distributions that is semantically-secure against $\cA$ with respect to $\aux$. That is, in the static case, it holds that for any $\aux \in \{0,1\}^{p(\secp)}$,

\[\bigg| \Pr\left[\cA_\secp(Z_\secp(\aux)) = 1\right] - \Pr\left[\cA_\secp(Z_\secp(0^{p(\secp)})) = 1\right]\bigg| \leq \negl(\secp),\] and in the interactive case,

\[\bigg| \Pr\left[\cA_\secp^{Z_\secp(\aux)} = 1\right] - \Pr\left[\cA_\secp^{Z_\secp(0^{p(\secp)})} = 1\right]\bigg| \leq \negl(\secp),\] where $\cA_\secp^{Z_\secp(\aux)}$ indicates that $\cA_\secp$ can interact with $Z_\secp(\aux)$, which is the description of an interactive machine initialized with $\aux$.

%$\cZ = \{\cZ_\secp(\cdot,\cdot)\}_{\secp \in \bbN}$ be any (potentially quantum) operation such that for any distribution $\cD = \{\cD_\secp\}_{\secp \in \bbN}$ over pairs $(x,\ket{\psi})$, where $x$ is a classical string and $\ket{\psi}$ is a quantum state, it holds that
%\[\bigg|\Pr_{(x,\ket{\psi}) \gets \cD_\secp}\left[\cA_\secp(\cZ_\secp(x,\ket{\psi})) = 1\right] -  \Pr_{(x,\ket{\psi})\gets \cD_\secp}\left[\cA_\secp(\cZ_\secp(0,\ket{\psi})) = 1\right] \bigg| = \negl(\secp).\] That is, $\cZ$ is semantically-secure in its first input (with respect to $\cA$). 

\begin{lemma}\label{lemma:compiler}
Given any $\cA,\cZ$ as described above, define the experiment $\mathsf{EvEnc}_{\cA,\cZ,\secp}(b)$ as follows.

\begin{itemize}
    \item Sample $(h,\td) \gets \Gen(1^\secp)$ and $(\ket{\ct},\vk) \gets \Enc(h,b)$.
    \item Run $\cA_\secp(h,\vk,\ket{\ct},Z_\secp(\td))$, and parse their output as a deletion certificate $\pi$ and a left-over quantum state $\rho$.
    \item If $\Vrfy(\vk,\pi) = \top$, output $\rho$, and otherwise output $\bot$.
\end{itemize}
Then it holds that
\[\TD\left(\mathsf{EvEnc}_{\cA,\cZ,\secp}(0),\mathsf{EvEnc}_{\cA,\cZ,\secp}(1)\right) = \negl(\secp).\]
\end{lemma}

\begin{proof}

First, we confirm that $\cH$ is $(\cM,\cZ)$-target-collision-resistant.  To see this, we first use the semantic security of $\cZ$ to switch to a hybrid where $\cA_\secp$ receives $Z_\secp(0)$ rather than $Z_\secp(\td)$, and then appeal directly to the fact that $\cH$ is $\cM$-target-collision-resistant (\cref{thm:target-CR}). Then by \cref{thm:TC-from-TCR} and \cref{thm:CETC-generalization}, we have that $\cH$ is certified everlasting $(\cM,\cZ)$-target-collapsing. Using the same hybrid argument as in the proof of \cref{thm:commitment} then completes the proof.

%Given $\widetilde{\cF}$, define $\widetilde{\cF}'$ to be the same family except that $\cZ_\secp(\td)$ is included in the description of the hash function. Note that if $\widetilde{\cF}$ is $(\cU,\cP)$-target-collision-resistant, then so is  $\widetilde{\cF}'$. This follows because we can first use the semantic security of $\cZ$ to switch $\cZ_\secp(\td)$ to $\cZ_\secp(0)$, and then appeal to $(\cU,\cP)$-target-collision-resistance of $\widetilde{\cF}$. Then, by \cref{thm:TC-from-TCR} and \cref{thm:CETC-generalization}, we have that $\widetilde{\cF}'$ is certified everlasting $(\cU,\cP)$-target-collapsing, which immediately implies the theorem.
\end{proof}

By instantiating $\cZ$ with various crytographic primitives, we immediately gives the following applications. We do not write formal definitions of each of these primitives, and instead refer the reader to \cite{cryptoeprint:2022/1178} for these.

\begin{corollary}\label{cor:compiler}
Assuming the existence of a balanced binary-measurement TCR hash with trapdoor phase-recoverability, and post-quantum \[X \in \left\{\begin{array}{r}\text{quantum fully-homormophic encryption, attribute-based encryption}, \\ \text{witness encryption, timed-release encryption}\end{array}\right\},\] there exists $X$ with $\PVD$.
\end{corollary}

The implications to witness encryption and timed-release encryption follow immediately by encrypting $\td$ with the appropriate encryption scheme (and in the case of timed-release encryption, considering the class of parallel-time-bounded adversaries). We briefly remark on the other two implications.

\begin{itemize}
    \item \textbf{Fully-homomorphic encryption.} If we encrypt $\td$ using a \emph{quantum} fully-homomorphic encryption (QFHE) scheme, then we obtain (Q)FHE with publicly-verifiable deletion. The reason we need QFHE for the compiler is for evaluation correctness: we need to decrypt $\ket{\ct}$ homomorphically under the QFHE (using $\td$) in order to obtain a (Q)FHE encryption of the plaintext, which can then be operated on.
    \item \textbf{Attribute-based encryption.} If we encrypt $\td$ using an attribute-based encryption (ABE) scheme, we immediately obtain a correct ABE scheme with certified deletion. In order to argue that this scheme has certified deletion security, we appeal to \cref{lemma:compiler} with an interactive $Z_\secp$ that runs the ABE security game, encrypting its input $\td$ into the challenge ciphertext.
\end{itemize}

\subsection{Balanced Binary-Measurement TCR from Almost-Regular OWFs}
\label{sec:almostreg}

\begin{definition}[Almost-Regular Function]\label{def:almost-regular}
A function $\cF = \{f_\secp : \{0,1\}^{m(\secp)} \to \{0,1\}^{n(\secp)}\}$ is almost-regular if there exists efficiently computable polynomials $r(\secp)$ and $p(\secp)$ such that for all $\secp \in \bbN$ and $x \in \{0,1\}^{m(\secp)}$, \[\frac{1}{p(\secp)} \cdot 2^{r(\secp)} \leq \big| \{x' \in \{0,1\}^{n(\secp)} : f_\secp(x') = f_\secp(x)\}\big| \leq p(\secp) \cdot 2^{r(\secp)}.\]
\end{definition}

Note that we assume $r(\secp)$ is efficiently computable, which means that the regularity of $\cF$ is \emph{known}. This is often contrasted with the more general class of functions that are \emph{unknown} regular. Throughout this work, we always means \emph{known} regular.

\begin{definition}[Balanced Function]\label{def:balanced}
A function $\cF = \{f_\secp : \{0,1\}^{m(\secp)} \to \{0,1\}^{n(\secp)}\}_{\secp \in \bbN}$ is $\delta$-\emph{balanced} for some constant $\delta \in [0,1)$ if there exists a family of sets $\{\mathsf{BAD}_\secp \subset \{0,1\}^{n(\secp)}\}_{\secp \in \bbN}$ such that
\begin{enumerate}
    \item $|\mathsf{BAD}_\secp|/2^{n(\secp)} = \negl(\secp)$.
    \item $\Pr_{x \gets \{0,1\}^{m(\secp)}}[f_\secp(x) \in \mathsf{BAD}_\secp] = \negl(\secp)$.
    \item For every $z \notin \mathsf{BAD}_\secp$, $\Pr_{x \gets \{0,1\}^{m(\secp)}}[f_\secp(x) = z] \cdot 2^{n(\secp)} \in [1-\delta,1+\delta]$.
\end{enumerate}
\end{definition}

\begin{definition}[One-Way Function]\label{def:OWF}
A function $\cF = \{f_\secp : \{0,1\}^{m(\secp)} \to \{0,1\}^{\ell(\secp)}\}$ is one-way if for any QPT adversary $\cA = \{\cA_\secp\}_{\secp \in \bbN}$,

\[\Pr\left[f_\secp(x')=f(x) : \begin{array}{r}x \gets \{0,1\}^{m(\secp)} \\ x' \gets \cA_\secp(f(x))\end{array}\right] = \negl(\secp).\]

We say that $\cF$ is one-way \emph{over its range} if for any QPT adversary $\cA = \{\cA_\secp\}_{\secp \in \bbN}$,
\[\Pr\left[f_\secp(x) = y : \begin{array}{r}y \gets \{0,1\}^{\ell(\secp)} \\ x \gets \cA_\secp(y)\end{array}\right] = \negl(\secp).\]
\end{definition}

\begin{definition}[Universal Hash]
A hash function family $\cH = \{H_\secp : \{0,1\}^{m(\secp)} \to \{0,1\}^{n(\secp)}\}_{\secp \in \bbN}$ is called $t(\secp)$-universal if for each distinct $x_1,\dots,x_{t(\secp)} \in \{0,1\}^{m(\secp)}$ and $y_1,\dots,y_{t(\secp)} \in \{0,1\}^{n(\secp)}$, it holds that 
\[\Pr_{h \gets H_\secp}\left[h(x_1) = y_1 \wedge \dots \wedge h(x_{t(\secp)}) = y_{t(\secp)}\right] = 2^{-n(\secp) \cdot t(\secp)}.\]
\end{definition}

\begin{importedtheorem}[\cite{balancedOWF}]\label{impthm:balancedOWF}
Let $\cF = \{f_\secp : \{0,1\}^{m(\secp)} \to \{0,1\}^{\ell(\secp)}\}$ be an almost-regular one-way function. Then there exists $n(\secp) < m(\secp)$ and $\delta \in [0,1)$ such that for any $3\secp$-universal hash family $\cH = \{H_\secp : \{0,1\}^{\ell(\secp)} \to \{0,1\}^{n(\secp)}\}_{\secp \in \bbN}$ where each $h \in H_\secp$ can be described by $s(\secp)$ bits, the function \[\cF' = \left\{f'_\secp : \{0,1\}^{s(\secp) + m(\secp)} \to \{0,1\}^{s(\secp) + n(\secp)}\right\}_{\secp \in \bbN}, ~~ \text{where} ~~ f'_\secp(h,x) \coloneqq (h,h(f_\secp(x))),\] is $\delta$-balanced and one-way over its range.
\end{importedtheorem}

Now consider any balanced function $\cF = \{f_\secp : \{0,1\}^{m(\secp)} \to \{0,1\}^{n(\secp)}\}_{\secp \in \bbN}$ that is one-way over its range, and define the family of hash functions \[\cH^\cF = \left\{H_\secp : \{0,1\}^{m(\secp)} \to \{0,1\}^{n(\secp)}\right\}_{\secp \in \bbN}\] as follows. For each $\Delta \in \{0,1\}^{n(\secp)}$, define $f_\Delta : \{0,1\}^{n(\secp)} \to \{0,1\}^{n(\secp)}$ to, on input $z$, output the lexicographically first element of $\{z,z \oplus \Delta\}$.\footnote{We don't need to worry about the case when $\Delta = 0^{n(\secp)}$ since we'll be sampling $\Delta$ uniformly, but one could define $f_\Delta$ to be the identity in that case.} Then we define

\[H_\secp \coloneqq \left\{h_{\secp,\Delta} \coloneqq f_\Delta \circ f_\secp\right\}_{\Delta \in \{0,1\}^{n(\secp)}}.\]

We will also define the family of measurement functions  \[\cM = \left\{\left\{M[h_{\secp,\Delta}]\right\}_{h_{\secp,\Delta} \in H_\secp}\right\}_{\secp \in \bbN}\] as follows. The predicate $M[h_{\secp,\Delta}] : \{0,1\}^m \to \{0,1\}$ takes $x$ as input, computes $z \coloneqq f_\secp(x)$, and outputs $0$ if $z < z \oplus \Delta$ and $1$ if $z > z \oplus \Delta$ (where ordering is lexicographical).

\begin{theorem}\label{thm:target-CR}
Let $\delta \in [0,1)$ be a constant and $\cF = \{f_\secp : \{0,1\}^{m(\secp)} \to \{0,1\}^{n(\secp)}\}_{\secp \in \bbN}$ be a $\delta$-balanced function that is one-way over its range. Let $\cH^\cF$ and $\cM$ be as defined above. Then, $\cH^\cF$ is a balanced binary-measurement TCR hash with associated measurement function $\cM$.
\end{theorem}

\cref{impthm:balancedOWF} and \cref{thm:commitment} immediately give the following corollary.

\begin{corollary}
Assuming almost-regular one-way functions, there exists a quantum canonical bit commitment with $\PVD$.
\end{corollary}

%\cref{thm:TC-from-TCR} and \cref{thm:CETC-generalization} immediately give the following corollary.

%\begin{corollary}\label{corollary:target-collapsing}
%$\cH^\cF$ is $(\cU,\cM)$-target-collapsing and certified everlasting $(\cU,\cM)$-target-collapsing.
%\end{corollary}

\begin{proof}(Of \cref{thm:target-CR})
First, we check property (2) of \cref{def:BBMhash}. By properties (2) and (3) of \cref{def:balanced}, it holds that with $1-\negl(\secp)$ probability over the sampling of $h \gets H_\secp$ and $x \gets \{0,1\}^m$, \[\bigg|\frac{|\{x' \in h^{-1}(h(x))\} : M[h] = 0| - |\{x' \in h^{-1}(h(x))\} : M[h] = 1|}{|\{x' \in h^{-1}(h(x))\} : M[h] = 0| + |\{x' \in h^{-1}(h(x))\} : M[h] = 1|}\bigg| \leq \delta.\]

Next, we check property (1). Throughout this proof, we will drop indexing by $\secp$ for convenience. Suppose there exists a QPT adversary $\cA$ that breaks the $\cM$-target-collision-resistance of $\cH$. That is, the following experiment outputs 1 with $\nonnegl(\secp)$ probability.\\

\noindent\underline{$\Exp_{\mathsf{TCR}}$}
\begin{itemize}
    \item The challenger samples $\Delta \gets \{0,1\}^n$ and prepares the state $1/\sqrt{2^m}\sum_{x \in \{0,1\}^m}\ket{x}$ on register $X$. It applies $h_\Delta$ on $X$ to a fresh register $Y$ and measures $y \in \{0,1\}^n$, and then measures $P[h_\Delta]$ on $X$ to obtain a bit $b$ and left-over state on register $X$. The challenger sends $(\Delta,y,b)$ and register $X$ to $\cA$.
    \item $\cA$ outputs a string $x' \in \{0,1\}^n$.
    \item Output 1 if $h_\Delta(x') = y$ and $M[h_\Delta](x') = 1-b$.
\end{itemize}

We now define an adversary $\cA'$ that breaks the one-wayness of $\cF$ over its range.\\

\noindent\underline{$\Exp_{\mathsf{OW}}$}

\begin{itemize}
    \item The challenger samples $z \gets \{0,1\}^n$ and sends $z$ to $\cA'$.
    \item $\cA'$ prepares the state $1/\sqrt{2^m}\sum_{x \in \{0,1\}^m}\ket{x}$ on register $X$, applies $f$ on $X$ to a fresh register $Z$, and measures $z' \in \{0,1\}^n$. If $z' = z$, then measure register $X$ to obtain $x'$, and return $x'$. Otherwise, set $\Delta \coloneqq z \oplus z'$, set $b = 0$ if $z' < z$ and $b = 1$ otherwise, and set $y = f_\Delta(z)$. Then, initialize $\cA$ with $(\Delta,y,b)$ and register $X$. Run $\cA$ and forward its output $x'$ to the challenger.
    \item Output 1 if $f(x') = z$.
\end{itemize}

It suffices to show that $\cA$'s input comes from the same distribution over $(X,\Delta,y,b)$ in both experiments. To see this, we describe an alternative but identical way to sample $(X,\Delta,y,b)$ in the experiment $\Exp_{\mathsf{TCR}}$. Recalling that $h_\Delta = f_\Delta \circ f$, the challenger could (1) apply $f$ on $X$ to a fresh register $Z$, (2) sample $\Delta \gets \{0,1\}^n$, (3) apply $f_\Delta$ on $Z$ to a fresh register $Y$, and (4) measure $Y$ to obtain $y$ and measure $M[h_\Delta]$ on $X$ to obtain $b$. Note that step (4) is equivalent to instead just measuring the $Z$ register to obtain $z$, defining $b = 0$ if $z < z \oplus \Delta$ and $b=1$ if $z > z \oplus \Delta$, and defining $y = f_\Delta(z)$. Thus, we can imagine first applying $f$ on $X$ to a fresh register $Z$, measuring $Z$ to obtain $z$, sampling $\Delta \gets \{0,1\}^n$, and defining $y = f_\Delta(z)$. Defining $z' = z \oplus \Delta$ and using the fact that $\Delta$ was sampled uniformly at random, we see that this is exactly the same distribution that is sampled in $\Exp_{\mathsf{OW}}$, except that $\cA$ is not initialized if $\Delta = 0^{n(\secp)}$ (in which case $\cA'$ wins the experiment anyway). 

\end{proof}

Now, we generalize the notion of almost-regularity (\cref{def:almost-regular}), balanced (\cref{def:balanced}), and one-wayness (\cref{def:OWF}) to function \emph{families}, where there is a set of of $f \in F_\secp$ associated with each security parameter. All previous definitions generalize to this setting with the requirement that they hold with $1-\negl(\secp)$ probability over $f \gets F_\secp$, and all previous claims follow. We consider families of functions with \emph{trapdoors} that allow us to invert the function and obtain public-key encryption along with other cryptographic primitives.

\begin{definition}[Superposition-invertible trapdoor function]
We say that a function family $\cF = \{F_\secp\}_{\secp \in \bbN}$ is a superposition-invertible trapdoor function if there exist algorithms $\Samp,\Invert$ with the following properties.
\begin{itemize}
    \item $\Samp(1^\secp)$: The sampling algorithm samples a uniformly random function $f \in F_\secp$ along with a trapdoor $\td$.
    \item $\Invert(\td,y)$: Given the trapdoor $\td$ and an image $y$, $\Invert$ outputs a state within negligible trace distance of \[\frac{1}{\sqrt{|f^{-1}(y)|}}\sum_{x:f(x)=y}\ket{x}.\]
\end{itemize}
\end{definition}

\begin{remark}
For the case of injective function families $\cF$, the notion of superposition-invertible trapdoor is equivalent to the standard notion of trapdoor, since there is only one preimage per image.
\end{remark}

\begin{claim}\label{claim:phase-recoverability}
Assuming injective trapdoor one-way functions (or more generally, superposition-invertible trapdoor almost-regular one-way functions), there exists a balanced binary-measurement TCR hash with trapdoor phase-recoverability.
\end{claim}

By \cref{thm:PKE} and \cref{cor:compiler}, we obtain the following corollary.

\begin{corollary}
Assuming the existence of injective trapdoor one-way functions (or more generally, superposition-invertible trapdoor almost-regular one-way functions), there exists PKE with $\PVD$. Additionally assuming post-quantum \[X \in \left\{\begin{array}{r}\text{quantum fully-homormophic encryption, attribute-based encryption}, \\ \text{witness encryption, timed-release encryption}\end{array}\right\},\] there exists $X$ with $\PVD$.
\end{corollary}

\begin{proof}(Of \cref{claim:phase-recoverability}) Given a superposition-invertible almost-regular one-way function, then we know from \cref{impthm:balancedOWF} that we can compose it with a $3\secp$-universal hash function to obtain a $\delta$-balanced function $\cF$ that is one-way over its range, and \cref{thm:target-CR} tells us that we can then obtain a balanced binary-measurement TCR hash $\cH^\cF = \{H_\secp\}_{\secp \in \bbN}$. It remains to check that the resulting hash has trapdoor phase-recoverability.

To see this, we observe that for any polynomials $m(\secp),n(\secp),t(\secp)$, there exists a superposition-invertible $t(\secp)$-universal hash function family $\{U_\secp : \{0,1\}^{m(\secp)} \to \{0,1\}^{n(\secp)}\}_{\secp \in \bbN}$ (without the need for a trapdoor). For example, we can use the Chor-Goldreich construction \cite{Chor-Goldreich}, where each hash in the family is defined by coefficients of a degree-$(t(\secp)-1)$ univariate polynomial over a finite field, and evaluation is polynomial evaluation. To invert, use a root-finding algorithm (e.g. \cite{Cantor1981ANA}) to recover the (at most polynomial) roots, and then arrange these in superposition. Note that for a compressing universal hash from $\{0,1\}^m \to \{0,1\}^n$, one would use a finite field of size at least $2^m$ and define the hash output to consist of (say) the first $n$ bits of the description of the finite field element that results from polynomial evaluation. In this case, the quantum inverter would first prepare a uniform superposition over all of the remaining $m-n$ bits of the field element, and run the above procedure in superposition.

Thus, given $h \in H_\secp$, where $h = f_\Delta \circ f$ for $\Delta \neq 0^n$, along with a trapdoor $\td$ for $f$, we can efficiently prepare the state \[\ket{\psi_{h,y,0}} = \frac{1}{\sqrt{|h^{-1}(y)|}}\sum_{x:h(x)=y}\ket{x}.\]

Then, the procedure $\Recover(\td,y,X)$ would measure register $X$ in the $\{\dyad{\psi_{h,y,0}}{\psi_{h,y,0}}, \bbI - \dyad{\psi_{h,y,0}}{\psi_{h,y,0}}\}$ basis, and output 0 if the first outcome is observed. We have that with probability $1-\negl(\secp)$ over the sampling of $h$,

\begin{align*}
    &\Pr_{x \gets \{0,1\}^m}\left[\Recover(\td,h(x),\ket{h,h(x),0}) \to 0\right] = 1,\\
    &\Pr_{x \gets \{0,1\}^m}\left[\Recover(\td,h(x),\ket{h,h(x),1}) \to 0\right] \leq (1-\delta)^2,
\end{align*}

by the proof of binding in \cref{thm:commitment}. This completes the proof.
%Given any superposition-invertible trapdoor almost-regular one-way function, by \cref{remark:invertible-hash} and \cref{impthm:balancedOWF} we can construct a superposition-invertible $\delta$-balanced function family $\cF$ that is one-way over its range. In turn, we can define superposition-invertible function family $\cH^\cF = \{H_\secp\}_{\secp \in \bbN}$ with associated algorithms $(\Samp,\Invert)$, and define predicate family $\cM = \{\{M[h]\}_{h \in H_\secp}\}_{\secp \in \bbN}$ as above.

\end{proof}

\subsection{Balanced Binary-Measurement TCR from Pseudorandom Group Actions}
\label{sec:hmy}
Finally, we show that the recent public-key encryption scheme of \cite{HMY} based on pseudorandom group actions has publicly-verifiable deletion, which follows fairly immediately from our framework. First, we need some preliminaries from \cite{JQSY,HMY}.

%\dakshita{Here, can we instead just rely on the fact that there is a trapdoor that helps distinguish the superposition from the mixture? and relax the trapdoor property from Section 7.3 to only ask for distinguishing given a trapdoor (we can say that it is implied by superposition-invertibility..). This way we will just prove that the group action based scheme is an instance of the appropriate trapdoored collapsing function.}

\begin{definition}[Group Action] 
Let $G$ be a (not necessarily abelian) group, $S$ be a set, and $\star: G \times S \to S$ be a function where we write $g \star s$ to mean $\star(g,s)$. We say that $(G,S,\star)$ is a group action if it satisfies the following:
\begin{itemize}
    \item For the identity element $e \in G$ and any $s \in S$, we have $e \star s = s$.
    \item For any $g,h \in G$ and any $s \in S$, we have $(gh) \star s = g \star (h \star s)$.
\end{itemize}
\end{definition}

\cite{JQSY,HMY} also require a number of efficiency properties from the group action, and we refer the reader to their papers for these specifications.

\begin{definition}[Pseudorandom Group Action]
A group action $(G,S,\star)$ is \emph{pseudorandom} if it satisfies the following:
\begin{itemize}
    \item We have that \[\Pr_{s,t \gets S}[\exists g \in G \text{ s.t. } g \star s = t] = \negl(\secp).\]
    \item For any QPT adversary $\{\cA_\secp\}_{\secp \in \bbN}$,
    \[\big| \Pr_{s \gets S, g \gets S}[\cA_\secp(s,g \star s) = 1] - \Pr_{s,t \gets S}[\cA_\secp(s,t) = 1]\big| = \negl(\secp).\]
\end{itemize}
\end{definition}

Given a pseudorandom group action $(G,S,\star)$, \cite{HMY} consider the following hash family $\cH^{(G,S,\star)} = \{H_h\}_{h \in S_G}$, where $S_G = \{(s_0,s_1) \in S^2: \exists g \in G \text{ s.t. } s_1 = g \star s_0\}$.
\begin{itemize}
    \item The algorithm $\Samp(1^\secp)$ samples $s_0 \gets S,g \gets G$ and outputs $h = (s_0,s_1)$ as the description of the hash and $\td = g$ as the trapdoor.
    \item For an input $(b,k)$ where $b \in \{0,1\}$ and $k \in G$, define $h(b,k) \coloneqq k \star s_b$.
\end{itemize}

\begin{claim}
$\cH^{(G,S,\star)}$ is a balanced binary-measurement TCR hash with trapdoor phase-recoverability.
\end{claim}

\begin{proof}
Define predicate family $\cM$ as $M[h](b,k) = b$. That is, it does not depend on $h$, and simply outputs the first bit of its input. Then, this claim actually follows immediately from what is already proven in \cite{HMY}. First, \cite[Theorem 4.10]{HMY} shows that given $\td$ and $y \in S$, it is possible to perfectly distinguish \[\frac{1}{\sqrt{2}}\ket{0,h_0^{-1}(y)} + \frac{1}{\sqrt{2}}\ket{1,h_1^{-1}(y)} ~~ \text{and} ~~ \frac{1}{\sqrt{2}}\ket{0,h_0^{-1}(y)} - \frac{1}{\sqrt{2}}\ket{1,h_1^{-1}(y)},\] where $h_b \coloneqq h(b,\cdot)$, which establishes trapdoor phase-recoverability. Next, \cite[Theorem 4.19]{HMY} shows that $\cH^{(G,S,\star)}$ satisfies \emph{conversion hardness}, which is equivalent to our notion of $\cM$-target-collision-resistance. 
\end{proof}

By \cref{thm:PKE} and \cref{cor:compiler}, we obtain the following corollary.

\begin{corollary}
Assuming the existence pseudorandom group actions, there exists PKE with $\PVD$. Additionally assuming post-quantum \[X \in \left\{\begin{array}{r}\text{quantum fully-homormophic encryption, attribute-based encryption}, \\ \text{witness encryption, timed-release encryption}\end{array}\right\},\] there exists $X$ with $\PVD$.
\end{corollary}

\printbibliography

@INPROCEEDINGS{892139,  author={Hales, L. and Hallgren, S.},  booktitle={Proceedings 41st Annual Symposium on Foundations of Computer Science},   title={An improved quantum Fourier transform algorithm and applications},   year={2000},  volume={},  number={},  pages={515-525},  doi={10.1109/SFCS.2000.892139}}

@article{Grover2002CreatingST,
  title={Creating superpositions that correspond to efficiently integrable probability distributions},
  author={Lov K. Grover and Terry Rudolph},
  journal={arXiv: Quantum Physics},
  year={2002}
}

@InProceedings{cryptoeprint:2009/285,
author="Stehl{\'e}, Damien
and Steinfeld, Ron
and Tanaka, Keisuke
and Xagawa, Keita",
editor="Matsui, Mitsuru",
title="Efficient Public Key Encryption Based on Ideal Lattices",
booktitle="Advances in Cryptology -- ASIACRYPT 2009",
year="2009",
publisher="Springer Berlin Heidelberg",
address="Berlin, Heidelberg",
pages="617--635",
isbn="978-3-642-10366-7"
}

@article{DBLP:journals/siamcomp/MicciancioR07,
  author    = {Daniele Micciancio and
               Oded Regev},
  title     = {Worst-Case to Average-Case Reductions Based on Gaussian Measures},
  journal   = {{SIAM} J. Comput.},
  volume    = {37},
  number    = {1},
  pages     = {267--302},
  year      = {2007},
  url       = {https://doi.org/10.1137/S0097539705447360},
  doi       = {10.1137/S0097539705447360},
  bibsource = {dblp computer science bibliography, https://dblp.org}
}

@inproceedings{DBLP:conf/stoc/Ajtai96,
  author    = {Mikl{\'{o}}s Ajtai},
  editor    = {Gary L. Miller},
  title     = {Generating Hard Instances of Lattice Problems (Extended Abstract)},
  booktitle = {Proceedings of the Twenty-Eighth Annual {ACM} Symposium on the Theory
               of Computing, Philadelphia, Pennsylvania, USA, May 22-24, 1996},
  pages     = {99--108},
  publisher = {{ACM}},
  year      = {1996},
  url       = {https://doi.org/10.1145/237814.237838},
  doi       = {10.1145/237814.237838},
  bibsource = {dblp computer science bibliography, https://dblp.org}
}

@book{Wilde13,
author = {Wilde, Mark M.},
title = {Quantum Information Theory},
year = {2013},
isbn = {1107034256},
publisher = {Cambridge University Press},
address = {USA},
edition = {1st}
}

@book{NielsenChuang11,
author = {Nielsen, Michael A. and Chuang, Isaac L.},
title = {Quantum Computation and Quantum Information: 10th Anniversary Edition},
year = {2011},
isbn = {1107002176},
publisher = {Cambridge University Press},
address = {USA},
edition = {10th}
}

@article{brakerski2021cryptographic,
author = {Brakerski, Zvika and Christiano, Paul and Mahadev, Urmila and Vazirani, Umesh and Vidick, Thomas},
title = {A Cryptographic Test of Quantumness and Certifiable Randomness from a Single Quantum Device},
year = {2021},
issue_date = {October 2021},
publisher = {Association for Computing Machinery},
address = {New York, NY, USA},
volume = {68},
number = {5},
issn = {0004-5411},
url = {https://doi.org/10.1145/3441309},
doi = {10.1145/3441309},
journal = {J. ACM},
month = {aug},
articleno = {31},
numpages = {47}
}

@inproceedings{hiroka2021quantum,
  author    = {Taiga Hiroka and
               Tomoyuki Morimae and
               Ryo Nishimaki and
               Takashi Yamakawa},
  editor    = {Mehdi Tibouchi and
               Huaxiong Wang},
  title     = {Quantum Encryption with Certified Deletion, Revisited: Public Key,
               Attribute-Based, and Classical Communication},
  booktitle = {Advances in Cryptology - {ASIACRYPT} 2021 - 27th International Conference
               on the Theory and Application of Cryptology and Information Security,
               Singapore, December 6-10, 2021, Proceedings, Part {I}},
  series    = {Lecture Notes in Computer Science},
  volume    = {13090},
  pages     = {606--636},
  publisher = {Springer},
  year      = {2021},
  url       = {https://doi.org/10.1007/978-3-030-92062-3\_21},
  doi       = {10.1007/978-3-030-92062-3\_21},
  bibsource = {dblp computer science bibliography, https://dblp.org}
}

@InProceedings{Brakerski18,
author="Brakerski, Zvika",
editor="Shacham, Hovav
and Boldyreva, Alexandra",
title="Quantum FHE (Almost) As Secure AsÂ Classical",
booktitle="Advances in Cryptology -- CRYPTO 2018",
year="2018",
publisher="Springer International Publishing",
address="Cham",
pages="67--95",
isbn="978-3-319-96878-0"
}

@INPROCEEDINGS{mahadev2018classical,
  author={Mahadev, Urmila},
  booktitle={2018 IEEE 59th Annual Symposium on Foundations of Computer Science (FOCS)}, 
  title={Classical Verification of Quantum Computations}, 
  year={2018},
  volume={},
  number={},
  pages={259-267},
  doi={10.1109/FOCS.2018.00033}}

@article{Banaszczyk1993,
author = {Banaszczyk, W.},
journal = {Mathematische Annalen},
number = {4},
pages = {625-636},
title = {New bounds in some transference theorems in the geometry of numbers.},
url = {http://eudml.org/doc/165105},
volume = {296},
year = {1993},
}

@inproceedings{hiroka2021certified,
  author    = {Taiga Hiroka and
               Tomoyuki Morimae and
               Ryo Nishimaki and
               Takashi Yamakawa},
  editor    = {Yevgeniy Dodis and
               Thomas Shrimpton},
  title     = {Certified Everlasting Zero-Knowledge Proof for {QMA}},
  booktitle = {Advances in Cryptology - {CRYPTO} 2022 - 42nd Annual International
               Cryptology Conference, {CRYPTO} 2022, Santa Barbara, CA, USA, August
               15-18, 2022, Proceedings, Part {I}},
  series    = {Lecture Notes in Computer Science},
  volume    = {13507},
  pages     = {239--268},
  publisher = {Springer},
  year      = {2022},
  url       = {https://doi.org/10.1007/978-3-031-15802-5\_9},
  doi       = {10.1007/978-3-031-15802-5\_9},
  bibsource = {dblp computer science bibliography, https://dblp.org}
}

@InProceedings{10.1007/978-3-540-70936-7_3,
author="M{\"u}ller-Quade, J{\"o}rn
and Unruh, Dominique",
editor="Vadhan, Salil P.",
title="Long-Term Security and Universal Composability",
booktitle="Theory of Cryptography",
year="2007",
publisher="Springer Berlin Heidelberg",
address="Berlin, Heidelberg",
pages="41--60",
isbn="978-3-540-70936-7"
}

@article{Unruh2013,
author = {Unruh, Dominique},
title = {Revocable Quantum Timed-Release Encryption},
year = {2015},
issue_date = {December 2015},
publisher = {Association for Computing Machinery},
address = {New York, NY, USA},
volume = {62},
number = {6},
issn = {0004-5411},
url = {https://doi.org/10.1145/2817206},
doi = {10.1145/2817206},
journal = {J. ACM},
month = {dec},
articleno = {49},
numpages = {76}
}

@inproceedings{10.5555/1756169.1756191,
author = {Jarecki, Stanis\l{}aw and Lysyanskaya, Anna},
title = {Adaptively Secure Threshold Cryptography: Introducing Concurrency, Removing Erasures},
year = {2000},
isbn = {3540675175},
publisher = {Springer-Verlag},
address = {Berlin, Heidelberg},
booktitle = {Proceedings of the 19th International Conference on Theory and Application of Cryptographic Techniques},
pages = {221–242},
numpages = {22},
location = {Bruges, Belgium},
series = {EUROCRYPT'00}
}

@inproceedings{10.1145/237814.238015,
author = {Canetti, Ran and Feige, Uri and Goldreich, Oded and Naor, Moni},
title = {Adaptively Secure Multi-Party Computation},
year = {1996},
isbn = {0897917855},
publisher = {Association for Computing Machinery},
address = {New York, NY, USA},
url = {https://doi.org/10.1145/237814.238015},
doi = {10.1145/237814.238015},
booktitle = {Proceedings of the Twenty-Eighth Annual ACM Symposium on Theory of Computing},
pages = {639–648},
numpages = {10},
location = {Philadelphia, Pennsylvania, USA},
series = {STOC '96}
}

@InProceedings{ananth2020secure,
author="Ananth, Prabhanjan
and La Placa, Rolando L.",
editor="Canteaut, Anne
and Standaert, Fran{\c{c}}ois-Xavier",
title="Secure Software Leasing",
booktitle="Advances in Cryptology -- EUROCRYPT 2021",
year="2021",
publisher="Springer International Publishing",
address="Cham",
pages="501--530",
isbn="978-3-030-77886-6"
}

@inproceedings{Aar09,
  title={Quantum copy-protection and quantum money},
  author={Aaronson, Scott},
  booktitle={2009 24th Annual IEEE Conference on Computational Complexity},
  pages={229--242},
  year={2009},
  organization={IEEE}
}

@inproceedings{GSW2013,
  author    = {Craig Gentry and
               Amit Sahai and
               Brent Waters},
  editor    = {Ran Canetti and
               Juan A. Garay},
  title     = {Homomorphic Encryption from Learning with Errors: Conceptually-Simpler,
               Asymptotically-Faster, Attribute-Based},
  booktitle = {Advances in Cryptology - {CRYPTO} 2013 - 33rd Annual Cryptology Conference,
               Santa Barbara, CA, USA, August 18-22, 2013. Proceedings, Part {I}},
  series    = {Lecture Notes in Computer Science},
  volume    = {8042},
  pages     = {75--92},
  publisher = {Springer},
  year      = {2013},
  url       = {https://doi.org/10.1007/978-3-642-40041-4\_5},
  doi       = {10.1007/978-3-642-40041-4\_5},
  bibsource = {dblp computer science bibliography, https://dblp.org}
}

@inproceedings{BB84,
  address = {India},
  author = {Bennett, C. H. and Brassard, G.},
  booktitle = {Proceedings of IEEE International Conference on Computers, Systems, and Signal Processing},
  location = {Bangalore},
  pages = 175,
  title = {{Quantum cryptography: Public key distribution and coin tossing}},
  year = 1984
}

@article{Wiesner83,
author = {Wiesner, Stephen},
title = {Conjugate Coding},
year = {1983},
issue_date = {Winter-Spring 1983},
publisher = {Association for Computing Machinery},
address = {New York, NY, USA},
volume = {15},
number = {1},
issn = {0163-5700},
url = {https://doi.org/10.1145/1008908.1008920},
doi = {10.1145/1008908.1008920},
journal = {SIGACT News},
month = jan,
pages = {78–88},
numpages = {11}
}

@inproceedings{cryptoeprint:2007:432,
author = {Gentry, Craig and Peikert, Chris and Vaikuntanathan, Vinod},
title = {Trapdoors for Hard Lattices and New Cryptographic Constructions},
year = {2008},
isbn = {9781605580470},
publisher = {Association for Computing Machinery},
address = {New York, NY, USA},
url = {https://doi.org/10.1145/1374376.1374407},
doi = {10.1145/1374376.1374407},
booktitle = {Proceedings of the Fortieth Annual ACM Symposium on Theory of Computing},
pages = {197–206},
numpages = {10},
location = {Victoria, British Columbia, Canada},
series = {STOC '08}
}

@article{Broadbent_2020,
   title={Quantum Encryption with Certified Deletion},
   ISBN={9783030643812},
   ISSN={1611-3349},
   url={http://dx.doi.org/10.1007/978-3-030-64381-2_4},
   DOI={10.1007/978-3-030-64381-2_4},
   journal={Lecture Notes in Computer Science},
   publisher={Springer International Publishing},
   author={Broadbent, Anne and Islam, Rabib},
   year={2020},
   pages={92–122}
}

@article{Tomamichel2017largelyself,
  doi = {10.22331/q-2017-07-14-14},
  url = {https://doi.org/10.22331/q-2017-07-14-14},
  title = {A largely self-contained and complete security proof for quantum key  distribution},
  author = {Tomamichel, Marco and Leverrier, Anthony},
  journal = {{Quantum}},
  issn = {2521-327X},
  publisher = {{Verein zur F{\"{o}}rderung des Open Access Publizierens in den Quantenwissenschaften}},
  volume = {1},
  pages = {14},
  month = jul,
  year = {2017}
}

@misc{cryptoeprint:2017/258,
      author = {Chris Peikert and Oded Regev and Noah Stephens-Davidowitz},
      title = {Pseudorandomness of Ring-LWE for Any Ring and Modulus},
      howpublished = {Cryptology ePrint Archive, Paper 2017/258},
      year = {2017},
      note = {\url{https://eprint.iacr.org/2017/258}},
      url = {https://eprint.iacr.org/2017/258}
}

@misc{cryptoeprint:2023/325,
      author = {Prabhanjan Ananth and Alexander Poremba and Vinod Vaikuntanathan},
      title = {Revocable Cryptography from Learning with Errors},
      howpublished = {Cryptology ePrint Archive, Paper 2023/325},
      year = {2023},
      note = {\url{https://eprint.iacr.org/2023/325}},
      url = {https://eprint.iacr.org/2023/325}
}

@article{Regev05,
  author = {Regev, Oded},
  title = {On lattices, learning with errors, random linear codes, and cryptography},
  journal = {Journal of the ACM},
  publisher = {ACM},
  address = {New York, NY, USA},
  volume = {56},
  number = {6},
  year = {2005},
  issn = {0004-5411},
  pages = {34:1--34:40},
  doi = {10.1145/1568318.1568324}
}

@misc{BBK22,
  doi = {10.48550/ARXIV.2203.02314},
  
  url = {https://arxiv.org/abs/2203.02314},
  
  author = {Bitansky, Nir and Brakerski, Zvika and Kalai, Yael Tauman},
  
  title = {Constructive Post-Quantum Reductions},
  
  publisher = {arXiv},
  
  year = {2022},
  
  copyright = {arXiv.org perpetual, non-exclusive license}
}

@inproceedings{Poremba22,
  author    = {Alexander Poremba},
  editor    = {Yael Tauman Kalai},
  title     = {Quantum Proofs of Deletion for Learning with Errors},
  booktitle = {14th Innovations in Theoretical Computer Science Conference, {ITCS}
               2023, January 10-13, 2023, MIT, Cambridge, Massachusetts, {USA}},
  series    = {LIPIcs},
  volume    = {251},
  pages     = {90:1--90:14},
  publisher = {Schloss Dagstuhl - Leibniz-Zentrum f{\"{u}}r Informatik},
  year      = {2023},
  url       = {https://doi.org/10.4230/LIPIcs.ITCS.2023.90},
  doi       = {10.4230/LIPIcs.ITCS.2023.90},
  bibsource = {dblp computer science bibliography, https://dblp.org}
}

@misc{cryptoeprint:2022/969,
      author = {Taiga Hiroka and Tomoyuki Morimae and Ryo Nishimaki and Takashi Yamakawa},
      title = {Certified Everlasting Functional Encryption},
      howpublished = {Cryptology ePrint Archive, Paper 2022/969},
      year = {2022},
      note = {\url{https://eprint.iacr.org/2022/969}},
      url = {https://eprint.iacr.org/2022/969}
}

@misc{cryptoeprint:2022/786,
      author = {Marcel Dall'Agnol and Nicholas Spooner},
      title = {On the necessity of collapsing},
      howpublished = {Cryptology ePrint Archive, Paper 2022/786},
      year = {2022},
      note = {\url{https://eprint.iacr.org/2022/786}},
      url = {https://eprint.iacr.org/2022/786}
}

@article{balancedOWF,
	Author = {Haitner, Iftach and Horvitz, Omer and Katz, Jonathan and Koo, Chiu-Yuen and Morselli, Ruggero and Shaltiel, Ronen},
	Da = {2009/07/01},
	Doi = {10.1007/s00145-007-9012-8},
	Id = {Haitner2009},
	Isbn = {1432-1378},
	Journal = {Journal of Cryptology},
	Number = {3},
	Pages = {283--310},
	Title = {Reducing Complexity Assumptions for Statistically-Hiding Commitment},
	Ty = {JOUR},
	Url = {https://doi.org/10.1007/s00145-007-9012-8},
	Volume = {22},
	Year = {2009}}

@article{Chor-Goldreich,
author = {Chor, B. and Goldreich, O.},
title = {On the Power of Two-Point Based Sampling},
year = {1989},
issue_date = {March 1989},
publisher = {Academic Press, Inc.},
address = {USA},
volume = {5},
number = {1},
issn = {0885-064X},
url = {https://doi.org/10.1016/0885-064X(89)90015-0},
doi = {10.1016/0885-064X(89)90015-0},
journal = {J. Complex.},
month = {apr},
pages = {96–106},
numpages = {11}
}

@article{Cantor1981ANA,
  title={A new algorithm for factoring polynomials over finite fields},
  author={David Geoffrey Cantor and Hans Zassenhaus},
  journal={Mathematics of Computation},
  year={1981},
  volume={36},
  pages={587-592}
}

@InProceedings{Yan,
author="Yan, Jun",
editor="Agrawal, Shweta
and Lin, Dongdai",
title="General Properties of Quantum Bit Commitments (Extended Abstract)",
booktitle="Advances in Cryptology -- ASIACRYPT 2022",
year="2022",
publisher="Springer Nature Switzerland",
address="Cham",
pages="628--657",
isbn="978-3-031-22972-5"
}

@misc{cryptoeprint:2022/1178,
      author = {James Bartusek and Dakshita Khurana},
      title = {Cryptography with Certified Deletion},
      howpublished = {Cryptology ePrint Archive, Paper 2022/1178},
      year = {2022},
      note = {\url{https://eprint.iacr.org/2022/1178}},
      url = {https://eprint.iacr.org/2022/1178}
}

@inproceedings{HMY,
  author    = {Minki Hhan and Tomoyuki Morimae and Takashi Yamakawa},
  title     = {From the Hardness of Detecting Superpositions to Cryptography: Quantum Public Key Encryption and Commitments},
  booktitle = {{Eurocrypt} 2023 (to appear)},
  year      = {2023},
}

@misc{BGGKMRR,
      author = {James Bartusek and Sanjam Garg and Vipul Goyal and Dakshita Khurana and Giulio Malavolta and Justin Raizes and Bhaskar Roberts},
      title = {Obfuscation and Outsourced Computation with Certified Deletion},
      howpublished = {Cryptology ePrint Archive, Paper 2023/265},
      year = {2023},
      url = {https://eprint.iacr.org/2023/265}
}

@inproceedings{JQSY,
author = {Ji, Zhengfeng and Qiao, Youming and Song, Fang and Yun, Aaram},
title = {General Linear Group Action on Tensors: A Candidate for Post-Quantum Cryptography},
year = {2019},
isbn = {978-3-030-36029-0},
publisher = {Springer-Verlag},
address = {Berlin, Heidelberg},
url = {https://doi.org/10.1007/978-3-030-36030-6_11},
doi = {10.1007/978-3-030-36030-6_11},
booktitle = {Theory of Cryptography: 17th International Conference, TCC 2019, Nuremberg, Germany, December 1–5, 2019, Proceedings, Part I},
pages = {251–281},
numpages = {31},
location = {Nuremberg, Germany}
}

@inproceedings{crypto-2022-32202,
  author    = {Mark Zhandry},
  editor    = {Yevgeniy Dodis and
               Thomas Shrimpton},
  title     = {New Constructions of Collapsing Hashes},
  booktitle = {Advances in Cryptology - {CRYPTO} 2022 - 42nd Annual International
               Cryptology Conference, {CRYPTO} 2022, Santa Barbara, CA, USA, August
               15-18, 2022, Proceedings, Part {III}},
  series    = {Lecture Notes in Computer Science},
  volume    = {13509},
  pages     = {596--624},
  publisher = {Springer},
  year      = {2022},
  url       = {https://doi.org/10.1007/978-3-031-15982-4\_20},
  doi       = {10.1007/978-3-031-15982-4\_20},
  bibsource = {dblp computer science bibliography, https://dblp.org}
}

@inproceedings{crypto-2022-32124,
  title={The Gap Is Sensitive to Size of Preimages: Collapsing Property Doesn't Go Beyond Quantum Collision-Resistance for Preimages Bounded Hash Functions},
  publisher={Springer-Verlag},
  author={Shujiao Cao and Rui Xue},
  year=2022
}

@InProceedings{10.1007/978-3-662-49896-5_18,
author="Unruh, Dominique",
editor="Fischlin, Marc
and Coron, Jean-S{\'e}bastien",
title="Computationally Binding Quantum Commitments",
booktitle="Advances in Cryptology -- EUROCRYPT 2016",
year="2016",
publisher="Springer Berlin Heidelberg",
address="Berlin, Heidelberg",
pages="497--527",
isbn="978-3-662-49896-5"
}

@InProceedings{10.1007/978-3-662-53890-6_6,
author="Unruh, Dominique",
editor="Cheon, Jung Hee
and Takagi, Tsuyoshi",
title="Collapse-Binding Quantum Commitments Without Random Oracles",
booktitle="Advances in Cryptology -- ASIACRYPT 2016",
year="2016",
publisher="Springer Berlin Heidelberg",
address="Berlin, Heidelberg",
pages="166--195",
isbn="978-3-662-53890-6"
}

@InProceedings{10.1007/BFb0054137,
author="Simon, Daniel R.",
editor="Nyberg, Kaisa",
title="Finding collisions on a one-way street: Can secure hash functions be based on general assumptions?",
booktitle="Advances in Cryptology --- EUROCRYPT'98",
year="1998",
publisher="Springer Berlin Heidelberg",
address="Berlin, Heidelberg",
pages="334--345",
isbn="978-3-540-69795-4"
}

@InProceedings{10.1007/978-3-030-26951-7_12,
author="Liu, Qipeng
and Zhandry, Mark",
editor="Boldyreva, Alexandra
and Micciancio, Daniele",
title="Revisiting Post-quantum Fiat-Shamir",
booktitle="Advances in Cryptology -- CRYPTO 2019",
year="2019",
publisher="Springer International Publishing",
address="Cham",
pages="326--355",
isbn="978-3-030-26951-7"
}

@inproceedings{AKNYY,
  author    = {Shweta Agarwal and Fuyuki Kitagawa and Ryo Nishimaki and Shota Yamada and Takashi Yamakawa},
  title     = {Public Key Encryption with Secure Key Leasing},
  booktitle = {{Eurocrypt} 2023 (to appear)},
  year      = {2023},
}

\newpage 
\appendix

\end{document}